\documentclass[journal]{IEEEtran}

\usepackage{amsmath,amsthm,amssymb,amsfonts, mathrsfs, microtype}
\usepackage{enumerate}
\usepackage{bbm}
\usepackage{mathtools}
\usepackage[mathscr]{euscript}

\usepackage[colorlinks=true,linkcolor=blue]{hyperref}%

\newtheorem{theorem}{Theorem}
\newtheorem{corollary}{Corollary}
\newtheorem{lemma}{Lemma}

\newtheorem{definition}{Definition}

\theoremstyle{remark}
\newtheorem{remark}{Remark}

\usepackage{tikz-cd}
\usepackage{hyperref}
\usepackage[utf8]{inputenc}
\usepackage[square, numbers, comma, sort&compress]{natbib}
\usepackage{tcolorbox}
\usepackage{graphicx}
\usepackage{todonotes}
\usepackage{subcaption}

\newcommand{\rvec}[1]{\pmb{#1}}

\newcommand{\dis}[2]{\mathsf{d}\left(#1,#2\right)}

\newcommand{\rdf}{\mathbb{R}}

\newcommand{\prob}[1]{\mathbb{P}\left[#1\right]}

\newcommand{\myE}[1]{\mathbb{E}\left[#1\right]}

\newcommand{\var}[1]{\mathrm{Var}\left[#1\right]}

\newcommand{\TS}{\mathcal{T}}

\newcommand{\LRB}[1]{\left(#1\right)}

\newcommand{\lrbb}[1]{\left\{#1\right\}}

\newcommand{\dmax}{d_{\mathrm{max}}}

\newcommand{\lrabs}[1]{\left |#1\right |}

\newcommand{\bec}{C_{\textsf{BE}}}

\begin{document}
\title{The Dispersion of the Gauss-Markov Source}
\author{Peida Tian,~\IEEEmembership{Student Member,~IEEE,}
Victoria Kostina,~\IEEEmembership{Member,~IEEE}
\thanks{P. Tian and V. Kostina are with the Department of Electrical Engineering, California Institute of Technology (e-mail: \hbox{\{ptian, vkostina\}@caltech.edu}). This research was supported in part by the National Science Foundation (NSF) under Grant CCF-1566567 and CCF-1751356. A preliminary version~\cite{tian2018dispersion} of this paper was published in the IEEE International Symposium on Information Theory, Vail, CO, USA, June 2018.}
\thanks{Copyright \textsuperscript{\textcopyright}2017 IEEE. Personal use of this material is permitted.  However, permission to use this material for any other purposes must be obtained from the IEEE by sending a request to pubs-permissions@ieee.org.}
}
\maketitle
\begin{abstract}
The Gauss-Markov source produces  $U_i = aU_{i-1} + Z_i$ for $i\geq 1$, where  $U_0 = 0$, $|a|<1$ and $Z_i\sim\mathcal{N}(0, \sigma^2)$ are i.i.d. Gaussian random variables. We consider lossy compression of a block of $n$ samples of the Gauss-Markov source under squared error distortion. We obtain the Gaussian approximation for the Gauss-Markov source with excess-distortion criterion for any distortion $d>0$, and we show that the dispersion has a reverse waterfilling representation. This is the \emph{first} finite blocklength result for lossy compression of \emph{sources with memory}. We prove that the finite blocklength rate-distortion function $R(n,d,\epsilon)$ approaches the rate-distortion function $\mathbb{R}(d)$ as $R(n,d,\epsilon) = \mathbb{R}(d) + \sqrt{\frac{V(d)}{n}}Q^{-1}(\epsilon) + o\left(\frac{1}{\sqrt{n}}\right)$, where $V(d)$ is the dispersion, $\epsilon \in (0,1)$ is the excess-distortion probability, and $Q^{-1}$ is the inverse $Q$-function. We give a reverse waterfilling integral representation for the dispersion $V(d)$, which parallels that of the rate-distortion functions for Gaussian processes. Remarkably, for all $0 < d\leq  \frac{\sigma^2}{(1+|a|)^2}$, $R(n,d,\epsilon)$ of the Gauss-Markov source coincides with that of $Z_i$, the i.i.d. Gaussian noise driving the process, up to the second-order term. Among novel technical tools developed in this paper is a sharp approximation of the eigenvalues of the covariance matrix of $n$ samples of the Gauss-Markov source, and a construction of a typical set using the maximum likelihood estimate of the parameter $a$ based on $n$ observations.   
\end{abstract}

\section{Introduction}
In rate-distortion theory~\cite{shannon1948}~\cite{shannon1959coding}, a source, modeled as a discrete stochastic process $\left\{U_i\right\}_{i = 1}^{\infty}$,  produces a random vector $\rvec{U} \triangleq (U_1,..., U_n)^\top$ and the goal is to represent $\rvec{U}$ by the minimum number of reproduction vectors $\rvec{V}$ such that the distortion is no greater than a given threshold $d$. For any such set of reproduction vectors, the associated rate is defined as the ratio between the logarithm of the number of vectors and $n$. The rate quantifies the minimum number of bits per symbol needed to describe the source with distortion $d$. 

Numerous studies have been pursued since the seminal paper~\cite{shannon1959coding}, where Shannon first proved the rate-distortion theorem for the discrete stationary memoryless sources (DMS) and then sketched the ideas on generalizing to continuous alphabets and the stationary ergodic sources. Shannon's rate-distortion theorem shows that the minimum rate needed to describe a DMS within distortion $d$ is given by the rate-distortion function (RDF) $\mathbb{R}(d)$, which is computed as a solution to a (single-letter) minimal mutual information convex optimization problem. Goblick~\cite{goblick1969coding} proved a coding theorem for a general subclass of strongly mixing stationary sources showing that the RDF is equal to the limit of $n$-letter minimal mutual information. That  limit has exponential computational complexity in general. Computable expressions for the RDF of sources with memory are known only in the following special cases. Gray~\cite{gray1970information} showed a closed-form expression for the RDF for a binary symmetric Markov source with bit error rate distortion in a low distortion regime. For higher distortions, Jalali and Weissman~\cite{jalali2007newbounds} recently showed upper and lower bounds allowing one to compute the rate-distortion function in this case with desired accuracy. Gray~\cite{gray1971markov} showed a lower bound to the rate-distortion function of finite-state finite-alphabet Markov sources with a balanced distortion measure, and the lower bound becomes tight when $d\in (0,d_c]$ for critical distortion $d_c$. For the mean squared error distortion measure (MSE),  Davisson~\cite{davisson1972rate}, and also Kolmogorov~\cite{kolmogorov1956shannon}, derived the rate-distortion function for stationary Gaussian processes by applying a unitary transformation to the process to decorrelate it and applying reverse waterfilling to the decorrelated Gaussians~\cite{cover2012elements}. Berger~\cite{berger1971rate} and Gray~\cite{gray1970information}, in separate contributions in the late 60's and early 70's, derived the MSE rate-distortion function for Gaussian autoregressive sources. See~\cite{berger1998lossy} for a detailed survey on the development of coding theorems for more general sources.

All of the above mentioned work~\cite{shannon1948, shannon1959coding, goblick1969coding, gray1970information, jalali2007newbounds, gray1971markov, davisson1972rate, cover2012elements, berger1971rate} apply to the operational regime where the coding length $n$ grows without bound. Asymptotic coding theorems are important since they set a clear boundary between the achievable and the impossible. However, practical compression schemes are of finite blocklength. A natural, but challenging, question to ask is: for a given coding blocklength $n$, what is the minimum rate to compress the source with distortion at most $d$? Answering this question \emph{exactly} is hard. An easier question is that of second-order analysis, which studies the dominating term in the gap between RDF and the finite blocklength minimum rate.   

In rate-distortion theorems, the landscape of second-order analyses consists of two criteria: average distortion and excess distortion. The average distortion constraint posits that the average distortion should be at most $d$, while excess distortion constraint requires that the probability of distortion exceeding $d$ be at most $\epsilon$. For average distortion criterion, Zhang, Yang and Wei~\cite{zhang1997redundancy} proved that for i.i.d. finite alphabet sources, the leading term in the gap $R(n,d) - \mathbb{R}(d)$ is  $\frac{\log n}{2n}$\footnote{This statement is translated from~\cite{zhang1997redundancy}, where the equivalent result was stated in terms of distortion-rate function.}, where $R(n,d)$ denotes the minimum rate compatible with average distortion $d$ at coding length $n$. Later, Yang and Zhang~\cite{yang1999redundancy} extended the achievability result of~\cite{zhang1997redundancy} to abstract sources.

For lossy compression of i.i.d. sources under excess distortion criterion, the minimum achievable finite blocklength rate admits the following expansion~\cite{kostina2012fixed},~\cite{ingber2011dispersion}, known as the Gaussian approximation: 
\begin{align}
R(n,d,\epsilon) = \mathbb{R}(d) + \sqrt{\frac{V(d)}{n}}Q^{-1}(\epsilon) + O\left(\frac{\log n}{n}\right),
\label{gaussian_approx}
\end{align}
where $V(d)$ is referred to as the source dispersion and $Q^{-1}(\cdot)$ denotes the inverse $Q$-function. Extensions of the result in~\eqref{gaussian_approx} to joint source-channel coding~\cite{kostina2013lossy}~\cite{wang2011dispersion} and  multiterminal source coding~\cite{watanabe2017second}~\cite{tan2014dispersions} have also been studied.

The dispersion of lossy compression of \emph{sources with memory} is unknown. In the context of variable-length lossy compression with guaranteed distortion, Kontoyiannis~\cite[Th. 6, Th. 8]{kontoyiannis2000pointwise} established a connection between the number of bits needed to represent $n$ given samples produced by an arbitrary source, and the logarithm of the reciprocal of distortion $d$-ball probability. Unfortunately computation of that probability has exponential in $n$ complexity. In contrast, the dispersions of lossless compression of sources with memory and of channel coding over channels with memory are known in some cases. The second-order expansion of the minimum encoded length in the lossless compression of Markov sources is computed in~\cite{yushkevich1953limit,kontoyiannis2014lossless}. Polyanskiy et al. found the channel dispersion of the Gilbert-Elliott channel in~\cite[Th. 4]{polyanskiy2011dispersionGE}.

In this paper, we derive an expansion of type~\eqref{gaussian_approx} on $R(n,d,\epsilon)$ for the Gauss-Markov source, one of the simplest models for sources with memory. We show that the dispersion $V(d)$ for the Gauss-Markov source is equal to the limiting variance of the $\mathsf{d}$-tilted information, and has a reverse waterfilling representation. We show that the dispersion $V(d)$ for low distortions is the same as that of the i.i.d. Gaussian noise driving the process, and becomes smaller for distortions above a critical value $d_c$,  which extend the corresponding result of Gray~\cite[Eq. (24)]{gray1970information} to the nonasymptotic regime. Section~\ref{sec:model} presents the problem formulation. The main results and proof techinques are presented in Section~\ref{sec:MainResults}. Our proofs of converse and achievability are presented in Sections~\ref{sec:converse} and~\ref{sec:achievability}, respectively. The converse proof generalizes to the Gaussian autoregressive processes~\cite{gray1970information}, but the achievability proof does not. In proving the converse and achievability, we develop several new tools including a nonasymptotic refinement of Gray's result~\cite[Eq. (19)]{gray1970information} on the eigenvalue distribution of the covariance matrix of the Gauss-Markov source. This refinement relies on a sharp bound on the differences of eigenvalues of two sequences of tridiagonal matrices, proved using the Cauchy interlacing theorem and the Gershgorin circle theorem from matrix theory. In proving achievability, we derive a maximum likelihood estimator of the parameter $a$ of the Gauss-Markov source and bound the estimation error using the Hanson-Wright inequality~\cite[Th. 1.1]{rudelson2013hanson}. Our key tool in the achievability proof is the construction of a typical set based on the maximum likelihood estimator. Finally, we conclude in Section~\ref{sec:conclusion} with brief discussions on some open problems. Fig.~\ref{fig:roadmap} in Appendix~\ref{app:roadmap} presents a roadmap containing the relations of all theorems, corollaries and lemmas in this paper. 

\emph{Notations:} Throughout, lowercase (uppercase) boldface letters denote vectors (random vectors) of length $n$. We omit the dimension when there is no ambiguity, i.e. $\pmb{u}\equiv u^n \equiv (u_1,\ldots,u_n)^\top$ and $\pmb{U}\equiv U^n \equiv (U_1,\ldots,U_n)^\top$. We write $U$ for $U^{\infty}$. For a random variable $X$, we use $\mathbb{E}[X]$ and $\var{X}$ to denote its mean and variance, respectively. We write matrices using sans serif font, e.g. matrix $\mathsf{A}$, and we write $\|\mathsf{A}\|_{F}$ and $\|\mathsf{A}\|$ to denote the Frobenius and operator norms of $\mathsf{A}$, respectively. The trace of $\mathsf{A}$ is denoted by $\text{tr}(\mathsf{A})$. For a vector $\rvec{v}$, we denote by $\|\rvec{v}\|_{p}$ the $\ell_p$-norm of $\rvec{v}$ ($p = 1$ or $p = 2$ in this paper). We also denote the sup norm of a function $F$ by $\|F\|_{\infty} \triangleq \sup_{x\in \mathcal
D}~|F(x)|$, where $\mathcal{D}$ denotes the domain of $F$. We use the standard $O(\cdot),~o(\cdot)$ and $\Theta(\cdot)$ notations to characterize functions according to their asymptotic growth rates. Namely, let $f(n)$ and $g(n)$ be two functions on $n$, then $f(n) = O(g(n))$ if and only if there exists positive real number $M$ and $n_0\in\mathbb{N}$ such that  $| f(n)|\leq M|g(n)|$ for any $n\geq n_0$; $f(n) = o(g(n))$ if and only if $\lim_{n\rightarrow\infty}f(n)/g(n) = 0$; $f(n) = \Theta(g(n))$ if and only there exist positive constants $c_1, c_2$ and $n_0\in\mathbb{N}$ such that $c_1g(n)\leq f(n)\leq c_2g(n)$ for any $n\geq n_0$. For any positive integer $m$, we denote by $[m]$ the set of intergers $\left\{1,2,...,m\right\}$. We denote by $\mathbbm{1}\left\{\cdot\right\}$ the indicator function. We use $n!!$ to denote the double factorial of $n$. The imaginary unit is denoted by $\mathrm{j}$. All exponents and logarithms are base $e$.
\section{Problem Formulation}
\label{sec:model}

\subsection{Operational definitions}
In single-shot lossy compression, we consider source and reproduction alphabets $\mathcal{X}$ and $\mathcal{Y}$, and a given source distribution $P_X$ over $\mathcal{X}$. The distortion measure is a mapping $\mathsf{d}(\cdot, \cdot)\colon \mathcal{X} \times \mathcal{Y} \mapsto [0, +\infty)$. An encoder $\mathsf{f}$ is a mapping $\mathsf{f}\colon\mathcal{X}\mapsto [M]$, and a decoder is $\mathsf{g}\colon [M]\mapsto \mathcal{Y}$. The image set of a decoder $\mathsf{g}$ is referred to as a codebook consisting of $M$ codewords $\{\mathsf{g}(i)\}_{i = 1}^M$. Given distortion threshold $d>0$ and excess-distortion probability $\epsilon \in (0,1)$, an $(M, d, \epsilon)$ code consists of an encoder-decoder pair $(\mathsf{f}, \mathsf{g})$ such that $\prob{\mathsf{d}(X, \mathsf{g}(\mathsf{f}(X))) > d} \leq \epsilon$. The nonasymptotic fundamental limit of lossy compression is the minimum achievable code size for a given distortion threshold $d$ and an excess-distortion probability $\epsilon \in (0,1)$:
\begin{align}
M^\star(d, \epsilon) \triangleq \min\left\{M: ~\exists \text{ an } (M, d, \epsilon)\text{ code} \right\}. \label{def:Mstar}
\end{align} 

In this paper, $\mathcal{X} = \mathcal{Y} = \mathbb{R}^n$, and the distortion measure is  the mean squared error (MSE) distortion: $\forall \pmb{u},~\pmb{v}\in\mathbb{R}^n$, 
\begin{align}
\mathsf{d}(\pmb{u}, \rvec{v}) \triangleq \frac{1}{n} \|\pmb{u} - \pmb{v}\|_2^2.
\end{align}
We refer to the set $\mathcal{B}(\rvec{x}, d)$, defined below, as a distortion $d$-ball centered at $\rvec{x}$:
\begin{align}
\mathcal{B}(\rvec{x}, d) \triangleq \left\{\rvec{x}'\in\mathbb{R}^n:  \dis{\rvec{x}}{\rvec{x}'}\leq d\right\}.
\label{def:d_ball}
\end{align}
We consider the Gauss-Markov source $\left\{U_i\right\}_{i = 1}^{\infty}$, which satisfies the following difference equation:
\begin{align}
U_i =  a U_{i-1} + Z_i, \quad \forall i\geq 1, 
\label{def:GaussMarkov}
\end{align}
and $U_0 = 0$. Here, $a\in [0,1)$ is the gain\footnote{Note that if $a\in (-1,0]$ in~\eqref{def:GaussMarkov}, then $\left\{(-1)^iU_i\right\}_{i = 0}^{\infty}$ is a Gauss-Markov source with nonnegative gain $-a$ and the same innovation variance. Thus restricting to $0\leq a<1$ is without loss of generality.}, and  $Z_i$'s are independently and identically distributed (i.i.d.) $\mathcal{N}(0, \sigma^2)$ that form the \emph{innovation} process. We adopt~(\ref{def:GaussMarkov}) as the simplest model capturing information sources with memory: gain $a$ determines how much memory (as well as the growth rate), and the innovation $Z_i$ represents new randomness being generated at each time step. In statistics, the Gauss-Markov source~\eqref{def:GaussMarkov} is also known as the first-order Gaussian autoregressive (AR) process.  

For a fixed blocklength $n\in\mathbb{N}$, a distortion threshold $d>0$ and an excess-distortion probability $\epsilon \in (0,1)$, an $(n, M, d, \epsilon)$ code consists of an encoder $\mathsf{f}_n\colon \mathbb{R}^n\mapsto [M]$ and a decoder $\mathsf{g}_n\colon [M] \mapsto \mathbb{R}^n$ such that $\mathbb{P}\left[\dis{\rvec{U}}{\mathsf{g}_n(\mathsf{f}_n(\rvec{U}))} \geq d \right] \leq \epsilon$, where $\rvec{U} = (U_1, \ldots, U_n)^\top$ denotes the source vector. The rate associated with an $(n, M, d, \epsilon)$ code is $R\triangleq \frac{\log M}{n}$. The nonasymptotic operational fundamental limit, that is, the minimum achievable code size at blocklength $n$, distortion $d>0$ and excess-distortion probability $\epsilon \in (0,1)$, is 
\begin{align}
M^\star(n,d, \epsilon) \triangleq \min\left\{M: \exists \text{ an }(n, M, d,\epsilon)\text{ code} \right\}, 
\end{align} 
and the corresponding minimum source coding rate is 
\begin{align}
R(n,d,\epsilon) \triangleq \frac{\log M^\star(n,d,\epsilon)}{n}. \label{Rstar}
\end{align}
The objective of this paper is to characterize $R(n, d, \epsilon)$ for the Gauss-Markov source.

\subsection{Informational definitions}
\label{subsec:sourcemodel}

The problem of characterizing the operational fundamental limit $M^\star(d, \epsilon)$ in~\eqref{def:Mstar} is closely related to the rate-distortion function (RDF) $\mathbb{R}_{X}(d)$ of the source $X$, which is defined as the solution to the following convex optimization problem~\cite{shannon1959coding}:  
\begin{align}
\mathbb{R}_{X}(d)  \triangleq  & \inf_{P_{Y|X}\colon \mathbb{E}[\mathsf{d}(X, Y)] \leq d} ~I(X; Y),  \label{eqn:cvx_rdf}
\end{align}
where the infimum is over all conditional distributions $P_{Y|X}\colon\mathcal{X} \mapsto \mathcal{Y}$ such that the expected distortion is less than or equal $d$, and $I(X; Y)$ denotes the mutual information between $X$ and $Y$. In this paper, we assume that 
\begin{enumerate}
\item $\mathbb{R}_X(d)$ is differetiable with respect to $d$;

\item there exists a minimizer in~\eqref{eqn:cvx_rdf}.
\end{enumerate}
The pair $(X, Y^\star)$ is referred to as the \emph{RDF-achieving pair} if $P_{Y^\star|X}$ is the minimizer in~\eqref{eqn:cvx_rdf}.  For any $x\in \mathcal{X}$, the $\mathsf{d}$-tilted information $\jmath_{X}(x, d)$ in $x$, introduced in~\cite[Def. 6]{kostina2012fixed}, is 
\begin{align}
\jmath_X(x, d) \triangleq -\lambda^\star d - \log \mathbb{E}\exp\left(-\lambda^\star\mathsf{d}(x, Y^\star)\right), \label{def:dtilted}
\end{align}
where $\lambda^\star$ is the negative slope of the curve $\mathbb{R}_X(d)$ at distortion $d$: 
\begin{align}
\lambda^\star\triangleq -\mathbb{R}'_X(d). \label{eqn:def_slope}
\end{align}
The $\mathsf{d}$-tilted information $\jmath_X(X, d)$ has the property that 
\begin{align}
\mathbb{R}_{X}(d) = \mathbb{E}[\jmath_X(X, d)].\label{relation:mean_rdf}
\end{align}
When $\mathcal{X}$ is a finite set,~\eqref{relation:mean_rdf} follows immediately from the Karush–Kuhn–Tucker (KKT) conditions for the optimization problem~\eqref{eqn:cvx_rdf}, see~\cite[Th. 9.4.1]{gallager1968information} and~\cite[Eq. (2.5.16)]{berger1971rate}. Csisz{\'a}r showed the validity of~\eqref{relation:mean_rdf} when $\mathcal{X}$ is an abstract probability space~\cite[Corollary, Lem. 1.4, Eqs. (1.15), (1.25), (1.27)-(1.32)]{csiszar1974extremum}, see Appendix~\ref{app:just} for a concise justification. For more properties of the $\mathsf{d}$-tilted information, see~\cite[Eq. (17)-(19)]{kostina2012fixed}.

Next, we introduce the conditional relative entropy minimization (CREM) problem, which plays a key role in our development. Let $P_X$ and $P_Y$ be probability distributions defined on alphabets $\mathcal{X}$ and $\mathcal{Y}$, respectively. For any $d>0$, the CREM problem is defined as 
\begin{align}
\mathbb{R}(X, Y, d) \triangleq   \inf_{P_{F|X}\colon \mathbb{E}[\mathsf{d}(X, F)] \leq d} ~D(P_{F|X} || P_Y|P_X), \label{eqn:cvx_crem}
\end{align}
where $F$ is a random variable taking values in $\mathcal{Y}$, and $D(P_{F|X} || P_Y|P_X)$ is the conditional relative entropy: 
\begin{align}
D(P_{F|X} || P_Y|P_X) \triangleq \int D(P_{F|X=x} || P_Y)dP_X(x). \label{def:conditionalentropy}
\end{align}
A well-known fact in the study of lossy compression is that the CREM problem~\eqref{eqn:cvx_crem} is related to the RDF~\eqref{eqn:cvx_rdf} as
\begin{align}
 \mathbb{R}_X(d) = \inf_{P_Y}~\mathbb{R}(X, Y, d), \label{rel:BA}
\end{align}
where the infimization is over all probability distributions $P_Y$ of the random variables $Y$ over $\mathcal{Y}$ that are independent of $X$; and the equality in~\eqref{rel:BA} is achieved when $P_Y$ is the $Y^\star$-marginal of the RDF-achieving pair $(X, Y^\star)$, see~\cite[Eq. (10.140)]{cover2012elements} and~\cite[Th. 4]{blahut1972computation} for the finite alphabets $\mathcal{X}$;~\cite[Eq. (3.3)]{yang1999redundancy} and~\cite[Eq. (13)]{kontoyiannis2000pointwise} for abstract alphabets $\mathcal{X}$. The property~\eqref{rel:BA} is a foundation of the Blahut–Arimoto algorithm, which computes iterative approximations to $\rdf_X(d)$ by alternating between inner and outer infimizations in~\eqref{rel:BA}. The CREM problem is also important in nonasymptotic analyses of lossy compressors, see~\cite[Eq. (3.3)]{yang1999redundancy},~\cite[Eq. (13)]{kontoyiannis2000pointwise} and~\cite[Eq. (27)]{kostina2012fixed}. Operationally, it relates to the mismatched-codebooks problem, that is, lossy compression of source $P_X$ using random codewords drawn from $P_{Y}$~\cite[Th. 12]{dembo2002source}. Similar to~\eqref{def:dtilted}, $\forall x\in\mathcal{X}, ~\delta >0,~ d>0$, the \emph{generalized tilted information} $\Lambda_Y(x, \delta, d)$, defined in~\cite[Eq. (28)]{kostina2012fixed}, is
\begin{align}
\Lambda_Y(x, \delta, d) \triangleq -\delta d - \log \mathbb{E}\exp\left(-\delta \mathsf{d}(x, Y)\right). \label{def:generalizeddtilted}
\end{align}
The optimizer $P_{F^\star|X}$ of~\eqref{eqn:cvx_crem} satisfies the following condition: $\forall x\in\mathcal{X}, y\in \mathcal{Y}$, 
\begin{align}
\log\frac{dP_{F^\star|X}(y|x)}{dP_{Y}(y)} = \Lambda_Y(x, \delta^\star, d) - \delta^\star\mathsf{d}(x, y) + \delta^\star d, \label{relation:generaltilted}
\end{align}
where
\begin{align}
\delta^\star \triangleq - \mathbb{R}'(X, Y, d).\label{CREM_slope}
\end{align}
When $\mathcal{X}$ and $\mathcal{Y}$ are discrete, \eqref{relation:generaltilted} can be verified by the KKT conditions for the optimization problem~\eqref{eqn:cvx_crem}. For abstract alphabets $\mathcal{X}$, see~\cite[Th. 2]{dembo2002source} and~\cite[Property 1]{yang1999redundancy} for an exposition. By comparing~\eqref{def:dtilted} and~\eqref{def:generalizeddtilted} and using the relation~\eqref{rel:BA}, we see that
\begin{align}
\jmath_X(x, d) = \Lambda_{Y^\star}(x, \lambda^\star, d), \label{rel:dtiltedandtilted}
\end{align}
where $\lambda^\star$ is in~\eqref{eqn:def_slope}.

For the Gauss-Markov source defined in~\eqref{def:GaussMarkov}, its $n$-th order rate-distortion function $\rdf_{\rvec{U}} (n,d)$ is defined by replacing $X$ by $\rvec{U}$ in~\eqref{eqn:cvx_rdf} and then normalizing by $n$:
\begin{align}
\rdf_{\rvec{U}} (n,d) \triangleq \frac{1}{n}\inf_{\substack{ P_{\rvec{V}|\rvec{U}}: \\ \mathbb{E}\left[\dis{\rvec{U}}{\rvec{V}}\right]\leq d }} I(\rvec{U}; \rvec{V}). \label{eqn:nthorderOp}
\end{align}
The rate-distortion function $\mathbb{R}_U(d)$ for the Gauss-Makov source~\eqref{def:GaussMarkov} is
\begin{align}
\rdf_{U}  (d) \triangleq \limsup_{n\rightarrow\infty} ~\rdf_{\rvec{U}}  (n,d).\label{def:RDF}
\end{align}
It immediately follows from~\eqref{relation:mean_rdf} that 
\begin{align}
\rdf_{U}  (d)  =   \limsup_{n\rightarrow\infty} \frac{1}{n}\myE{\jmath_{\rvec{U}}(\rvec{U}, d)}, \label{limiting_expectation_relation}
\end{align}
where $\jmath_{\rvec{U}}(\rvec{U}, d)$ is the $\mathsf{d}$-tilted information random variable defined in~\eqref{def:dtilted}, that is,
\begin{align}
\jmath_{\rvec{U}}(\rvec{u}, d) = -\lambda^\star nd - \log\mathbb{E}\left[\exp\left(-\lambda^\star n \dis{\rvec{u}}{\rvec{V}^\star}\right)\right],
\label{def:d_tilted_def}
\end{align}
where $(\rvec{U}, \rvec{V}^\star)$ forms a RDF-achieving pair in~\eqref{eqn:nthorderOp} and 
\begin{align}
\lambda^\star = -\rdf'_{\rvec{U}} (n,d).
\label{lambda_double_star}
\end{align}
The variance of the $\mathsf{d}$-tilted information is important in capturing the second-order performance of the best source code. Define 
\begin{align}
\mathbb{V}_{\rvec{U}}(n, d) &\triangleq \var{\jmath_{\rvec{U}}(\rvec{U}, d)},\\
\mathbb{V}_U(d) &\triangleq \limsup_{n\rightarrow\infty}\frac{1}{n}\mathbb{V}_{\rvec{U}}(n, d). \label{def:info_dispersion}
\end{align}
The quantity $\mathbb{V}_U(d)$ is referred to as the \emph{informational dispersion}, in contrast to the operational dispersion $V_{U}(d)$ defined in the next subsection. Reverse waterfilling solutions for rate-distortion functions of the Gauss-Markov source were well-known, see~\cite[Eq. (15)]{gray1970information} for $\mathbb{R}_{\rvec{U}}(n, d)$, ~\cite[Eq. (22)]{gray1970information} for $\mathbb{R}_U(d)$ and our discussions in Section~\ref{subsec:RelatedWorks} below. In this paper, we derive similar parametric expressions for both $\mathbb{V}_{\rvec{U}}(n, d)$ and $\mathbb{V}_U(d)$.

\subsection{Operational fundamental limits}
\label{subsec:OpDis}

In terms of coding theorems, the equality between $R(d)$ (the minimum achievable source coding rate under average distortion criterion when the blocklength $n$ goes to infinity) and $\rdf_{U}(d)$ (the informational rate-distortion function defined in~\eqref{def:RDF}) has been established, e.g.~\cite[Th. 6.3.4]{berger1971rate} and~\cite[Th. 2]{gray1970information}. For the Gauss-Markov source, by the Markov inequality, the achievability result under average distortion criterion, e.g.~\cite[Th. 2]{gray1970information}, can be converted into an achievability result under the excess distortion criterion. A matching converse follows from Kieffer's strong converse~\cite[Th. 1]{kieffer1991strong} for the stationary ergodic sources. Therefore, for any $d>0$ and $\epsilon\in (0,1)$, we have 
\begin{align}
\lim_{n\rightarrow\infty}R(n,d,\epsilon) = \rdf_{U} (d), \label{eqn:limit_rnd}
\end{align} 
where $R(n, d, \epsilon)$ is defined in~\eqref{Rstar} and $\rdf_{U} (d)$ in~\eqref{def:RDF}.

The main result of this paper is the following Gaussian approximation for the minimum achievable rate $R(n, d, \epsilon)$ in lossy compression of the Gauss-Markov source~\eqref{def:GaussMarkov}:
\begin{align}
R(n, d, \epsilon) = \rdf_{U}  (d) + \sqrt{\frac{\mathbb{V}_U(d)}{n}}Q^{-1}(\epsilon) + o\left(\frac{1}{\sqrt{n}}\right), \label{mainresult:preview}
\end{align}
where $Q^{-1}(\cdot)$ denotes the inverse $Q$-function, and the term $o(\cdot)$ will be refined in Theorems~\ref{thm:converse} and~\ref{thm:achievability} in Sections~\ref{sec:converse} and~\ref{sec:achievability} below. Our main result~\eqref{mainresult:preview} is a nonasymptotic refinement of~\eqref{eqn:limit_rnd}, implying that the convergence rate in the limit~\eqref{eqn:limit_rnd} is of order $\frac{1}{\sqrt{n}}$ with the optimal constant factor given in~\eqref{mainresult:preview}. Formally, the rate-dispersion function $V_U(d)$, introduced in~\cite[Def. 7]{kostina2012fixed} and simply referred to as \emph{(operational) dispersion}, is
\begin{align}
V_U(d) \triangleq \lim_{\epsilon\rightarrow 0}\lim_{n\rightarrow\infty} n\left(\frac{R(n,d,\epsilon) - \mathbb{R}_{U}(d)}{Q^{-1}(\epsilon)}\right)^2. 
\label{def:op_dispersion}
\end{align} 
Equivalently, our main result~\eqref{mainresult:preview}  establishes the equality between the operational and informational dispersions for the Gauss-Markov source:
\begin{align}
V_U(d) = \mathbb{V}_U(d).
\end{align}

\subsection{Related work}
\label{subsec:RelatedWorks}

The $n$-th order RDF $\mathbb{R}_{\rvec{U}}(n, d)$ defined in~\eqref{eqn:nthorderOp} for the Gauss-Markov source is given by the reverse waterfilling~\cite[Eq. (17)]{gray1970information} and~\cite[Eq. (6.3.34)-(6.3.36)]{berger1971rate}:
\begin{align}
\mathbb{R}_{\rvec{U}}(n,d) &= \frac{1}{n} \sum_{i = 1}^n \max\left(0, \frac{1}{2}\log \frac{\sigma_i^2}{\theta_n}\right), \label{eqn:para_r_n}\\
d &= \frac{1}{n}\sum_{i = 1}^n\min(\theta_n, \sigma_i^2), \label{eqn:para_d}
\end{align}
where $\sigma_i^2$'s are the eigenvalues of the covariance matrix (see the discussions in Section~\ref{subsec:pfoutline} below): 
\begin{align}
\mathsf{\Sigma}_{\rvec{U}}\triangleq \mathbb{E}[\rvec{U}\rvec{U}^\top], \label{SigmaU}
\end{align}
and $\theta_n > 0$ is the water level matched to $d$ at blocklength $n$. The rate-distortion function $\mathbb{R}_U(d)$ for the Gauss-Markov source~\eqref{def:GaussMarkov} is obtained by passing to the limit of infinite $n$ in~\eqref{eqn:para_r_n} and \eqref{eqn:para_d} via invoking the limiting theorems on the eigenvalues of the covariance matrix $\mathsf{\Sigma}_{\rvec{U}}$~\cite[Eq. (22)]{gray1970information} and~\cite[Th. 6.3.2]{berger1971rate}, given by 
\begin{align}
\mathbb{R}_{U}(d)&= \frac{1}{2\pi} \int_{-\pi}^\pi \max\left[0, \frac{1}{2}\log \frac{S(w)}{\theta}\right]~dw,\label{eqn:para_rate_inf}\\
d &=  \frac{1}{2\pi}\int_{-\pi}^\pi \min\left[\theta, S(w)\right]~dw\label{eqn:para_d_inf},
\end{align}
where the power spectrum of the Gauss-Markov source~\eqref{def:GaussMarkov} is given by 
\begin{align}
S(w) = \frac{\sigma^2}{g(w)}, \label{spectrum}
\end{align}
and the function $g$ is defined as 
\begin{align}
g(w) \triangleq 1 + a^2 -2a\cos(w), \quad \forall w\in [-\pi, \pi]. \label{exp:g}
\end{align}
We refer to~\eqref{eqn:para_r_n}-\eqref{eqn:para_d} as the $n$-th order reverse waterfilling, and to~\eqref{eqn:para_rate_inf}-\eqref{eqn:para_d_inf} as the limiting reverse waterfilling. Fig.~\ref{fig:RWF} depicts the limiting reverse waterfilling~\eqref{eqn:para_d_inf}. 

Results similar to~\eqref{eqn:para_r_n}-\eqref{eqn:para_d_inf} hold for the stationary Gaussian processes~\cite[Eq.~(17)-(18)]{kolmogorov1956shannon}, as well as for the higher-order Gaussian AR processes (not necessarily stationary)~\cite[Eq. (22)]{gray1970information}. We discuss the subtle differences between the rate-distortion functions of the (asymptotically) stationary and nonstationary Gaussian AR processes in Section~\ref{sec:conclusion} below. This paper considers the (asymptotically) stationary Gauss-Markov sources, i.e., \eqref{def:GaussMarkov} with $|a| <1$.  

The converse results in this paper extend partly to the higher-order Gaussian AR processes, studied by Gray~\cite{gray1970information} and Berger~\cite[Sec. 6.3.2]{berger1971rate}. The Gaussian AR process is~\cite[Eq. (1)]{gray1970information}
\begin{align}
U_i = \sum_{\ell = 1}^{i} a_\ell U_{i-\ell} + Z_i,\quad i\geq 1,
\label{eqn:GeneralGaussianAR}
\end{align}
and $U_i = 0$ for $i\leq 0$, where $Z_i$'s are i.i.d. $\mathcal{N}(0,\sigma^2)$, and the real constants $a_\ell$'s satisfy~\cite[Eq. (10)]{gray1970information}
\begin{align}
\sum_{\ell = 0}^{\infty} |a_\ell| < \infty. \label{condition_general}
\end{align}
The Gauss-Markov source in~\eqref{def:GaussMarkov} is a special case of~\eqref{eqn:GeneralGaussianAR} with $a_1 = a$ and $a_\ell = 0$ for $\ell\geq 2$. The following relation between the rate-distortion functions of the Gaussian AR process $\{U_i\}_{i = 1}^{+\infty}$ in~\eqref{eqn:GeneralGaussianAR} and the i.i.d. Gaussian process $\{Z_i\}_{i = 1}^{+\infty}$ is due to Gray~\cite[Eq. (24)]{gray1970information}:
\begin{align}
\begin{cases}
\mathbb{R}_{U}(d) = \mathbb{R}_{Z}(d), & 0 <  d\leq d_c, \\
\mathbb{R}_{U}(d) >  \mathbb{R}_{Z}(d), & d_c < d \leq \dmax,
\end{cases}
\label{eqn:rel_UZ}
\end{align}
where $d_c$ is referred to as the \emph{critical distortion}, defined as $d_c \triangleq \theta_{\min}$, where 
\begin{align}
\theta_{\min} \triangleq \min_{w\in [-\pi,\pi]} S(w).
\label{def:general_d_c}
\end{align}
Accordingly, denote the maximum value of $S(w)$ over the inverval $[-\pi, \pi]$ as 
\begin{align}
\theta_{\max} \triangleq \max_{w\in [-\pi,\pi]} S(w). \label{theta_max}
\end{align}
In~\eqref{eqn:rel_UZ}, $\dmax$ is the \emph{maximum distortion} achievable in~\eqref{eqn:para_d_inf} (that is, when $\theta \geq \theta_{\max}$):
\begin{align}
\dmax \triangleq \frac{1}{2\pi}\int_{-\pi}^{\pi}~S(w)~dw, 
\label{def:dmax}
\end{align}
and $\mathbb{R}_{Z}(d)$ is the RDF for i.i.d. Gaussian sources derived by Shannon~\cite[Equation below Fig. 9]{shannon1959coding}:
\begin{align}
\mathbb{R}_{Z}(d) = \max\left(0, \frac{1}{2}\log \frac{\sigma^2}{d}\right ).  \label{RDF:iidGaussian}
\end{align}
The power spectrum $S(w)$ of the Gaussian AR process is~\cite[Eq. (21)]{gray1970information}
\begin{align}
S(w) = \sigma^2 \left|\sum_{\ell = 0}^{+\infty}a_\ell e^{-\mathrm{j}\ell w}\right|^{-2}, \quad w\in [-\pi, \pi].
\end{align}
Equality in~\eqref{eqn:rel_UZ} is a deep result stating that in a range of low distortions, the asymptotic rate-distortion tradeoff of a Gaussian AR process and that of its driving innovation process are the same. See Fig.~\ref{fig:RDF} for an illustration of~\eqref{eqn:rel_UZ} in the special case of a Gauss-Markov source with $a = 0.5$.

The critical distortion $d_c$ and the maximum distortion $\dmax$ can be understood pictorially as follows. In Fig.~\ref{fig:RWF} and equivalently in~\eqref{eqn:para_d_inf}, as the water level $\theta$ rises from 0 to $\theta_{\min}$,  the minimum on the right side of~\eqref{eqn:para_d_inf} equals $\theta$, meaning that $d = \theta$ for $0 < \theta \leq \theta_{\min}$ (equivalently, $0\leq d\leq d_c$). As the water level $\theta$ rises further, lower parts of the spectrum $S(w)$ start to play a role in~\eqref{eqn:para_d_inf}. When the water level $\theta$ rises above the peak in Fig.~\ref{fig:RWF}: $\theta \geq \theta_{\max}$, the distortion $d$ in~\eqref{eqn:para_d_inf} remains as $\dmax$. In the case of the Gauss-Markov source, from~\eqref{spectrum}, it is easy to see that $d_c$ and $\dmax$ are given by 
\begin{align}
d_c &= \frac{\sigma^2}{(1+|a|)^2}, \label{critical_distortion} \\
\dmax &= \frac{\sigma^2}{1-a^2}. \label{d_max}
\end{align}
Note that $\dmax$ in~\eqref{d_max} equals the stationary variance of the source (Appendix~\ref{app:int_dmax}), i.e.,
\begin{align}
\dmax = \lim_{n\rightarrow\infty}\var{U_n}.
\label{variance:U_lim}
\end{align}
For the nonstationary Gauss-Markov sources ($|a| \geq 1$), $\dmax = +\infty$.

\begin{figure*}[t]
    \centering
        \includegraphics[width=0.7\textwidth, height = 0.49\textwidth]{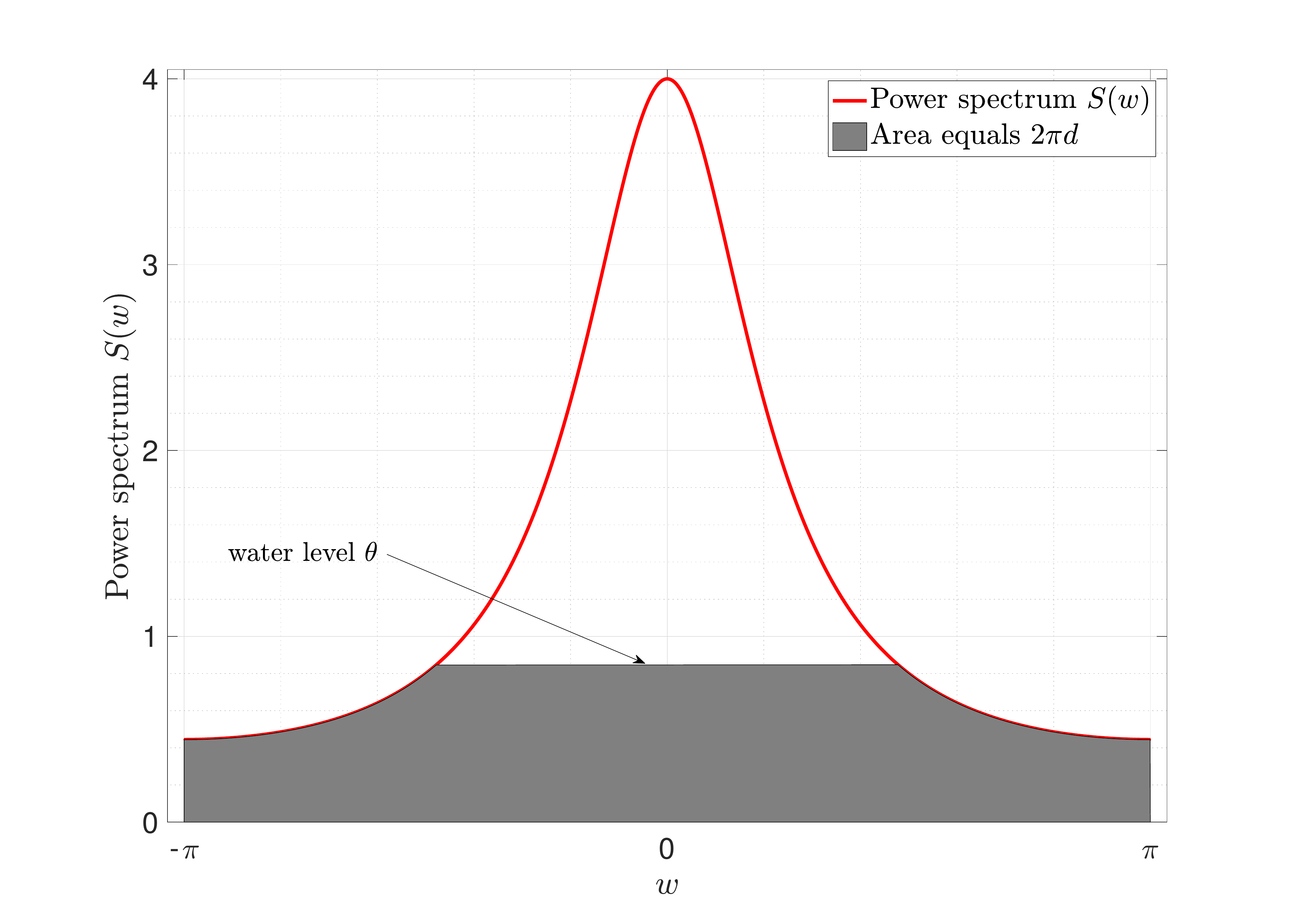}
        \caption{Reverse waterfilling~\eqref{eqn:para_d_inf} for $a = 0.5$: the water level $\theta$ is chosen such that the shaded area equals $2\pi d$.}
        \label{fig:RWF}
\end{figure*}

\begin{figure*}[t]
       \centering
        \includegraphics[width=0.7\textwidth, height = 0.49\textwidth]{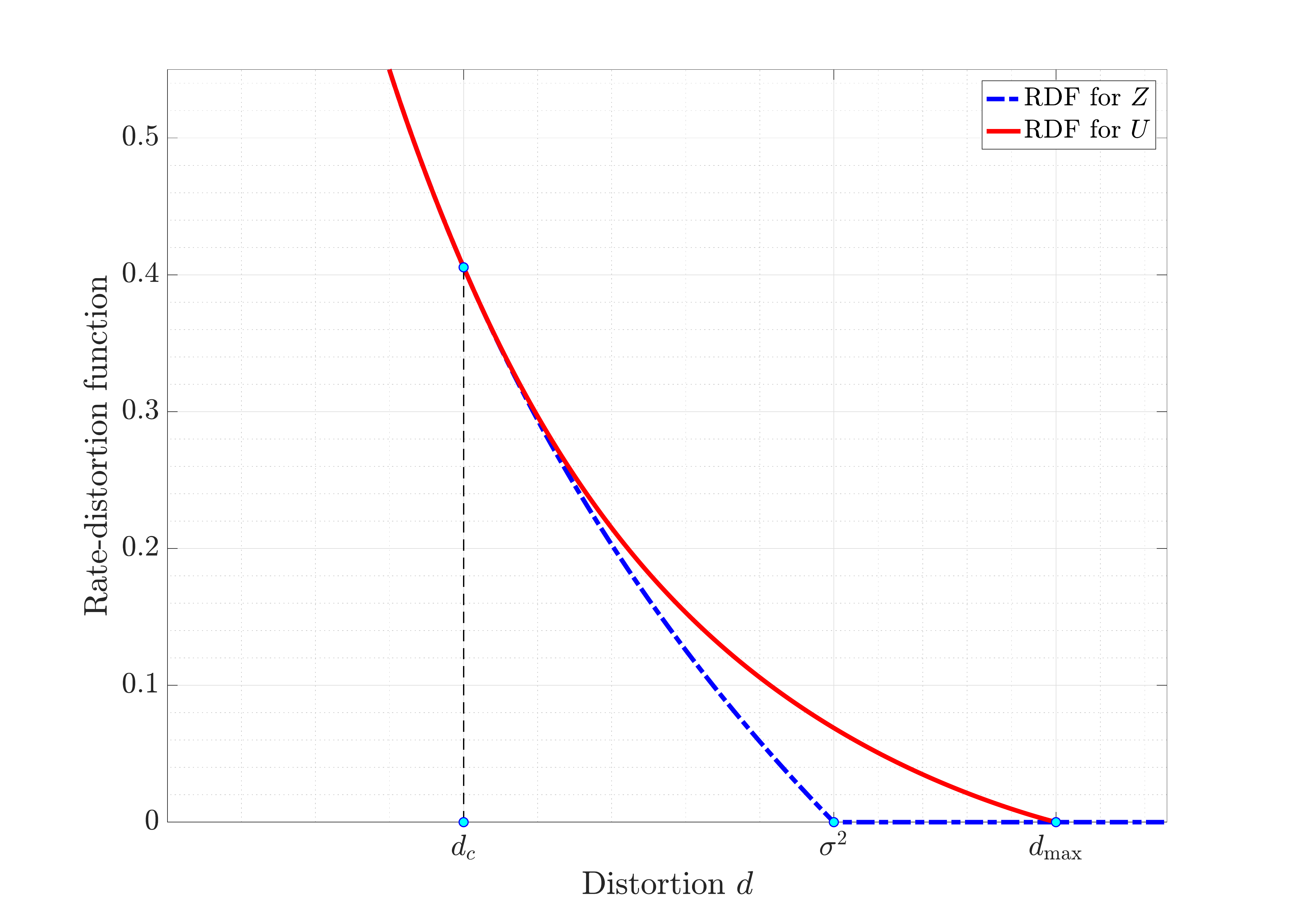}
        \caption{The rate-distortion functions for the Gauss-Markov source $\left\{U_i\right\}$ with $a = 0.5$ and $\sigma^2 = 1$, and for the innovation process $\lrbb{Z_i}_{i = 1}^{\infty}$, $Z_i\sim\mathcal{N}(0,1)$ driving that source.}
        \label{fig:RDF}
\end{figure*}


\section{Main Results}
\label{sec:MainResults}

\subsection{Preliminary: decorrelation}
\label{subsec:pfoutline}

We first make a simple but important observation on the equivalence between lossy compression of the Gauss-Markov sources and parallel independent Gaussian sources. For any $n\in \mathbb{N}$, the random vector $\rvec{U} = (U_1, \ldots, U_n)^\top$ generated by the model~\eqref{def:GaussMarkov} follows the multivariate Gaussian distribution $\mathcal{N}(\rvec{0}, \mathsf{\Sigma}_{\rvec{U}})$, where $\mathsf{\Sigma}_{\rvec{U}} = \sigma^2 (\mathsf{A}^\top\mathsf{A})^{-1}$ is its covariance matrix and $\mathsf{A}$ is an $n\times n$ lower triangular matrix with $\det \mathsf{A} = 1$: 
\begin{align}
\mathsf{A} \triangleq \begin{bmatrix}
1  & 0  &  0  & \ldots & 0 \\
-a & 1  & 0   & \ldots & 0\\
0  & -a & 1   & \ldots & 0\\
\vdots & \ddots & \ddots & \ddots & \vdots\\
0 & \ldots & 0 & -a & 1
\end{bmatrix}.
\label{def:A}
\end{align}
Since~\eqref{def:GaussMarkov} can be rewritten as $\rvec{Z} = \mathsf{A}\rvec{U}$, and $\rvec{Z}\sim\mathcal{N}(\mathbf{0}, \sigma^2 \mathsf{I})$, the covariance matrix of $\rvec{U}$ is given by
\begin{align}
\mathsf{\Sigma}_{\rvec{U}} =  \mathbb{E}[\mathsf{A}^{-1}\rvec{Z}\rvec{Z}^\top(\mathsf{A}^{-1})^\top] = \sigma^2 (\mathsf{A}^\top\mathsf{A})^{-1}.
\label{eqn:cov_u}
\end{align}
We refer to the random vector $\rvec{X}$ as the \emph{decorrelation} of $\rvec{U}$:
\begin{align}
\rvec{X} \triangleq \mathsf{S}^\top\rvec{U}, 
\label{def:decorrelation}
\end{align}
where $\mathsf{S}$ is the unitary matrix in the eigendecomposition of the positive definite matrix $(\mathsf{A}^\top\mathsf{A})^{-1}$:
\begin{align}
(\mathsf{A}^\top\mathsf{A})^{-1} & = \mathsf{S}\mathsf{\Lambda} \mathsf{S}^\top, 
\label{eqn:B_eigen_decom} \\
\mathsf{\Lambda} &= \mathrm{diag}\LRB{\frac{1}{\mu_1},...,\frac{1}{\mu_n}},\label{diag:B}
\end{align}
where $0< \mu_1 \leq  \ldots \leq \mu_n$ are the eigenvalues of $\mathsf{A}^\top\mathsf{A}$. From~\eqref{eqn:cov_u} and~\eqref{eqn:B_eigen_decom}, it is clear that $\rvec{X}\sim\mathcal{N}(\rvec{0},  \sigma^2\mathsf{\Lambda})$, i.e., $X_1,\ldots, X_n$ are independent zero-mean Gaussian random variables with variances $\sigma_i^2$'s being the eigenvalues of $\mathsf{\Sigma}_{\rvec{U}}$:
\begin{align}
 \sigma_i^2 \triangleq \frac{\sigma^2}{\mu_i}, \quad i\in [n].
 \label{eig:sigma_i}
\end{align}

Since they are related via the unitary transformation $\mathsf{S}$ which preserves the geometry of the underlying Euclidean space, $\rvec{U}$ and $\rvec{X}$ are equivalent in terms of their fundamental limits. Indeed, any $(n, M, d, \epsilon)$ code for $\rvec{U}$ (recall the definition in~Section~\ref{subsec:sourcemodel} above) can be transformed, via $\mathsf{S}$, into an $(n, M, d, \epsilon)$ code for $\rvec{X}$, and vice versa; therefore, the finite blocklength minimum achievable rates $R(n,d,\epsilon)$ for $\rvec{U}$ and $\rvec{X}$ are the same. Since $I(\mathsf{S}\rvec{X}; \mathsf{S}\rvec{Y}) = I(\rvec{X}; \rvec{Y})$ and $\mathbb{E}\|\mathsf{S}\rvec{X} - \mathsf{S}\rvec{Y}\|_2^2 =\mathbb{E} \|\rvec{X} - \rvec{Y}\|_2^2$,  their $n$-th order and limiting rate-distortion functions are the same: $\forall n\in\mathbb{N},~d\in (0, \dmax)$, we have $\mathbb{R}_{\rvec{X}}(n, d) = \mathbb{R}_{\rvec{U}}(n, d),$ and hence $\mathbb{R}_{X}(d) = \mathbb{R}_{U}(d).$ By the same transformation, it is easy to verify that this equivalence also extends to the $\mathsf{d}$-tilted information: $\forall \rvec{u}\in\mathbb{R}^n$, let 
\begin{align}
\rvec{x} \triangleq \mathsf{S}^\top \rvec{u},\label{eqn:xu}
\end{align}
then 
\begin{align}
\jmath_{\rvec{U}}(\rvec{u}, d) = \jmath_{\rvec{X}}(\rvec{x}, d). \label{eqiv:DTUX}
\end{align}
Due to the above equivalence, we will refer to both $\rvec{U}$ and its decorrelation $\rvec{X}$ in our analysis. Decorrelation is a well-known tool, which was used to find the rate-distortion functions for the general Gaussian AR processes~\cite{gray1970information}.

\subsection{Gaussian approximations for the Gauss-Markov sources}
\label{subsec:GA}
We now formally state the main contributions of this paper. 
\begin{theorem} 
For the Gauss-Markov source in~\eqref{def:GaussMarkov} with $a\in [0, 1)$, fix any excess-distortion probability $\epsilon\in (0,1)$ and distortion threshold $d\in (0,\dmax)$, where $\dmax$ is defined in~\eqref{d_max}. The minimum achievable source coding rate for the Gauss-Markov source in~(\ref{def:GaussMarkov}) satisfies
\begin{align}
R(n,d,\epsilon) = \rdf_{U} (d) + \sqrt{\frac{V_U(d)}{n}}Q^{-1}(\epsilon) + o\left(\frac{1}{\sqrt{n}}\right),
\label{eqn:main_GA}
\end{align}
where $\rdf_{U}(d)$ is the rate-distortion function of the Gauss-Markov source, given in~\eqref{eqn:para_rate_inf}; and the operational dispersion $V_U(d)$, defined in~\eqref{def:op_dispersion}, is given by
\begin{align}
V_U(d) = \frac{1}{4\pi}\int_{-\pi}^{\pi} \min\left[1, \left(\frac{S(w)}{\theta}\right)^2\right]~dw, \label{eqn:wf_dis}
\end{align}
where $\theta > 0$ is the water level matched to the distortion $d$ via~\eqref{eqn:para_d_inf}, and the power spectrum $S(w)$ is in~\eqref{spectrum}.
\label{thm:main_thm_intro}
\end{theorem}

The proof of Theorem~\ref{thm:main_thm_intro} is in given in Sections~\ref{sec:converse} (converse) and~\ref{sec:achievability} (achievability).
 
Pleasingly, the new reverse waterfilling solution for the dispersion in~\eqref{eqn:wf_dis} parallels the classical reverse waterfilling representation of the rate-distortion function in~\eqref{eqn:para_rate_inf}. Furthermore, just like their rate-distortion functions (recall~\eqref{eqn:rel_UZ}), the dispersions of the Gauss-Markov source $U$ in~\eqref{def:GaussMarkov} and its innovation process $Z$ are comparable:
\begin{corollary}
Let $V_{U}(d)$ and $V_{Z}(d)$ be the dispersions of the Gauss-Markov source~\eqref{def:GaussMarkov} and the memoryless Gaussian source $\lrbb{Z_i}_{i = 1}^{\infty}$, respectively, then 
\begin{align}
\begin{cases}
V_{U}(d) = V_{Z}(d), & 0< d \leq d_c, \\
V_{U}(d) <  V_{Z}(d), & d_c< d< \sigma^2. \\
\end{cases}
\label{comparison:dispersions}
\end{align}
\label{coro:comparison}
\end{corollary}

\begin{proof}
From~\cite[Th. 2]{ingber2011dispersion} and~\cite[Th. 40]{kostina2012fixed}, we know that the dispersion of the memoryless Gaussian source is
\begin{align}
V_Z(d) = \frac{1}{2}, \quad \forall d\in (0, \sigma^2), \label{dispersion:GaussianIID}
\end{align}  
which is also shown in Fig.~\ref{fig:dispersion}. By the definition of $d_c$ in~\eqref{critical_distortion} and the discussion around~\eqref{critical_distortion} and~\eqref{d_max}, we see that~\eqref{eqn:wf_dis} satisfies
\begin{align}
\min\left[1, \left(\frac{\sigma^2}{\theta g(w)}\right)^2\right] ~
\begin{cases} 
= 1, & \text{ if }d\in (0, d_c], \\
< 1, & \text{ if }d\in (d_c, \dmax),  
\end{cases}
\end{align}
from which Corollary~\ref{coro:comparison} follows.
\end{proof}

Corollary~\ref{coro:comparison} parallels Gray's result~\eqref{eqn:rel_UZ}~\cite[Eq. (24)]{gray1970information} for the rate-distortion functions of $U$ and $Z$, and they together imply that for $d\in(0, d_c]$, the fundamental limits of lossy compression of the Gauss-Markov source and the i.i.d. Gaussian source $\left\{Z_i\right\}_{i=1}^{\infty}$ are the same, up to the second-order term. For $d\in (d_c, \sigma^2)$, the Gauss-Markov source is harder to compress in the limit of $n$ going to infinity since $\rdf_U(d)>\rdf_Z(d)$, but the Gauss-Markov source approaches its asymptotic fundamental limit faster since $ V_U(d) < V_Z(d)$. See the discussions following Theorem~\ref{thm:limiting_variance} below for an intuitive explanation.

The dispersions for $a = 0$ and $a = 0.5$ are plotted in Fig.~\ref{fig:dispersion}, where the dotted line (for $a = 0,\sigma^2 = 1$) recovers the dispersion result \eqref{dispersion:GaussianIID} in~\cite{kostina2012fixed,ingber2011dispersion} for the i.i.d. Gaussian sources $\left\{Z_i\right\}_{i=1}^{\infty}$, as expected. The solid line ( for $a = 0.5, \sigma^2 = 1$) coincides with the dotted line in the region $d\in (0,d_c]$, which means that the Gauss-Markov source has the same dispersion as its innovation process in the region of low $d$'s. For $d \in (d_c, \sigma^2)$, the dispersion of the Gauss-Markov source is smaller than that of its innovation process and decreases with $d$, as indicated by Corollary~\ref{coro:comparison}. 

Using the residue theorem from complex analysis, we also derive the coordinates of the two corner points $P_1$ and $P_2$ on the solid line (Appendix~\ref{app:cornerpoints}):
\begin{align}
P_1 = (d_c, 1/2), \quad P_2 = \left(\dmax, \frac{(1+a^2)(1-a)}{2(1+a)^3}\right).
\label{corner_points}
\end{align}
The vertical segment between $(\dmax,0)$ and $P_2$ corresponds to the case when the water level $\theta$ is above the spectrum peak~$\theta_{\max}$ in Fig.~\ref{fig:RWF}, and the dispersion $V_U(d)$ in~\eqref{eqn:wf_dis} becomes
\begin{align}
V_U(d) = \frac{1}{4\pi\theta^2}\int_{-\pi}^{\pi}  S(w)^2~dw, \label{eqn:wf_dis_above}
\end{align}
which continues decreasing as $\theta$ increases, even as the distortion $d$ remains as $\dmax$, as seen from~\eqref{eqn:para_d_inf} and~\eqref{eqn:wf_dis}.
\begin{figure*}[t]
\centering
\includegraphics[width=0.7\textwidth, height = 0.49\textwidth]{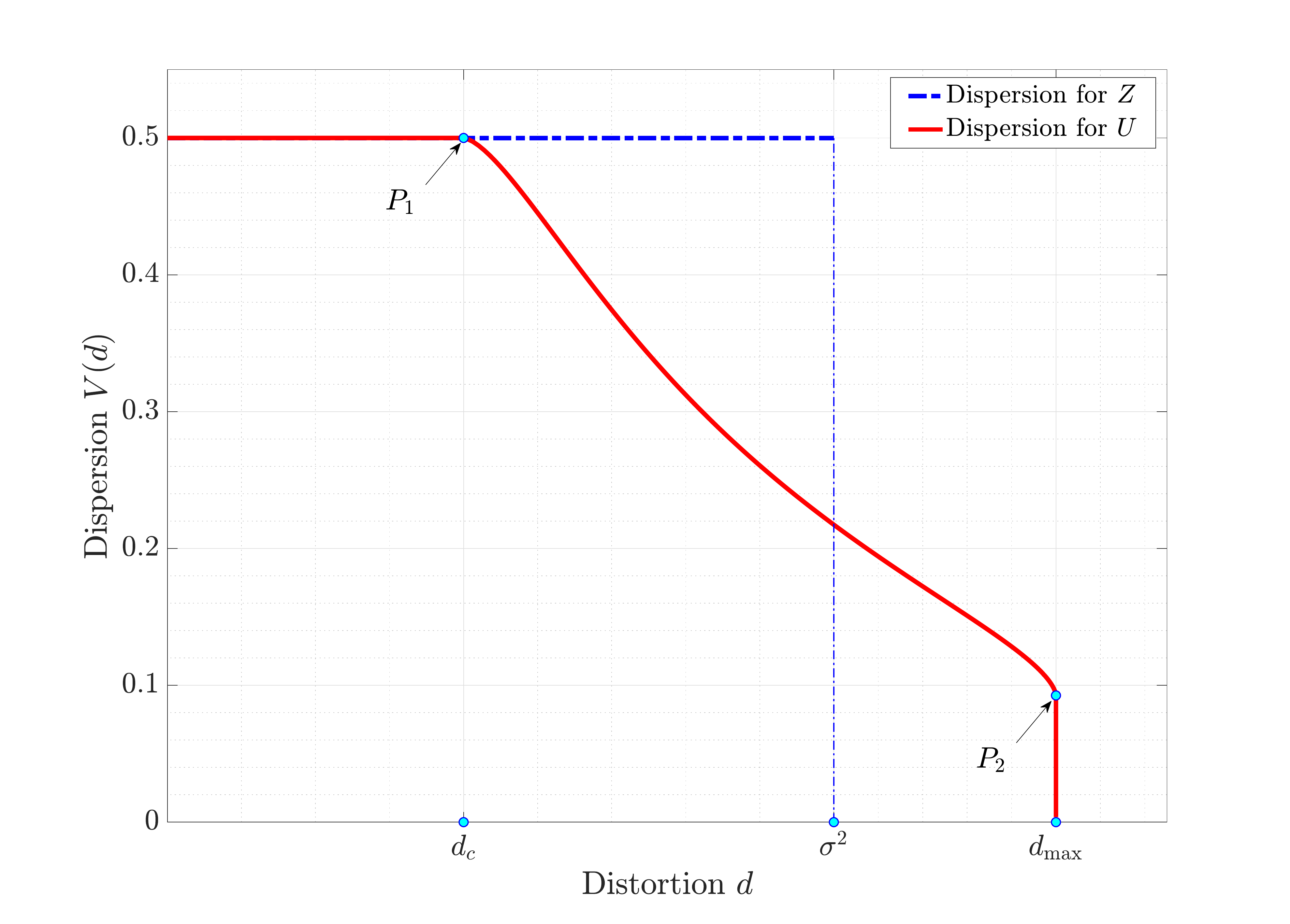}
\caption{Dispersion versus distortion: The Gauss-Markov source with $a=0$ and $\sigma^2 = 1$ degenerates to the i.i.d. Gaussian source, i.e., the innovation process $\left\{Z_i\right\}_{i=1}^{\infty}$. The dispersion of the Gauss-Markov source with $a = 0.5$ and $\sigma^2 = 1$ is given by the solid line. Two corner points on the solid line are labeled as $P_1$ (corresponding to $d_c$) and $P_2$ (corresponding to $\dmax$). } 
\label{fig:dispersion}
\end{figure*}

Theorem~\ref{thm:limiting_variance} below gives formulas for the $\mathsf{d}$-tilted information defined in~\eqref{def:d_tilted_def} and informational dispersion defined in~\eqref{def:info_dispersion}.
\begin{theorem}
For the Gauss-Markov source $U$ in~\eqref{def:GaussMarkov}, for any $d\in (0,\dmax)$ and $n\geq 1$, the $\mathsf{d}$-tilted information is given by
\begin{align}
\jmath_{\rvec{U}}(\rvec{u}, d)  = \jmath_{\rvec{X}}(\rvec{x}, d) & = \sum_{i = 1}^n \frac{\min(\theta_n,\sigma_i^2)}{2\theta_n}\left(\frac{x_i^2}{\sigma_i^2} - 1\right) +\notag \\
&~ \frac{1}{2}\log\frac{\max(\theta_n,\sigma_i^2)}{\theta_n}, \label{expression:dtilted}
\end{align}
and the informational dispersion satisfies 
\begin{align}
\mathbb{V}_U(d) = \frac{1}{4\pi}\int_{-\pi}^{\pi} \min\left[1, \left(\frac{S(w)}{\theta}\right)^2\right]~dw, \label{limiting_var}
\end{align}
where $\theta_n > 0$ is the water level matched to $d$ via the $n$-th order reverse waterfilling~\eqref{eqn:para_d}; $\theta > 0$ is the water level matched to the distortion $d$ via the limiting reverse waterfilling~\eqref{eqn:para_d_inf}; and the power spectrum $S(w)$ is defined in~\eqref{spectrum}.
\label{thm:limiting_variance}
\end{theorem}

\begin{proof}
Appendix~\ref{app:derivation_vd}.
\end{proof}

Theorem~\ref{thm:limiting_variance} computes the informational dispersion $\mathbb{V}_U(d)$ defined in~\eqref{def:info_dispersion}. The formula~\eqref{eqn:main_GA} in Theorem~\ref{thm:main_thm_intro} is an equivalent reformulation of the definition of the operational dispersion $V_U(d)$ in~\eqref{def:op_dispersion}, while~\eqref{eqn:wf_dis} together with Theorem~\ref{thm:limiting_variance} establish the equality $\mathbb{V}_U(d) = V_U(d)$. Theorem~\ref{thm:main_thm_intro} and Theorem~\ref{thm:limiting_variance} establish that for the Gauss-Markov source, the operational RDF and dispersion are given by~\eqref{limiting_expectation_relation} and~\eqref{def:info_dispersion}, respectively, providing a natural extension to the fact that in lossy compression of i.i.d. sources $\left\{X_i\right\}$, the mean $\mathbb{E}[\jmath_{X_1}(X_1,d)]$ and the variance $\var{\jmath_{X_1}(X_1,d)}$ of the single-letter $\mathsf{d}$-tilted information $\jmath_{X_1}(X_1,d)$ are equal to the RDF and the dispersion, respectively~\cite[Th. 12]{kostina2012fixed}.

Theorem~\ref{thm:limiting_variance} also provides intuition on our result in Corollary~\ref{coro:comparison} for $d>d_c$. Since $(\rvec{X},\rvec{Y}^\star)$ forms a RDF-achieving pair in $\rdf_{\rvec{X}}(n, d)$ (recall~\eqref{eqn:nthorderOp}), it is well-known~\cite[Th. 10.3.3]{cover2012elements} that $\rvec{Y}^\star$ has independent coordinates and $\forall~i\in [n]$, 
\begin{align}
Y_i^\star \sim \mathcal{N}(0, \max\LRB{\sigma_i^2 - \theta_n, 0}), \label{Ystar}
\end{align}
where $\theta_n >0 $ is the water level matched to the distortion $d$ in the $n$-th order reverse waterfilling over $\{\sigma_i^2\}_{i = 1}^n$ in~\eqref{eqn:para_d}. Since $d>d_c$, there are some $X_i$ in~\eqref{expression:dtilted} which are ``inactive", that is, $\sigma_i^2 < \theta_n$, which makes the variance of~\eqref{expression:dtilted} smaller. Geometrically, since $\rvec{X}$ concentrates inside an ellipsoid, we are covering such an ellipsoid by balls of radius $\sqrt{nd}$. The centers of these distortion $d$-balls lie on another lower dimensional ellipsoid. That lower dimensional ellipsoid is the one on which the random vector $\rvec{Y}^\star$ concentrates. For $d>d_c$, although centered at a lower dimensional ellipsoid (since $Y_i^\star \equiv 0$ for inactive $X_i$'s), these $d$-balls are large enough to also cover those ``inactive" dimensions.

\subsection{Technical tools}
\label{subsec:maintools}

\subsubsection{Eigenvalues of the covariance matrices}
\label{subsubsec:Teig}
Although decorrelation simplifies the problem by transforming a source with memory into a memoryless one, the real challenge is to study the evolution of the variances $\sigma_i^2$'s in~\eqref{eig:sigma_i}, as $n$ increases. For finite $n$, there is no closed-form expression for the eigenvalues of $\mathsf{\Sigma}_{\rvec{U}}$ for $a\in (0,1)$.\footnote{A closed-form expression for the eigenvalues of $\mathsf{\Sigma}_{\rvec{U}}$ is known only for $a = 1$~\cite[Eq. (2)]{berger1970information}.} Since the inverse of $\mathsf{\Sigma}_{\rvec{U}}$ is $\frac{1}{\sigma^2} \mathsf{A}^\top \mathsf{A}$, which is almost a Toeplitz matrix except the $(n, n)$-th entry, the limiting distribution of the eigenvalues of $\mathsf{\Sigma}_{\rvec{U}}$ can be deduced from the limiting distribution of eigenvalues of Toeplitz matrices~\cite[Eq. (19)]{gray1970information}.
\begin{theorem}[Reformulation of Gray{~\cite[Eq. (19)]{gray1970information}}]
Fix any $a\in [0,1)$. For any continuous function $F(t)$ over the interval 
\begin{align}
t\in\left[\theta_{\min}, \theta_{\max}\right],
\label{interval_t}
\end{align}
the eigenvalues $\sigma^2_i$'s of $\mathsf{\Sigma}_{\rvec{U}}$ satisfy 
\begin{align}
\lim_{n\rightarrow\infty}\frac{1}{n}\sum_{i = 1}^n F(\sigma_i^2) = \frac{1}{2\pi}\int_{-\pi}^{\pi} F\left(S(w)\right)~dw,
\label{eqn:limiting_equality}
\end{align}
where $S(w)$ is defined in~\eqref{spectrum}.
\label{thm:LimitingThm}
\end{theorem}

There are more general results in the form of Theorem~\ref{thm:LimitingThm}, known as Szeg{\"o}'s theorem, see~\cite[Chap. 5]{grenander2001toeplitz} for Toeplitz forms and~\cite[Cor. 2.3]{gray2006toeplitz} for asymptotically Toeplitz matrices. In the context of rate-distortion theory, applying Theorem~\ref{thm:LimitingThm} to~\eqref{eqn:para_r_n}-\eqref{eqn:para_d} leads to~\eqref{eqn:para_rate_inf}-\eqref{eqn:para_d_inf}.

Unfortunately, Theorem~\ref{thm:LimitingThm} is insufficient to obtain the fine asymptotics in our Theorem~\ref{thm:main_thm_intro}. To derive our finite blocklength results, we need to understand the rate of convergence in~\eqref{eqn:limiting_equality}. Towards that end, we develop a nonasymptotic refinement of Theorem~\ref{thm:LimitingThm}, presented next.
\begin{theorem}[Nonasymptotic eigenvalue distribution of $\mathsf{\Sigma}_{\rvec{U}}$]
Fix any $a\in [0,1)$. For any bounded, $L$-Lipschitz and non-decreasing function $F(t)$ over the interval in~\eqref{interval_t}, and for any $n\geq 1$, the eigenvalues $\sigma_i^2$'s of $\mathsf{\Sigma}_{\rvec{U}}$ satisfy
\begin{align}
\left |\frac{1}{n}\sum_{i = 1}^n F(\sigma_i^2) - \frac{1}{2\pi}\int_{-\pi}^{\pi} F\left(S(w)\right) dw\right| \leq \frac{C_L}{n}, 
\label{eqn:n_limiting_equality}
\end{align}
where $C_L > 0 $ is a constant that depends on the Lipschitz constant $L$ and the sup norm $\|F\|_{\infty}$ of $F$, and $S(w)$ is in~\eqref{spectrum}.
\label{thm:n_LimitingThm}
\end{theorem}
\begin{proof}
Theorem~\ref{thm:n_LimitingThm} follows from Lemma~\ref{lemma:eig_approx} below and elementary analyses on Riemann sums. See Appendix~\ref{app:pf_n_LimitingThm} for details.
\end{proof}

In the course of the proof of Theorem~\ref{thm:n_LimitingThm}, we obtain the following nonasymptotic bounds on each eigenvalue $\mu_i$ of $\mathsf{A}^\top\mathsf{A}$, which is of independent interest.
\begin{lemma}[Sharp approximation of the eigenvalues of $\mathsf{A}^\top\mathsf{A}$]
Fix any $a\in [0,1)$. For any $n\in\mathbb{N}$, let $0 < \mu_1\leq\mu_2\ldots\leq \mu_n$ be the eigenvalues of $\mathsf{A}^\top\mathsf{A}$, and let
\begin{align}
\xi_i \triangleq g\left(\frac{i\pi}{n+1}\right), \label{eig:W}
\end{align}
where $g$ is in \eqref{exp:g}. Then, we have 
\begin{align}
0\leq \xi_i - \mu_i \leq \frac{2a\pi}{n}, \quad\forall~i\in [n]. \label{bound:ximu}
\end{align}
\label{lemma:eig_approx}
\end{lemma}

\begin{proof}
The idea in proving Lemma~\ref{lemma:eig_approx} is that $\mathsf{A}^\top\mathsf{A}$ is almost a tridiagonal Toeplitz matrix, whose eigenvalues are given by~\eqref{eig:W}. The bound~\eqref{bound:ximu} is obtained via the Cauchy interlacing theorem and the Gershgorin circle theorem. See Appendix~\ref{app:EigenApprox} for details.
\end{proof}

\begin{remark}
In view of~\eqref{eig:W} and~\eqref{bound:ximu} in Lemma~\ref{lemma:eig_approx}, we have $\forall~n \in \mathbb{N}$ and $\forall~i\in [n],$
\begin{align}
(1-a)^2\leq \mu_i\leq (1+a)^2.
\label{constant_eig_bound}
\end{align}
The key difference between the asymptotically stationary case ($a\in [0,1)$) and the nonstationary case ($a\geq 1$) is that, in the later case, $\mu_1$ decreases to zero as $n$ increases to infinity, see~\cite[Lemma]{hashimoto1980rate} and~\cite[Eq. (2)]{berger1970information}. In the asymptotically stationary case, $\mu_1$ is bounded away from zero according to~\eqref{constant_eig_bound}.
\end{remark}

\subsubsection{An estimation problem}
\label{subsubsec:Teit}

Our achievability proof relies on the analysis of the following parameter estimation problem. Given a source sequence $\rvec{u} = (u_1, \ldots, u_n)^\top$, drawn from the model~\eqref{def:GaussMarkov} with unknown $a$, the maximum likelihood estimate (MLE) of the parameter $a$ is (Appendix~\ref{app:MLEderivation})
\begin{align}
\hat{a}(\rvec{u}) = \frac{\sum_{i = 1}^{n-1} u_{i} u_{i+1}}{\sum_{i = 1}^{n-1} u_i^2}.
\label{MLEformula}
\end{align}
We show that the estimation error of the MLE decays exponentially in $n$ for any $a\in [0,1)$. 
\begin{theorem}
\label{thm:ECG}
Fix $a\in [0,1)$. Let $\eta \in (0,1)$. Then, there exists a universal constant $c>0$ and two constants $c_1, c_2 >0$ ($c_1$ and $c_2$ only depend on $a$, see~\eqref{c1c2} in Appendix~\ref{proof:thm_error_concentration} below) such that for all $n$ large enough, the estimation error of the MLE satisfies 
\begin{align}
\mathbb{P}\left[|\hat{a}(\rvec{U}) - a| > \eta \right] \leq 2\exp\left[-c\min\left( c_1 \eta^2 n, c_2 \eta n\right)\right]. \label{mle1}
\end{align}
\end{theorem}

\begin{proof}
Appendix~\ref{proof:thm_ECG}.
\end{proof}

Finally, we present a strengthened version of Theorem~\ref{thm:ECG}, which is used in our achievability proof. Let $\alpha >0$ be a constant. Define $\eta_n$ as
\begin{align}
\eta_n \triangleq \sqrt{\frac{\alpha \log\log n}{n}}.\label{rhon}
\end{align}
\begin{theorem}
\label{thm:error_concentration}
Fix $a\in [0,1)$. Given a constant $\alpha>0$, let $\eta_n$ be in~\eqref{rhon}. Then, for all $n$ large enough, the estimation error of the MLE satisfies 
\begin{align}
\mathbb{P}\left[|\hat{a}(\rvec{U}) - a| > \eta_n \right] \leq \frac{2}{\left(\log n\right)^{\kappa\alpha}}, \label{mle2}
\end{align} 
where $\kappa$ is a constant given by 
\begin{align}
\kappa \triangleq \frac{c}{8(1-a^2)}, \label{kappa}
\end{align}
and $c>0$ is the constant in Theorem~\ref{thm:ECG}. 
\end{theorem}

\begin{proof}
Appendix~\ref{proof:thm_error_concentration}.
\end{proof}

See Section~\ref{subsec:ML_TS} for the construction of a typical set based on $\hat{a}(\rvec{u})$.
\section{Converse}
\label{sec:converse}

\begin{theorem}[Converse]
For the Gauss-Markov source~\eqref{def:GaussMarkov} with the constant $a\in [0,1)$, for any excess-distortion probability $\epsilon\in (0,1)$, and for any distortion threshold $d\in \LRB{0,\dmax}$, the minimum achievable source coding rate satisfies 
\begin{align}
R(n,d,\epsilon) \geq \rdf_{U}(d) + \sqrt{\frac{\mathbb{V}_U(d)}{n}}Q^{-1}(\epsilon) -\frac{\log n}{2n} + O\left(\frac{1}{n}\right), 
\label{eqn:converse_bd}
\end{align}
where $\rdf_{U}(d)$ is the rate-distortion function given in~\eqref{eqn:para_rate_inf}, and $\mathbb{V}_U(d)$ is the informational dispersion, defined in~\eqref{def:info_dispersion} and computed in~\eqref{limiting_var}.
\label{thm:converse}
\end{theorem}

We present two converse proofs in the following. The first one is a volumetric argument; while the second one relies on a general converse derived in~\cite[Th. 7]{kostina2012fixed} and a new concentration result on the $\mathsf{d}$-tilted information of the Gauss-Markov source. 

\subsection{A geometric proof}
\label{converse:geo_pf}
Geometrically, any $(n, M, d, \epsilon)$ code induces a covering of $(\mathbb{R}^n, P_{\rvec{U}}$): the union of $d$-balls centered at the codewords have probability mass at least $1-\epsilon$. Converting the underlying probability to $P_{\rvec{Z}}$ and using the symmetry of $\mathcal{N}(\rvec{0}, \sigma^2\mathsf{I})$, we obtain the following lower bound on the number of codewords $M$. The argument relies on $\det\mathsf{A} = 1$, where $\mathsf{A}$ is in~\eqref{def:A}.

\begin{theorem}
Given $\epsilon \in (0,1)$ and $d\in (0,\dmax)$, the size of any $(n,M,d,\epsilon)$ code for the Gauss-Markov source~\eqref{def:GaussMarkov} must satisfy 
\begin{align}
M \geq \left(\frac{r(n,\epsilon)}{d}\right)^{n/2},
\label{eqn:geo_bd}
\end{align}
where $r(n,\epsilon)$ is such that 
\begin{align}
\mathbb{P}(G < n\cdot r(n,\epsilon)/\sigma^2) = 1-\epsilon, \label{r_n_epsilon}
\end{align}
and $G$ is a random variable distributed according to the $\chi^2$-distribution with $n$ degrees of freedom.
\label{thm:geometric}
\end{theorem}

\begin{proof}[Proof of Theorem~\ref{thm:geometric}]
Appendix~\ref{app:proofGeoBd}.
\end{proof}

\begin{remark}
Theorem~\ref{thm:geometric}, which applies to the Gauss-Markov source, parallels~\cite[Th. 36]{kostina2012fixed}, which applies to the i.i.d. Gaussian source. Both proofs rely on the volumetric method, though the proof of~Theorem~\ref{thm:geometric} requires additional arguments related to linear transformations of the underlying space. Theorem~\ref{thm:geometric} yields the optimal second-order coding rate for the Gauss-Markov source only in the low distortion regime (as we will see in the proof of Theorem~\ref{thm:converse} below), while an analysis of~\cite[Th. 36]{kostina2012fixed} gives the optimal second-order coding rate for the i.i.d. Gaussian source of any distortion~\cite[Th. 40]{kostina2012fixed}.
\end{remark}

Equipped with Theorem~\ref{thm:geometric}, we are ready to prove the converse in Theorem~\ref{thm:converse} for $d\in (0,d_c]$.

\begin{proof}[Proof of Theorem~\ref{thm:converse} below the critical distortion]
Applying the Berry-Esseen Theorem in Appendix~\ref{app:classicalthms} to~\eqref{r_n_epsilon} yields
\begin{align}
r(n, \epsilon) \geq \sigma^2 \left[ 1 + \sqrt{\frac{2}{n}}Q^{-1}\left(\epsilon + \frac{\bec}{\sqrt{n}}\right) \right]. \label{lowerboundrnepsilon}
\end{align}
Plugging~\eqref{lowerboundrnepsilon} into~\eqref{eqn:geo_bd} and taking logarithms, we obtain
\begin{align}
R(n,d, \epsilon) \geq \frac{1}{2}\log\frac{\sigma^2}{d} + \sqrt{\frac{1}{2n}}Q^{-1}(\epsilon) + O\LRB{\frac{1}{n}}, 
\label{converse:low_dis}
\end{align}
where we use the Taylor expansions of $\log (1 + x)$ and the inverse $Q$-function. The converse bound~\eqref{converse:low_dis} holds for any $\epsilon\in (0, 1)$ and $d\in (0, \dmax)$. By~\eqref{eqn:rel_UZ} and~\eqref{comparison:dispersions}, we see that~\eqref{converse:low_dis} is the same as~\eqref{eqn:converse_bd} for $d\in (0,d_c]$, up to the second-order term. In addition,~\eqref{converse:low_dis} is slightly stronger than~\eqref{eqn:converse_bd} in the third-order term. For $d\in (d_c, \dmax)$,~\eqref{converse:low_dis} is not tight, even in the first order since $\rdf_{U}(d)>\frac{1}{2}\log\frac{\sigma^2}{d}$ for $d\in(d_c, \dmax)$, by~\eqref{eqn:rel_UZ} and~\eqref{RDF:iidGaussian}.
\end{proof}

\begin{remark}
The converse~\eqref{converse:low_dis} holds for the general Gaussian AR processes defined in~\eqref{eqn:GeneralGaussianAR}. The proof stays the same, except that the matrix $\mathsf{A}$ in~\eqref{def:A} is replaced by
\begin{align}
\mathsf{A} \triangleq \begin{bmatrix}
1  & 0  &  0  & \ldots & 0 \\
-a_1 & 1  & 0   & \ldots & 0\\
-a_2  & -a_1 & 1   & \ldots & 0\\
\vdots & \ddots & \ddots & \ddots & \vdots\\
-a_{n-1} & \ldots & -a_2 & -a_1 & 1
\end{bmatrix}.
\label{def:A_general}
\end{align}
\end{remark}

\subsection{Converse proof}
\label{converse:gen_pf}
The second proof is based on a general converse by Kostina and Verd{\'u}~\cite{kostina2012fixed}, restated here for convenience, and a concentration result which bounds the difference between $\jmath_{\rvec{X}}(\rvec{X}, d)$ and its approximation $\jmath_{\rvec{X}}(\rvec{X}, d_n)$, for $d_n$ defined in~\eqref{def:dn} below.
\begin{theorem}[{\cite[Th. 7]{kostina2012fixed}}]
Fix $d\in\LRB{0,\dmax}$. Any $(n, M, d, \epsilon)$ code must satisfy 
\begin{align}
\epsilon \geq \sup_{\gamma\geq 0} \mathbb{P}\left[\jmath_{\rvec{X}}(\rvec{X},d)\geq \log M +\gamma\right] - \exp(-\gamma).
\label{converse:general}
\end{align}
\label{thm:general_converse}
\end{theorem}

The converse bound in Theorem~\ref{thm:general_converse} above provides a lower bound on $\epsilon$ for any $(n, M, d, \epsilon)$ code using the $\mathsf{d}$-tilted information, and is used to derive a converse result on the dispersion of the stationary memoryless sources in~\cite[Eq. (103)-(106)]{kostina2012fixed}. The key step in the proof of~\cite[Eq. (103)-(106)]{kostina2012fixed} is to write the $\mathsf{d}$-tilted information as a sum of $n$ i.i.d. random variables, to which the Berry-Esseen Theorem is applied.  

For the Gauss-Markov source $\rvec{X}$, using~\eqref{def:d_tilted_def},~\eqref{eqn:para_d} and~\eqref{Ystar}, we can write the $\mathsf{d}$-tilted information $\jmath_{\rvec{X}}(\rvec{X}, d)$ as a sum of $n$ independent (but not identical) random variables:
\begin{align}
\jmath_{\rvec{X}}(\rvec{X}, d) = \sum_{i = 1}^n \jmath_{X_i}(X_i, \min(\theta_n, \sigma_i^2)), \label{nsum}
\end{align}
where $\theta_n$ is given in~\eqref{eqn:para_d}. Indeed, \eqref{nsum} is further simplified to~\eqref{expression:dtilted} in the proof of~Theorem~\ref{thm:limiting_variance}. However, it is hard to conduct nonasymptotic analysis using~\eqref{nsum} since understanding the evolution of both $\theta_n$ and $\sigma_i^2$'s as $n$ grows in~\eqref{nsum} is challenging. Therefore, we approximate $\jmath_{\rvec{X}}(\rvec{X}, d)$ using 
\begin{align}
\jmath_{\rvec{X}}(\rvec{X}, d_n) = \sum_{i = 1}^n \jmath_{X_i}(X_i, \min(\theta, \sigma_i^2)), \label{ssum}
\end{align}
where
\begin{align}
d_n \triangleq \frac{1}{n} \sum_{i  = 1}^n \min(\theta, \sigma_i^2),\label{def:dn}
\end{align}
and $\theta$ is the water level matched to $d$ via the limiting reverse waterfilling~\eqref{eqn:para_d_inf}. Then, $\theta$ does not dependent on $n$ in~\eqref{ssum}. Since our Theorem~\ref{thm:n_LimitingThm} and Lemma~\ref{lemma:eig_approx} in Section~\ref{subsubsec:Teig} capture the evolution of $\sigma_i^2$'s as $n$ grows,~\eqref{ssum} is easier to analyze than~\eqref{nsum} in the nonasymptotic regime. Throughout the paper, the relations among a given distortion $d$, the water levels $\theta$, $\theta_n$, and the distortion $d_n$ defined in~\eqref{def:dn}, are $
\theta_n \stackrel{\eqref{eqn:para_d}}{\longleftrightarrow} d \stackrel{\eqref{eqn:para_d_inf}}{\longleftrightarrow} \theta \stackrel{\eqref{eqn:para_d}}{\longleftrightarrow} d_n.$ Note that there is no direct reverse waterfilling relation between $d_n$ in~\eqref{def:dn} and $\theta_n$ in~\eqref{eqn:para_d}. As shown by our concentration result Theorem~\ref{thm:tilted_info_concentrate} in the following, the approximation $\jmath_{\rvec{X}}(\rvec{X}, d_n)$ stays within a constant from $\jmath_{\rvec{X}}(\rvec{X}, d)$ with probability at least $1 - O\left(\frac{1}{n}\right)$.
 
\begin{theorem}[Approximation of the $\mathsf{d}$-tilted information]
For any $d\in \LRB{0,\dmax}$, let $\theta>0$ be the water level matched to $d$ via the limiting reverse waterfilling~\eqref{eqn:para_d_inf}. Suppose we have a sequence of distortion levels $d_n\in (0,\dmax)$ with the property that there exists a constant $h_1 > 0$ such that for all $n$ large enough, 
\begin{align}
\lrabs{d-d_n}\leq \frac{h_1}{n}.
\label{assumption:diff_concentrate}
\end{align}
Then, there exists a constant $\tilde{c}\in (0,1)$ such that for any $u>\frac{2h_1}{\tilde{c}\theta}$ and all $n$ large enough, we have 
\begin{align}
\prob{\lrabs{\jmath_{\rvec{X}}\LRB{\rvec{X},d} - \jmath_{\rvec{X}}\LRB{\rvec{X},d_n}}\leq u} \geq 1 - \frac{1}{n\left(\frac{\tilde{c}\theta u}{2h_1} - 1\right)^2}.
\end{align}
\label{thm:tilted_info_concentrate}
\end{theorem}
\begin{proof}[Proof of Theorem~\ref{thm:tilted_info_concentrate}]
Appendix~\ref{app:tilted_info_concentrate}.
\end{proof}

In the rest of this section, we present the detailed proof of Theorem~\ref{thm:converse} for any $d\in (0, \dmax)$. The $\mathsf{d}$-tilted information $\jmath_{\rvec{X}}\LRB{\rvec{X},d}$ is first approximated by $\jmath_{\rvec{X}}\LRB{\rvec{X},d_n}$ defined in~\eqref{ssum}, which is a sum of independent random variables whose expectations and variances approximate the rate-distortion function $\rdf_{U}(d)$ and the informational dispersion $\mathbb{V}_U(d)$, respectively. Combining these approximation bounds and Theorem~\ref{thm:general_converse}, we obtain the converse in~\eqref{eqn:converse_bd}. The details follow.

\begin{proof}[Proof of Theorem~\ref{thm:converse}]
\label{proof:converse_2nd_pf}
Fix $d\in\LRB{0,\dmax}$. Let $\theta>0$ be the water level matched to $d$ via the limiting reverse waterfilling~\eqref{eqn:para_d_inf}. Notice that $d_n$, defined in~\eqref{def:dn}, is the distortion matched to the water level $\theta$ via the $n$-th order reverse waterfilling~\eqref{eqn:para_d} over $\sigma_i^2$'s. Comparing~\eqref{eqn:para_d} and~\eqref{def:dn}, and applying Theorem~\ref{thm:n_LimitingThm} to the function $t\mapsto\min\LRB{\theta, t}$, we deduce that there exists a constant $C_d>0$ such that for any $n\geq 1$, 
\begin{align}
\lrabs{d - d_n}\leq \frac{C_d}{n}. \label{d_d_n}
\end{align}
Let $\bar{\rvec{Y}}^\star$ be the $n$-dimensional Gaussian random vector such that $(\rvec{X}, \bar{\rvec{Y}}^\star )$ forms a RDF-achieving pair in $\rdf_{\rvec{X}}(n, d_n)$ defined in~\eqref{eqn:nthorderOp}. Note that $\bar{\rvec{Y}}^\star$ defined here is indeed different from $\rvec{Y}^\star$ in~\eqref{Ystar}, where $(\rvec{X}, \rvec{Y}^\star)$ forms a RDF-achieving pair in $\rdf_{\rvec{X}}(n, d)$. It is well-known~\cite[Th. 10.3.3]{cover2012elements} that $\bar{\rvec{Y}}^\star$ has independent coordinates and similar to~\eqref{Ystar}, 
\begin{align}
\bar{Y}_i^\star \sim \mathcal{N}(0, \max\LRB{\sigma_i^2 - \theta, 0}).
\end{align}
By the independence of $\bar{Y}_i^\star$'s,~\eqref{rel:dtiltedandtilted} and~\eqref{def:d_tilted_def}, we have 
\begin{align}
\jmath_{\rvec{X}}(\rvec{X}, d_n) = \sum_{i = 1}^n \Lambda_{\bar{Y}_i^\star}(X_i, \lambda^\star, \min\LRB{\theta, \sigma_i^2}), \label{sumofnrv}
\end{align}
where 
\begin{align}
\lambda^\star = -\rdf'_{\rvec{X}}(n, d_n) = -\rdf'({\rvec{X}}, \bar{\rvec{Y}}^\star, d_n).
\end{align}
Denote by $\mathbb{E}_i$ and $\mathbb{V}_i $ the means and the variances of $\Lambda_{\bar{Y}_i^\star}(X_i, \lambda^\star, \min\LRB{\theta, \sigma_i^2})$ (the summands in~\eqref{sumofnrv}). By the same computations leading to~\eqref{tilted_mean} and~\eqref{tilted_variance} in Appendix~\ref{app:derivation_vd}, we have
\begin{align}
\mathbb{E}_i &= \max\left(0,~\frac{1}{2}\log\frac{\sigma_i^2}{\theta}\right),\label{conversepf:mean} \\
\mathbb{V}_i &= \min\left(\frac{1}{2},~\frac{\sigma_i^4}{2\theta^2}\right). \label{conversepf:var}
\end{align} 

We now derive the approximation of $\rdf_{U}(d)$ and $\mathbb{V}_U(d)$ using the means $\mathbb{E}_i$'s and the variances $\mathbb{V}_i$'s, respectively. Applying Theorem~\ref{thm:n_LimitingThm} to the function $t\mapsto \max\left(0,\frac{1}{2}\log\frac{t}{\theta}\right)$ in~\eqref{limiting_expectation_relation} and~\eqref{conversepf:mean}, and to the function $t\mapsto \min\left(\frac{1}{2}, \frac{t^2}{2\theta^2}\right)$ in~\eqref{limiting_var} and~\eqref{conversepf:var}, we conclude that there exist two constants $c_r, c_v > 0$ (depending on $d$ only) such that
\begin{align}
\lrabs{n\rdf_{U}(d) - \sum_{i = 1}^n \mathbb{E}_i }&\leq  c_r,
\label{approx:rate}\\
\lrabs{\sqrt{ n\mathbb{V}_{U}(d) } - \sqrt{\sum_{i = 1}^n \mathbb{V}_i} } &\leq  c_v.
\label{approx:dis_2}
\end{align} 

Next, we consider the sequence of distortion levels $\left\{d_n\right\}_{n = 1}^{\infty}$, which satisfies the condition~\eqref{assumption:diff_concentrate} due to~\eqref{d_d_n}. Define the event 
\begin{align}
\mathcal{E}\triangleq \lrbb{ \jmath_{\rvec{X}}\LRB{\rvec{X}, d}\geq \jmath_{\rvec{X}}\LRB{\rvec{X}, d_n} - \frac{4C_d}{\tilde{c}\theta}}, 
\label{E_t}
\end{align}
where $\tilde{c} \in (0,1)$ is the constant in Theorem~\ref{thm:tilted_info_concentrate} and $C_d >0 $ is the constant in~\eqref{d_d_n}. Theorem~\ref{thm:tilted_info_concentrate} implies that  
\begin{align}
\prob{\mathcal{E}} \geq 1 - \frac{1}{n}.
\label{large:E_t}
\end{align}

Letting $\gamma = \frac{1}{2}\log n$ in Theorem~\ref{thm:general_converse}, we see that if an $(n, M, d, \epsilon')$-excess-distortion code exists, then
\begin{align}
\epsilon' &\geq \prob{\jmath_{\rvec{X}}(\rvec{X},d)\geq \log M + \frac{\log n}{2}} - \frac{1}{\sqrt{n}}\label{converse_gen_pf:step1} \\
&\geq \prob{\jmath_{\rvec{X}}(\rvec{X},d)\geq \log M + \frac{\log n}{2} | \mathcal{E}}\prob{\mathcal{E}}- \frac{1}{\sqrt{n}} \label{converse_gen_pf:step2}\\
&\geq \LRB{1 - \frac{1}{n}}\prob{\jmath_{\rvec{X}}(\rvec{X},d_n)\geq \log M + \frac{\log n}{2} + \frac{4C_d}{\tilde{c}\theta}}- \frac{1}{\sqrt{n}}, \label{converse_gen_pf:step3}
\end{align}
where~\eqref{converse_gen_pf:step3} is due to~\eqref{large:E_t} and~\eqref{E_t}. For any fixed $\epsilon\in (0,1)$, define $\epsilon_n$ as  
\begin{align}
\epsilon_n \triangleq \epsilon + \exp(-\gamma) + \frac{\bec}{\sqrt{n}} + \frac{1}{n}, 
\label{choose:epsilonn}
\end{align}
where $\bec$ is the constant in the Berry-Esseen Theorem in Appendix~\ref{app:classicalthms}. Then, we have $\epsilon_n \in (0,1)$ for all $n$ large enough. We choose $M$ as 
\begin{align}
\log M &\triangleq n\rdf_{U}(d) + \sqrt{n\mathbb{V}_U(d)} Q^{-1}(\epsilon_n) \notag \\ &\quad\quad - \gamma - c_r - c_v\lrabs{Q^{-1}(\epsilon_n)} - \frac{4C_d}{\tilde{c}\theta}.
\label{choice:M} 
\end{align}
From~\eqref{approx:rate},~\eqref{approx:dis_2} and~\eqref{choice:M}, we have 
\begin{align}
\log M & \leq \sum_{i = 1}^n \mathbb{E}_i + Q^{-1}(\epsilon_n)\sqrt{\sum_{i = 1}^n \mathbb{V}_i} -\gamma - \frac{4C_d}{\tilde{c}\theta}.
\label{bound:M}
\end{align}
Continuing the inequality in~(\ref{converse_gen_pf:step3}), we have 
\begin{align}
\epsilon' &\geq \LRB{1 - \frac{1}{n}} \mathbb{P}\Bigg [ \sum_{i = 1}^n \Lambda_{\bar{Y}_i^\star}(X_i, \lambda^\star, \min\LRB{\theta, \sigma_i^2}) \geq \notag \\
&\quad\quad\quad \sum_{i = 1}^n \mathbb{E}_i + Q^{-1}(\epsilon_n)\sqrt{\sum_{i = 1}^n\mathbb{V}_i} \Bigg ] - \frac{1}{\sqrt{n}} \label{converse_gen_pf:step4}\\
&\geq \LRB{1 - \frac{1}{n}}\LRB{\epsilon_n - \frac{\bec}{\sqrt{n}}} - \frac{1}{\sqrt{n}}\label{converse_gen_pf:step5} \\
&\geq \epsilon_n - \frac{1}{n} - \frac{1+\bec}{\sqrt{n}} \\
& = \epsilon,\label{converse_gen_pf:step6}
\end{align}
where~\eqref{converse_gen_pf:step4} is by~\eqref{sumofnrv} and the bound~\eqref{bound:M};~\eqref{converse_gen_pf:step5} is by the Berry-Esseen Theorem in Appendix~\ref{app:classicalthms}; and~\eqref{converse_gen_pf:step6} is by the choice of $\epsilon_n$ in~\eqref{choose:epsilonn}. Consequently, for all $n$ large enough, any $(n, M, d, \epsilon')$-excess-distortion code must satisfy $\epsilon'\geq \epsilon$, so we must have 
\begin{align}
R(n,d,\epsilon) \geq \frac{\log M}{n}.\label{final_bound}
\end{align}
Plugging~\eqref{choice:M} into~\eqref{final_bound} and applying the Taylor expansion to $Q^{-1}(\epsilon_n)$ yields~\eqref{eqn:converse_bd}.
\end{proof}
\section{Achievability}
\label{sec:achievability}
\begin{theorem}[Achievability] 
Fix any $\alpha > 0$. Consider the Gauss-Markov source defined in~\eqref{def:GaussMarkov}. For any excess-distortion probability $\epsilon\in (0,1)$, and any distortion threshold $d\in (0,\dmax)$, the minimum achievable source coding rate is bounded as
\begin{align}
R(n,d,\epsilon) \leq \rdf_{U}(d)+ \sqrt{\frac{\mathbb{V}_U(d)}{n}}Q^{-1}(\epsilon) + O\left(\frac{1}{(\log n)^{\kappa\alpha}\sqrt{n}}\right), 
\label{eqn:achievability}
\end{align}
where $\rdf_{U}(d)$ is the rate-distortion function given in~\eqref{eqn:para_rate_inf}; $\mathbb{V}_U(d)$ is the informational dispersion, defined in~\eqref{def:info_dispersion} and computed in~\eqref{limiting_var}; and $\kappa >0$ is the constant in~\eqref{kappa}.
\label{thm:achievability}
\end{theorem}

This section presents the proof of Theorem~\ref{thm:achievability}. We first discuss how bounds on the covering number of $n$-dimensional ellipsoids can be converted into achievability results. We then proceed to present our proof of Theorem~\ref{thm:achievability}, which relies on random coding~\cite[Cor. 11]{kostina2012fixed}, and a lower bound on the probability of a distortion $d$-ball using the $\mathsf{d}$-tilted information $\jmath_{\rvec{X}}(\rvec{X}, d)$.

\subsection{Connections to covering number}
Dumer \textit{et al.}~\cite{dumer2004coverings} considered the problem of covering an ellipsoid using the minimum number of balls in $\mathbb{R}^n$, and derived lower and upper bounds on that number. Although any upper bound on the covering number implies an upper bound on $R(n,d,\epsilon)$, the upper bound on covering number in~\cite{dumer2004coverings} is not tight enough to yield the achievability direction of the Gaussian approximation~\eqref{eqn:main_GA}. We proceed to explain how to obtain a bound on $R(n,d,\epsilon)$ from the results in~\cite{dumer2004coverings}. An ellipsoid $E_{\rvec{r}}^n$ is defined by 
\begin{align}
E_{\rvec{r}}^n \triangleq \left\{\rvec{x}\in\mathbb{R}^n: \sum_{i = 1}^n\frac{x_i^2}{r_i^2}  \leq 1\right\}, \label{def:ellipsod}
\end{align}
where $\rvec{r} = (r_1,\ldots,r_n)$, and $r_i >0$ is one half of the length of the $i$-axis of $E_{\rvec{r}}^n$. We say that a subset $\mathcal{M}_d\subset\mathbb{R}^n$ is a $d$-covering\footnote{In~\cite{dumer2004coverings}, the term $\varepsilon$-covering was used instead of $d$-covering used here. They are related by $\varepsilon = \sqrt{nd}$.} of the ellipsoid $E_{\rvec{r}}^n$ if 
\begin{align}
E_{\rvec{r}}^n \subseteq \bigcup_{\rvec{y}\in\mathcal{M}_d } \mathcal{B}(\rvec{y}, d), 
\end{align}
where $\mathcal{B}(\rvec{y}, d)$ is the $d$-ball centered at $\rvec{y}$, defined in~\eqref{def:d_ball}. The covering number $\mathscr{N}(n,d)$ of an ellipsoid $E_{\rvec{r}}^n$ is defined as the size of its minimal $d$-covering. The $d$-entropy $H_d(E_{\rvec{r}}^n)$ is the logarithm of the covering number
\begin{align}
H_d(E_{\rvec{r}}^n) \triangleq \log \mathscr{N}(n,d).
\end{align}
The result in~\cite[Th. 2]{dumer2004coverings} states that 
\begin{align}
H_d(E_{\rvec{r}}^n)  = K_d + o\left(K_d\right), 
\label{bound:d_entropy}
\end{align}
where 
\begin{align}
K_d \triangleq \sum_{i: r_i^2 > nd} \frac{1}{2}\log \frac{r_i^2}{nd}.
\label{def:K_d}
\end{align}
Despite the similarity between~\eqref{def:K_d} and the reverse waterfilling~\eqref{eqn:para_r_n}, the result in~\eqref{bound:d_entropy} is not strong enough to recover even the asymptotic rate-distortion tradeoff~\eqref{eqn:para_r_n} unless $d\leq d_c$. 

In our problem, let $\rvec{X}$ be the decorrelation of $\rvec{U}$ in~\eqref{def:decorrelation}, then $X_1,\ldots,X_n$ are independent zero-mean Gaussian distributed with variances being $\sigma_i^2$ defined in~\eqref{eig:sigma_i}. The random vector $\rvec{X}$ concentrates around an ellipsoid with probability mass at least $1-\epsilon$. Applying the Berry-Esseen theorem to express $r_i$'s in~\eqref{def:K_d}, we deduce that for any $\epsilon\in (0,0.5)$ and $d\in (0,d_c]$, 
\begin{align}
R(n,d,\epsilon) \leq \frac{1}{2}\log \frac{\sigma^2}{d} + \frac{Q^{-1}(\epsilon)}{\sqrt{2n}} + o(1), \label{eqn:coveringBD}
\end{align}
where the extra $o(1)$ term comes from the $o(K_d)$ term in~\eqref{bound:d_entropy}. Due to that $o(1)$ term, the bound~\eqref{eqn:coveringBD} is first-order optimal, but not second-order optimal. Strenghthening~\eqref{bound:d_entropy} to $H_d(E_{\rvec{r}}^n)  = K_d + o\left(\sqrt{K_d}\right)$ would allow one to replace the $o(1)$ term in~\eqref{eqn:coveringBD} by $o\left(\frac{1}{\sqrt{n}}\right)$, yielding the $\leq$ (achievability) direction of the Gaussian approximation~\eqref{eqn:main_GA} in the regime of $d\in (0, d_c]$. We do not pursue this approach here. Instead, we prove~\eqref{eqn:main_GA} via the tilted information.

\subsection{Outline of the achievability proof}
\label{subsec:ML_TS}
We describe the main ideas in our achievability proof and present the details in next subsection. Our proof is inspired by the work of Kostina and Verd{\'u}~\cite[Th. 12]{kostina2012fixed}, where the same problem was addressed for the stationary memoryless sources. However, the proof there cannot be directly applied to the Gauss-Markov source. The random coding bound, stated next, provides an upper bound on the excess-distortion probability $\epsilon$ using the probability of the distortion $d$-balls.

\begin{lemma}[Random coding bound]
Let $\rvec{X}$ be the decorrelation of $\rvec{U}$ in~\eqref{def:decorrelation}. There exists an $(n,M,d,\epsilon)$ code with 
\begin{align}
\epsilon \leq \inf_{P_{\rvec{Y}}} \mathbb{E}_{\rvec{X}}\left[e^{-MP_{\rvec{Y}}(\mathcal{B}(\rvec{X},d))}\right],
\label{bound:rcb}
\end{align}
where the infimum is over all pdf's $P_{\rvec{Y}}$ on $\mathbb{R}^n$ with $\rvec{Y}$ independent of $\rvec{X}$.
\label{lemma:random} 
\end{lemma}

\begin{proof}
A direct application of {\cite[Cor. 11]{kostina2012fixed}} to $\rvec{X}$.
\end{proof}

The next lemma provides a lower bound on the probability of the distortion $d$-balls using the $\mathsf{d}$-tilted information $\jmath_{\rvec{X}}(\rvec{X}, d)$. 
\begin{lemma}[Lossy AEP for the Gauss-Markov sources]
Fix any $\alpha > 0$ and let $\eta_n$ be in~\eqref{rhon} in Section~\ref{subsubsec:Teit} above. For any $d\in (0,\dmax)$ and $\epsilon\in (0,1)$, there exists a constant $K>0$ such that for all $n$ large enough,
\begin{align}
& \mathbb{P}\left[\log \frac{1}{P_{\rvec{Y}^\star}\left(\mathcal{B}(\rvec{X}, d)\right)} \leq \jmath_{\rvec{X}}(\rvec{X},d) + \beta_1\log^q n +\beta_2 \right] \notag\\
& \geq 1 - \frac{K}{(\log n)^{\kappa\alpha}},
 \label{eqn:lemma_2}
\end{align}
where $\rvec{X}$ is the decorrelation of $\rvec{U}$ in~\eqref{def:decorrelation}; $(\rvec{X},\rvec{Y}^\star)$ forms a RDF-achieving pair in $\rdf_{\rvec{X}}(n, d)$, and $q>1, \beta_1 > 0,\beta_2$ are constants, see~\eqref{choice:beta_1} and~\eqref{choice:beta_2} in Appendix~\ref{sec:main_lemma} below. The constant $\kappa > 0$ is in~\eqref{kappa} in Section~\ref{subsubsec:Teit} above.
\label{lemma:lemma_2}
\end{lemma}

\begin{proof}
Appendix~\ref{sec:main_lemma}.
\end{proof}

Together with $\log\frac{1}{P_{Y^\star}(\mathcal{B}(x, d))}\geq \jmath_X(x, d)$ in~\cite[Eq. (26)]{kostina2012fixed}, obtained by applying Markov's inequality to~\eqref{def:d_tilted_def}, Lemma~\ref{lemma:lemma_2} establishes the link between the probability of distortion $d$-ball and the $\mathsf{d}$-tilted information: $\log \frac{1}{P_{\rvec{Y}^\star}\left(\mathcal{B}(\rvec{X}, d)\right)} \approx \jmath_{\rvec{X}}(\rvec{X},d)$ for the Gauss-Markov source. Results of this kind were referred to as lossy asymptotic equipartition property (AEP) in~\cite[Sec. I.B]{dembo2002source}. 

Lemma~\ref{lemma:lemma_2} is the key lemma in our achievability proof for the Gauss-Markov sources. The proof of Lemma~\ref{lemma:lemma_2} is one of the main technical contributions of this paper. An analog of Lemma~\ref{lemma:lemma_2} for the stationary memoryless sources~\cite[Lem. 2]{kostina2012fixed} has been used to prove the non-asymptotic achievability result~\cite[achievability proof of Th. 12]{kostina2012fixed}. Showing a lower bound on the probability of distortion $d$-balls in terms of $\jmath_{\rvec{X}}(\rvec{X}, d)$, that is, in the form of~\eqref{eqn:lemma_2}, is technical even for i.i.d. sources. To derive such a bound for the Gauss-Markov sources, we rely on fundamentally new ideas, including the maximum likelihood estimator $\hat{a}(\rvec{u})$ defined in~\eqref{MLEformula} and analyzed in Theorem~\ref{thm:error_concentration} in Section~\ref{subsubsec:Teit} above. We proceed to discuss the estimator and its role in the proof of Lemma~\ref{lemma:lemma_2} next.

A major step in proving~\cite[Lem. 2]{kostina2012fixed} for the i.i.d. source $\{X_i\}$ with $X_i \sim P_X$ involves the empirical probability distribution $P_{\hat{X}}$: given a source sequence $\rvec{x}$, $P_{\hat{X}}(x) \triangleq \frac{1}{n}\sum_{i = 1}^n \mathbbm{1}\left\{x_i = x\right\}$. The product of the empirical distributions $P_{\hat{\rvec{X}}} \triangleq P_{\hat{X}} \times \ldots P_{\hat{X}}$ was used in the proof of~\cite[Lem. 2]{kostina2012fixed} for the i.i.d. sources~\cite[Eq. (270)]{kostina2012fixed} to form a typical set of source outcomes.

To describe a typical set of outcomes of the Gauss-Markov source, to each source outcome $\rvec{x}$ (equivalently, $\rvec{u}$) we associate a proxy random variable $\hat{\rvec{X}}(\rvec{x})$ as follows. We first estimate the parameter $a$ in~\eqref{def:GaussMarkov} from the source outcome $\rvec{u}$ using the maximum likelihood estimator $\hat{a}(\rvec{u})$ in~\eqref{MLEformula} in Section~\ref{subsubsec:Teit} above. Then, the proxy random variable $\hat{\rvec{X}}(\rvec{x})$ is defined as a Gaussian random vector with independent (but not identical) coordinates $\hat{X}_i(\rvec{x}) \sim \mathcal{N}(0, \hat{\sigma}_i^2(\rvec{x}))$, where $\hat{\sigma}_i^2(\rvec{x})$'s are the proxy variances defined using $\hat{a}(\rvec{u})$:
\begin{align}
\hat{\sigma}_i^2(\rvec{x}) & \triangleq \frac{\sigma^2}{1+\hat{a}(\rvec{u})^2 -2 \hat{a}(\rvec{u})\cos\LRB{i\pi/(n+1)}}.\label{eqn:XhatVar}
\end{align}
Equivalently, $\hat{\rvec{X}}(\rvec{x})$ is a zero-mean Gaussian random vector whose distribution is given by 
\begin{align}
\hat{\rvec{X}}(\rvec{x})\sim\mathcal{N}\LRB{\rvec{0}, \mathrm{diag}(\hat{\sigma}_1^2(\rvec{x}),...,\hat{\sigma}_n^2(\rvec{x}))}. \label{eqn:Xhat} 
\end{align}

To simplify notations, when there is no ambiguity, we will write $\hat{\rvec{X}}$ and $\hat{\sigma}_i^2$ for $\hat{\rvec{X}}(\rvec{x})$ and $\hat{\sigma}_i^2(\rvec{x})$, respectively. Intuitively, the formula~\eqref{eqn:XhatVar} approximates the eigenvalues of the covariance matrix of $\rvec{U}$ (or equivalently, that of $\rvec{X}$) for a typical $\rvec{x}$. Due to Theorem~\ref{thm:error_concentration}, with probability approaching 1, we have $\hat{a} (\rvec{U})\approx a$, which implies $\hat{\sigma}_i^2 \approx \sigma_i^2$ and $\hat {\rvec{X}}\approx \rvec{X}$. The accuracy of these approximations is quantified in Theorem~\ref{thm:typicalset} below, which is the main tool in the proof of Lemma~\ref{lemma:lemma_2}.

We need a few notations before presenting Theorem~\ref{thm:typicalset}. First, we particularize the CREM problem~\eqref{eqn:cvx_crem}-\eqref{CREM_slope} to the Gauss-Markov source. Let $\rvec{X}$ be the decorrelation of $\rvec{U}$ in~\eqref{def:decorrelation}. For any random vector $\rvec{Y}$ with density, replacing $X$ by $\rvec{X}$ in~\eqref{eqn:cvx_crem} and normalizing by $n$, we define
\begin{align}
\rdf\left(\rvec{X}, \rvec{Y}, d\right) \triangleq  \inf_{P_{\rvec{F}|\rvec{X}}: \mathbb{E}[\mathsf{d}(\rvec{X}, \rvec{F})] \leq d} ~\frac{1}{n}D(P_{\rvec{F}|\rvec{X}} || P_{\rvec{Y}}|P_{\rvec{X}}).
\label{crem:GMdef}
\end{align}
Properties of the CREM~\eqref{crem:GMdef} for the two special cases: when $(i)$ $\rvec{Y}$ is a Gaussian random vector with independent coordinates and $(ii)$ $(\rvec{X}, \rvec{Y})$ forms a RDF-achieving pair, are presented in~Appendix~\ref{app:properties_CREM}. Let $\hat{\rvec{F}}^\star$ be the optimizer of $\mathbb{R}(\hat{\rvec{X}}, \rvec{Y}^\star, d)$, where $\hat{\rvec{X}}$ is defined in~\eqref{eqn:Xhat} and $\rvec{Y}^\star$ in~\eqref{Ystar}. For $\rvec{x}\in\mathbb{R}^n$, define $m_i(\rvec{x})$ as
\begin{align}
m_i(\rvec{x}) \triangleq \mathbb{E}\left[ (\hat{F}_i^\star - x_i )^2~|\hat{X}_i = x_i\right].\label{def:m_i}
\end{align}

\begin{definition}[MLE-typical set]
Fix any $d\in (0,\dmax)$. Given a constant $\alpha>0$, let $\eta_n$ be in~\eqref{rhon} in Section~\ref{subsubsec:Teit} above. For any constant $p>0$ and any $n\in\mathbb{N}$, define $\TS (n, \alpha, p)$ as the set of vectors $\rvec{u}\in \mathbb{R}^n$ satisfying the following conditions:
\begin{align}
\left|\hat{a}(\rvec{u})- a\right| &\leq \eta_n, \label{estimation_error_small}\\
\lrabs{\frac{1}{n}\sum_{i = 1}^n m_i(\rvec{x}) - d} &\leq p\eta_n, \label{mean_near_d}\\
\lrabs{\frac{1}{n}\sum_{i = 1}^n \LRB{\frac{x_i^2}{\sigma_i^2}}^k - (2k-1)!!} &\leq 2,~\text{for } k = 1,2,3, \label{moments_bd}
\end{align}
where $\rvec{x} = \mathsf{S}^\top\rvec{u}$, and $m_i(\rvec{x})$'s are functions of $\rvec{x}$ defined in~\eqref{def:m_i} above.
\label{def:TS}
\end{definition}

The condition~\eqref{estimation_error_small} requires that $\rvec{u}\in \TS(n,\alpha,p)$ should yield a small estimation error, which holds with probability approaching 1 due to Theorem~\ref{thm:error_concentration}. We will explain the condition~\eqref{mean_near_d} in Appendix~\ref{app:proof_lemma_4} below. To gain insight into the condition~\eqref{moments_bd}, note that due to~\eqref{eig:sigma_i}, we have $\frac{X_i}{\sigma_i}\sim\mathcal{N}(0,1)$ and 
\begin{align}
\mathbb{E}\left[\left(\frac{X_i^2}{\sigma_i^2}\right)^k\right] = (2k-1)!!.
\end{align}
Therefore, the condition~\eqref{moments_bd} bounds the variations of $\rvec{X}$, up to its sixth moments, and this condition holds with probability approaching 1 by the Berry-Esseen theorem. Theorem~\ref{thm:typicalset} below summarizes the properties of the typical set $\TS (n, \alpha, p)$ used in the proof of Lemma~\ref{lemma:lemma_2}.

\begin{theorem}[Properties of the MLE-typical set]
\label{thm:typicalset}
For any $d\in (0,\dmax)$ and any constant $\alpha >0$, let $\eta_n$ be given in~\eqref{rhon} in Section~\ref{subsubsec:Teit} above and $p$ be a sufficiently large constant (specifically, $p\geq \eqref{value_p}$ in Appendix~\ref{app:typical_set} below), then we have:
\begin{itemize}
\item [(1).] The probability mass of $\TS(n,\alpha, p)$ is large: there exists a constant $A_1>0$ such that for all $n$ large enough, 
\begin{align}
\mathbb{P}\left[\rvec{U}\in \TS(n,\alpha, p) \right] \geq 1 - \frac{A_1}{\left(\log n\right)^{\kappa\alpha}}, 
\label{TS_large}
\end{align}
where the constant $\kappa$ is defined in~\eqref{kappa} in Section~\ref{subsubsec:Teit} above.

\item [(2).] The proxy variances are good approximations: there exists a constant $A_2>0$ such that for all $n$ large enough, for any $\rvec{u}\in \TS(n,\alpha, p)$, it holds that 
\begin{align}
\left|\hat{\sigma}_i^2(\rvec{x}) - \sigma_i^2\right|\leq A_2\eta_n, \quad~\forall i\in [n], 
\label{eqn:error_sigma}
\end{align}
where $\hat{\sigma}_i^2(\rvec{x})$'s are defined in~\eqref{eqn:XhatVar}.

\item [(3).] Let $\theta > 0$ be the water level matched to $d$ via the limiting reverse waterfilling~\eqref{eqn:para_d_inf}. For all $n$ large enough, for any $\rvec{u}\in\TS(n,\alpha, p)$, it holds that 
\begin{align}
\left|\hat{\lambda}^\star(\rvec{x})  - \lambda^\star \right|\leq \frac{9A_2}{4\theta^2}\eta_n,
\label{eqn:error_lambda}
\end{align}
where $\rvec{x} = \mathsf{S}^\top \rvec{u}$ with $\mathsf{S}$ in~\eqref{eqn:B_eigen_decom}; $\lambda^\star$ is given by~\eqref{lambda_double_star}; 
\begin{align}
\hat{\lambda}^\star(\rvec{x}) = -\rdf'\LRB{\hat{\rvec{X}},\rvec{Y}^\star,d};
\label{lambda_x_hat_star}
\end{align}
$\rvec{X}$ is the decorrelation of $\rvec{U}$ in~\eqref{def:decorrelation}; $(\rvec{X},\rvec{Y}^\star)$ forms a RDF-achieving pair in $\rdf_{\rvec{X}}(n, d)$; and $\hat{\rvec{X}}$ is the proxy Gaussian random variable defined in~\eqref{eqn:Xhat}.
\end{itemize} 
\end{theorem}

\begin{proof}
Appendix~\ref{app:typical_set}.
\end{proof}

\subsection{Achievability proof}

\begin{proof}[Proof of Theorem~\ref{thm:achievability}]
The proof is based on the random coding bound Lemma~\ref{lemma:random} and the lower bound in Lemma~\ref{lemma:lemma_2}. Fix any $d\in (0,\dmax)$ and $\epsilon\in (0,1)$, and let $\theta>0$ be the water level matched to $d$ via the limiting reverse waterfilling~\eqref{eqn:para_d_inf}. We reuse the notations in~\eqref{d_d_n}-\eqref{approx:dis_2}. Similar to the event $\mathcal{E}$ in~\eqref{E_t}, we define the event $\mathcal{F}$ as
\begin{align}
\mathcal{F} \triangleq \lrbb{ \jmath_{\rvec{X}}\LRB{\rvec{X}, d}\leq \jmath_{\rvec{X}}\LRB{\rvec{X}, d_n} + \frac{4C_d}{\tilde{c}\theta}}.
\end{align}
Theorem~\ref{thm:tilted_info_concentrate} implies that 
\begin{align}
\prob{\mathcal{F}} \geq 1 - \frac{1}{n}.
\label{F_t}
\end{align}
Define $\epsilon_n$ as
\begin{align}
\epsilon_n \triangleq \epsilon -  \frac{C_{\textsf{BE}}+1}{\sqrt{n}} - \frac{K}{(\log n)^{\kappa\alpha}} -\frac{1}{n}.
\label{choice:epsilon_n}
\end{align}
Since $\epsilon \in (0,1)$, we have $\epsilon_n\in (0,1)$ for all $n$ large enough. Choose $M$ as
\begin{align}
\log M &\triangleq n\rdf_{U}(d) + \sqrt{n\mathbb{V}_U(d)}Q^{-1}(\epsilon_n) + \log\frac{\log n}{2}+ \notag \\
& \beta_1\log^q n + \beta_2 + c_r +c_v\lrabs{Q^{-1}(\epsilon_n)} +\frac{4C_d}{\tilde{c}\theta}, 
\label{eqn:M}
\end{align}
where $q>1, \beta_1>0, \beta_2$ are the constants in Lemma~\ref{lemma:lemma_2}; and $c_r,c_v$ are the positive constants in~\eqref{approx:rate} and~\eqref{approx:dis_2}. Define the random variable $G_n$ as 
\begin{align}
G_n \triangleq \log M - \jmath_{\rvec{X}}\LRB{\rvec{X}, d_n} - \beta_1\log^q n - \beta_2 - \frac{4C_d}{\tilde{c}\theta}, 
\label{eqn:Gn}
\end{align}
where $\jmath_{\rvec{X}}\LRB{\rvec{X}, d_n}$ is in~\eqref{sumofnrv}. By~\eqref{approx:rate},~\eqref{approx:dis_2}, and~\eqref{eqn:Gn}, we have
\begin{align}
G_n \geq \sum_{i = 1}^n \mathbb{E}_i +Q^{-1}(\epsilon_n)\sqrt{\sum_{i = 1}^n\mathbb{V}_i} - \jmath_{\rvec{X}}\LRB{\rvec{X}, d_n}+ \log\frac{\log n}{2},
\label{bound_G_n}
\end{align}
where $\mathbb{E}_i$'s and $\mathbb{V}_i$'s are defined in~\eqref{conversepf:mean} and~\eqref{conversepf:var}, respectively. Define the event $\mathcal{G}$ as 
\begin{align}
\mathcal{G}\triangleq \left\{G_n < \log \frac{\log n}{2}\right\}.
\end{align}
By~\eqref{bound_G_n},~\eqref{sumofnrv} and the Berry-Esseen Theorem, we have 
\begin{align}
\prob{\mathcal{G}} &\leq \prob{\jmath_{\rvec{X}}\LRB{\rvec{X}, d_n}-\sum_{i = 1}^n \mathbb{E}_i > Q^{-1}(\epsilon_n)\sqrt{\sum_{i = 1}^n\mathbb{V}_i}} \\
&\leq \epsilon_n + \frac{C_{\mathsf{BE}}}{\sqrt{n}}.
\label{BE_En}
\end{align}
Define the event $\mathcal{L}$ as:
\begin{align}
& \mathcal{L} \notag \\
\triangleq & \left\{\log \frac{1}{P_{\rvec{Y^\star}}\left(\mathcal{B}(\rvec{X},d)\right)} \leq \log M  - G_n\right\}\\
= &\left\{\log \frac{1}{P_{\rvec{Y^\star}}\left(\mathcal{B}(\rvec{X},d)\right)} \leq \jmath_{\rvec{X}}\LRB{\rvec{X}, d_n} + \beta_1\log^q n + \beta_2+ \frac{4C_d}{\tilde{c}\theta}\right\}, 
\end{align}
where $\rvec{Y}^\star$ is given in~\eqref{Ystar}. Combining Lemma~\ref{lemma:lemma_2} and~\eqref{F_t} yields
\begin{align}
\prob{\mathcal{L}} \geq 1 - \frac{1}{n} - \frac{K}{\LRB{\log n}^{\kappa\alpha}}.
\label{gap_step}
\end{align}
Indeed, denoting the probability on the left-hand side of~\eqref{eqn:lemma_2} by $\prob{\mathcal{H}}$, we have 
\begin{align}
\prob{\mathcal{H}} &= \prob{\mathcal{H} \cap \mathcal{F}} + \prob{\mathcal{H}\cap \mathcal{F}^c} \\
&\leq \prob{\mathcal{L}} + \frac{1}{n},\label{gapstep}
\end{align}
where~\eqref{gapstep} holds since $\mathcal{H}\cap \mathcal{F}\subseteq \mathcal{L}$.

We have now gathered all the ingredients to prove Theorem~\ref{thm:achievability}. Replacing $\rvec{Y}$ by $\rvec{Y}^\star$ in Lemma~\ref{lemma:random}, we conclude that there exists an $(n, M, d, \epsilon')$ code with 
\begin{align}
& \epsilon' \notag \\
\leq & \mathbb{E}_{\rvec{X}}\left[e^{-MP_{\rvec{Y}^\star}(\mathcal{B}(\rvec{X},d))}\right] \label{steps:1}\\ 
= & \mathbb{E}_{\rvec{X}}\left[e^{-MP_{\rvec{Y}^\star}(\mathcal{B}(\rvec{X},d))} \mathbbm{1}\left\{\mathcal{L}\right\}\right]  + \mathbb{E}_{\rvec{X}}\left[e^{-MP_{\rvec{Y}^\star}(\mathcal{B}(\rvec{X},d))} \mathbbm{1}\left\{\mathcal{L}^c\right\}\right]  \label{steps:2}\\
\leq & \mathbb{E}_{\rvec{X}}\left[e^{-e^{G_n}} \right] + \frac{K}{(\log n)^{\kappa\alpha}} + \frac{1}{n}\label{steps:3}\\
= & \mathbb{E}_{\rvec{X}}\left[e^{-e^{G_n}} \mathbbm{1}\left\{\mathcal{G}\right\}\right]  +\mathbb{E}_{\rvec{X}}\left[e^{-e^{G_n}} \mathbbm{1}\left\{\mathcal{G}^c\right\}\right] + \frac{K}{(\log n)^{\kappa\alpha}} + \frac{1}{n}  \label{steps:4}\\
\leq & \mathbb{P}(\mathcal{G}) + \frac{1}{\sqrt{n}}\mathbb{P}(\mathcal{G}^c)  + \frac{K}{(\log n)^{\kappa\alpha}} + \frac{1}{n} \label{steps:5}\\
\leq &\epsilon_n + \frac{\bec+1}{\sqrt{n}} +\frac{K}{(\log n)^{\kappa\alpha}} + \frac{1}{n}\label{steps:6}\\
= &\epsilon, \label{steps:7}
\end{align}
where~\eqref{steps:1} is by weakening~\eqref{bound:rcb} using $\rvec{Y} = \rvec{Y}^\star$;~\eqref{steps:3} holds by~\eqref{gap_step} and $\mathbbm{1}\left\{\mathcal{L}\right\}M P_{\rvec{Y}^\star}(\mathcal{B}(\rvec{X},d))\geq e^{G_n}$;~\eqref{steps:5} holds since $e^{-e^{G_n}}\leq 1$ and $\mathbbm{1}\left\{\mathcal{G}^c\right\}e^{-e^{G_n}}\leq \frac{1}{\sqrt{n}}$;~\eqref{steps:6} is by~\eqref{BE_En}; and~\eqref{steps:7} is by the choice of $\epsilon_n$ in~\eqref{choice:epsilon_n}.  Consequently, since there exists an $(n,M,d,\epsilon')$ code with $\epsilon'\leq \epsilon$, we must have 
\begin{align}
R(n,d,\epsilon) \leq \frac{\log M}{n},
\label{achi_last}
\end{align}
where $\log M$ is given by~\eqref{eqn:M}. Similar to the converse proof, plugging~\eqref{eqn:M} into~\eqref{achi_last} and then using the Taylor expansion of $Q^{-1}(\epsilon_n)$ yields~\eqref{eqn:achievability}.
\end{proof}

\section{Conclusion}
\label{sec:conclusion}
In this paper, we derived the reverse waterfilling characterization~\eqref{eqn:wf_dis} of the dispersion for lossy compression of the Gauss-Markov source~\eqref{def:GaussMarkov} with $|a|<1$ (Theorem~\ref{thm:main_thm_intro}). This is the first dispersion result for lossy compression of sources with memory. In doing so, we developed several novel technical tools, which are highlighted below.
\begin{itemize}
\item We derived the expression for the limiting variance of the $\mathsf{d}$-tilted information for the Gauss-Markov source in Theorem~\ref{thm:limiting_variance}. Its proof relies on our parametric representation for the $\mathsf{d}$-tilted information, presented in Lemma~\ref{lemma:para_tilted} in Appendix~\ref{app:derivation_vd}.

\item Theorem~\ref{thm:n_LimitingThm} presented a nonasymptotic refinement of Gray's result~\cite{gray1970information} (restated in Theorem~\ref{thm:LimitingThm}) on the eigenvalue distribution of the covariance matrix of the random vector $\rvec{U}$ from the Gauss-Markov source. The key tool we developed to prove Theorem~\ref{thm:n_LimitingThm} is Lemma~\ref{lemma:eig_approx} in Section~\ref{subsubsec:Teig}, which is a sharp bound relating the eigenvalues of two sequences of symmetric tridiangonal matrices.

\item The maximum likelihood estimator $\hat{a}(\rvec{u})$, defined in~\eqref{MLEformula} and analyzed in Theorems~\ref{thm:ECG} and~\ref{thm:error_concentration}, is of independent interest as it allows one to estimate the distribution of $\rvec{u}$ drawn from the class of the Gauss-Markov sources with unknown $a$. The error bounds in Theorem~\ref{thm:ECG} rely on the Hanson-Wright inequality~\cite[Th. 1.1]{rudelson2013hanson}. That inequality applies beyond the case when $Z_i$'s are Gaussian, which means that our approach can be applied to other sources with memory.

\item To prove achievability, we constructed a typical set in Definition~\ref{def:TS} based on the maximum likelihood estimator. This idea of constructing typical sets via estimators could also find its use in other problems.
\end{itemize}

Finally, we discuss several open problems. 
\begin{itemize}
\item The dispersion for Gauss-Markov sources with $|a| \geq 1$ is unknown. This paper treats the asymptotically stationary case, i.e., $|a| <1$. The case $|a| \geq 1$ is fundamentally different, since that source is nonstationary. The rate-distortion functions for nonstationary Gaussian autoregressive processes were first derived by Gray~\cite[Eq. (22)]{gray1970information} in 1970, and later in 1980 by Hashimoto and Arimoto~\cite[Eq.(6)]{hashimoto1980rate} in an equivalent but distinct form; that equivalence was shown by Gray and Hashimoto~\cite{gray2008note} in 2008. Gray's reverse waterfilling~\cite[Eq. (22)]{gray1970information} is different from Kolmogorov's reverse waterfilling~\eqref{eqn:para_rate_inf} in the nonstationary case, where the later does not apply. Therefore, in order to characterize the dispersion for the case $|a|\geq 1$, one would need to use Gray's reverse waterfilling~\cite[Eq. (22)]{gray1970information} for $\rdf_{U}(d)$.

\item A natural generalization of this work would be to consider the dispersion for the general stationary Gaussian autoregressive processes~\eqref{eqn:GeneralGaussianAR}. The geometric converse proof in Section~\ref{sec:converse} already yields a converse bound on $R(n,d,\epsilon)$, which is tight in the low distortion regime $d\in (0,d_c]$ in the first-order term; we conjecture it is also tight in the second-order term. A possible way to show a matching achievability bound for the Gaussian AR processes of order $m$, inspired by the estimation idea in this paper, is to analyze an estimator which estimates the vector $\rvec{a} = (a_1, ..., a_m)^\top$ in~\eqref{eqn:GeneralGaussianAR} instead of the scalar $a$. To deal with large distortions, i.e. $d > d_c$, sharp bounds on eigenvalues of $\mathsf{A}^\top\mathsf{A}$ with $\mathsf{A}$ given by~\eqref{def:A_general} need to be derived, similar to Lemma~\ref{lemma:eig_approx} in Section~\ref{subsubsec:Teig}; the tools in Appendix~\ref{app:EigenApprox} might be useful.

\item A formula (analogous to~\eqref{eqn:main_GA}) for the channel dispersion of the Gaussian intersymbol interference (ISI) channels, see~\cite[Eq. (29)]{polyanskiy2009dispersion},  was presented in~\cite[Th. 5]{polyanskiy2009dispersion} without proof. The channel capacity of the Gaussian ISI channel is well-known, e.g.~\cite[Theorem]{tsybakov1970capacity} and~\cite[Th. 1]{hirt1988capacity}. The tools in this paper might be useful in obtaining a proof of the channel dispersion for Gaussian ISI channels in~\cite[Th. 5]{polyanskiy2009dispersion}.

\item A fundamental problem left open is how widely the limiting formula for the dispersion
\begin{align}
V(d) = \limsup_{n\rightarrow\infty }\frac{1}{n}\var{\jmath_{\rvec{X}}(\rvec{X},d)}
\end{align}
applies. Theorem~\ref{thm:main_thm_intro} and Theorem~\ref{thm:limiting_variance} established its validity for the Gauss-Markov source. We conjecture that it continues to apply whenever the central limit theorem type of results can be derived for $\jmath_{\rvec{X}}(\rvec{X},d)$.
\end{itemize}

\section*{Acknowledgment} 
\label{sec:ack}
We would like to thank the associate editor Dr. Shun Watanabe and the anonymous reviewers for their insightful comments that are reflected in the final version. 
\appendices

\section{}
\label{app:ATools}

\subsection{A roadmap of the paper}
\label{app:roadmap}
The relations of our main theorems, lemmas, corollaries are presented in Fig.~\ref{fig:roadmap}.
\begin{figure*}
\centering
\includegraphics[scale=0.31]{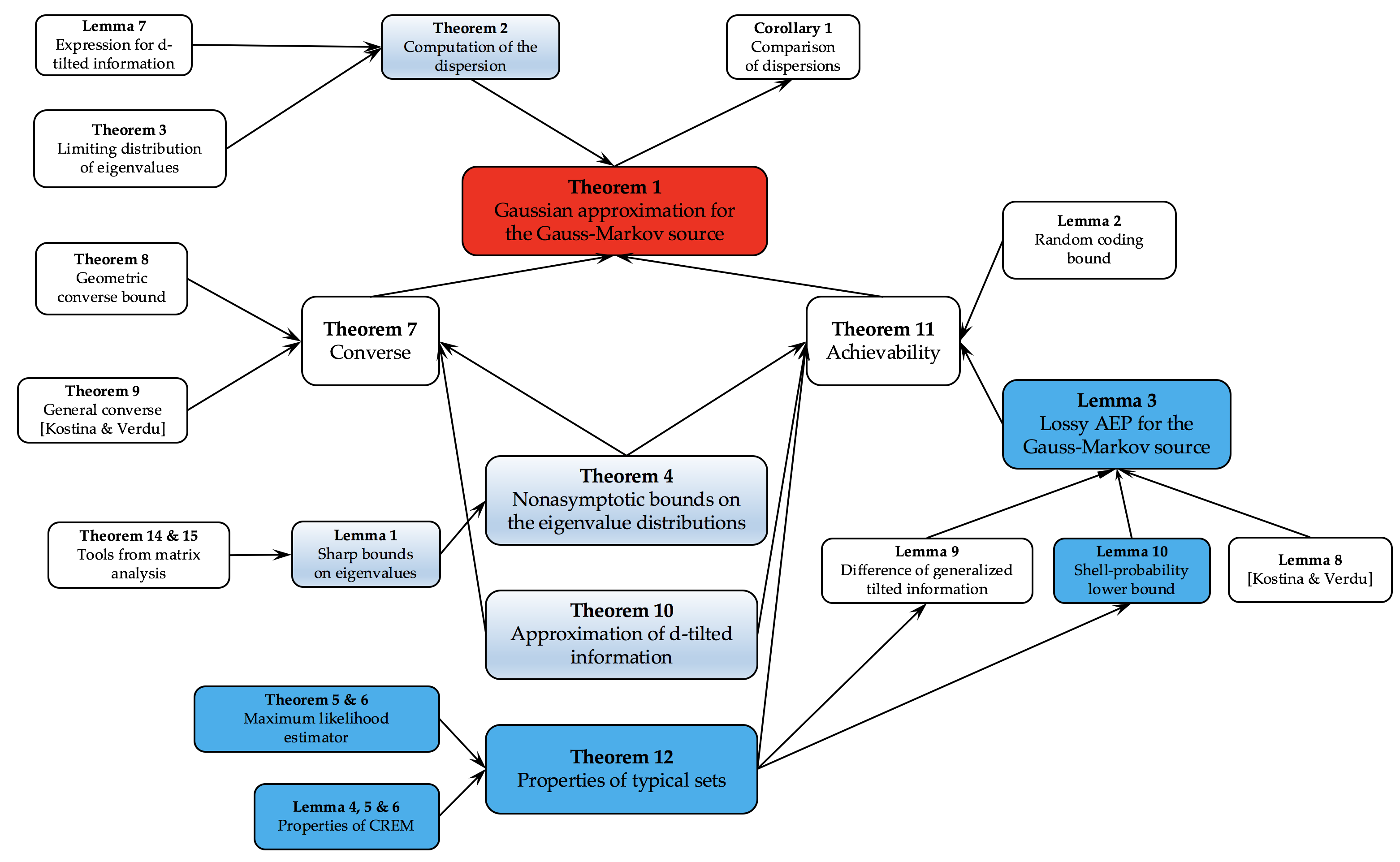}
\caption{A roadmap of the paper: an arrow from block $A$ to block $B$ means that the derivation of block $B$ is based on block $A$. Our main result is Theorem~\ref{thm:main_thm_intro}. The lightly shaded blocks consist of results that are novel and relatively easier to obtain, while the heavily shaded blocks consist of our main technical contributions.}
\label{fig:roadmap}
\end{figure*}

\subsection{Classical theorems}
\label{app:classicalthms}
\begin{theorem}[Berry-Esseen Theorem, e.g. {\cite[Chap. 16.5]{feller1971intro}}]
Let $W_1,\ldots,W_n$ be a collection of independent zero-mean random variables with variances $V_i^2 > 0$ and finite third absolute moment $T_i\triangleq \mathbb{E}[|W_i|^3] < +\infty$. Define the average variance $V^2$ and average third absolute moment $T$ as 
\begin{align}
V^2\triangleq \frac{1}{n}\sum_{i = 1}^n V_i^2,  \quad T \triangleq \frac{1}{n}\sum_{i = 1}^n T_i.
\end{align}
Then for $n\in\mathbb{N}$, we have 
\begin{align}
\sup_{t\in\mathbb{R}}\left |\prob{\frac{1}{V \sqrt{n}}\sum_{i = 1}^n W_i < t}  - \Phi(t) \right | \leq \frac{6T}{V^3\sqrt{n}},
\end{align}
where $\Phi$ is the cdf of the standard normal distribution $\mathcal{N}(0,1)$.
\label{thm:BE}
\end{theorem}

\begin{remark}
Since in this paper, we only consider random variables $W_i$'s with bounded $p$-th moment for any finite $p$, it is easy to check that there exists a constant $\bec > 0 $ such that 
\begin{align}
\sup_{t\in\mathbb{R}}\left |\prob{\frac{1}{V \sqrt{n}}\sum_{i = 1}^n W_i < t}  - \Phi(t) \right | \leq \frac{\bec}{\sqrt{n}}.
\end{align}
While the constant $\bec$ depends on the random variables $W_i$'s, to simplify notations, we use $\bec$ in all applications of the Berry-Esseen Theorem.
\end{remark}


\subsection{Justification of~\eqref{relation:mean_rdf}}
\label{app:just}
We provide a short justification of how~\eqref{relation:mean_rdf} follows from~\cite{csiszar1974extremum}. We use the same notations in~\cite{csiszar1974extremum}. We denote $P$ the distribution of the source $X$, and $Q_0$ the optimal reproduction distribution in $\mathbb{R}_X(d)$, that is, $Q_0$ is the $Y$-marginal of a minimizer $\Tilde{P}_0$. First, \cite[Corrollary, Eq. (1.25)]{csiszar1974extremum} shows that 
\begin{align}
    \mathbb{R}_X(d) = \max_{\alpha (x),~s}~\mathbb{E}_P\left[\log\alpha (X)\right]-sd.\label{eqn:max}
\end{align}
Next, the proof of \cite[Lem. 1.4]{csiszar1974extremum} in \cite[Eq. (1.27)-(1.32)]{csiszar1974extremum} shows that the maximizer $(\alpha^\star(\cdot), s^\star)$ is given by \cite[Eq. (1.15)]{csiszar1974extremum} with $s^\star = \lambda^\star$ due to \cite[Eq. (1.12)]{csiszar1974extremum}. For convenience, we write down \cite[Eq. (1.15)]{csiszar1974extremum}: 
\begin{align}
    \alpha^\star (x) =  \frac{1}{\mathbb{E}_{Q_0}[\exp(-\lambda^\star\mathsf{d}(x,Y^\star))]},\label{eqn:alphax}
\end{align}
where $Y^\star \sim Q_0$. Finally, plugging~\eqref{eqn:alphax} into~\eqref{eqn:max} yields~\eqref{relation:mean_rdf}. \qed

\section{Proofs in Section~\ref{sec:MainResults}}
\label{app:MainResults}

\subsection{Corner points on the dispersion curve}
\label{app:cornerpoints}
We derive~\eqref{corner_points} using the residue theorem from complex analysis~\cite[Th. 17]{ahlfors1979complex}. Similar ideas have been applied by Berger~\cite[Chap. 6, p. 232]{berger1971rate} and Gray~\cite[Eq. (12)]{gray2008note} to study the rate-distortion functions for the nonstationary Gaussian AR processes~\eqref{eqn:GeneralGaussianAR}. The coordinate of $P_1$ in Fig.~\ref{fig:dispersion} can be easily obtained as follows. The water level matched to $d_c$ via~\eqref{eqn:para_d_inf} is $\theta_{\min}$ in~\eqref{def:general_d_c}. Hence,~\eqref{eqn:wf_dis} is simplified as 
\begin{align}
V_U(d_c) = \frac{1}{4\pi}\int_{-\pi}^{\pi}~1~dw = \frac{1}{2}.
\end{align}
To treat $P_2$, note that the water level matched to $\dmax$ via~\eqref{eqn:para_d_inf} is $\theta_{\max}$ in~\eqref{theta_max}, which, due to~\eqref{spectrum}, equals, 
\begin{align}
\theta_{\max} = \frac{\sigma^2}{(1-a)^2}.
\end{align}
This implies that~\eqref{eqn:wf_dis} evaluates as
\begin{align}
V_U(\dmax)= \frac{\sigma^4}{4\pi\theta_{\mathrm{max}}^2}\int_{-\pi}^{\pi}\frac{1}{\left(g(w)\right)^2}~dw.
\end{align}
Invoking the residue theorem~\cite[Th. 17]{ahlfors1979complex}, we will obtain the integral 
\begin{align}
\mathcal{I} \triangleq \int_{-\pi}^{\pi}\frac{1}{\left(g(w)\right)^2}~dw = \frac{2\pi (1+a^2)}{(1-a^2)^3}, 
\end{align}
which will complete the derivation. To that end, change variables using $z = e^{\mathrm{j}w}$ and rewrite 
\begin{align}
g(w) &= 1+a^2 - a(z + z^{-1}) \\
&= (z^{-1}-a)(z-a).
\end{align}
The integral $\mathcal{I}$ is then 
\begin{align}
\mathcal{I} &= \oint_{|z| = 1} \frac{z}{\mathrm{j}a^2(z - a^{-1})^2(z-a)^2}~dz \label{RT1}\\
&= 2\pi\mathrm{j}~\mathrm{Res}_{z = a} \frac{z}{\mathrm{j}a^2(z - a^{-1})^2(z-a)^2} \label{RT2}\\
&=  \frac{2\pi}{a^2}\lim_{z\rightarrow a}\frac{d}{dz}\frac{z}{(z-a^{-1})^2} \label{RT3}\\
&= \frac{2\pi (1+a^2)}{(1-a^2)^3},\label{RT4}
\end{align}
where~\eqref{RT1} is by the change of variable $z = e^{\mathrm{j}w}$;~\eqref{RT2} is due to the residue theorem and $a\in [0,1)$; and~\eqref{RT3} is the standard method of computing residues.\qed

\subsection{Two interpretations of the maximum distortion}
\label{app:int_dmax}
We present the computation details of~\eqref{variance:U_lim} and how~\eqref{def:dmax} leads to~\eqref{d_max}. Using the same technique as in~\eqref{RT1}-\eqref{RT4}, we compute~\eqref{def:dmax} as
\begin{align}
\dmax &= \frac{\sigma^2}{2\pi}\oint_{|z| = 1} \frac{1}{\mathrm{j}z(z^{-1} - a)(z-a)}~dz \\
& =  \sigma^2 \mathrm{Res}_{z = a} \frac{1}{-a(z - a^{-1})(z-a)} \\
& =  \sigma^2\lim_{z\rightarrow a} \frac{1}{-a(z - a^{-1})} \\
& = \frac{\sigma^2}{1-a^2}.
\end{align}
To compute the stationary variance, take the variance on both sides of~\eqref{def:GaussMarkov},
\begin{align}
\var{U_i} = a^2\var{U_{i-1}} + \sigma^2, 
\label{eqn:finite_k}
\end{align}
then taking the limit on both sides of~\eqref{eqn:finite_k}, we have
\begin{align}
\lim_{i\rightarrow\infty} \var{U_i} = a^2\lim_{i\rightarrow\infty} \var{U_{i-1}} + \sigma^2, 
\end{align}
which implies 
\begin{align}
\lim_{i\rightarrow\infty} \var{U_i} = \frac{\sigma^2}{1-a^2}.
\end{align}
\qed

\section{Proofs in Section~\ref{sec:model}}
\label{app:model}

\subsection{Eigenvalues of nearly Toeplitz tridiagonal matrices}
\label{app:EigenApprox}
For convenience, we record two import results from matrix theory.
\begin{theorem}[{Cauchy Interlacing Theorem for eigenvalues~\cite[p.59]{bhatia2013matrix}}]
Let $\mathsf{H}$ be an $n\times n$ Hermitian matrix partitioned as $\mathsf{H} = \begin{pmatrix}
\mathsf{P} & \star\\
\star & \star
\end{pmatrix},$ where $\mathsf{P}$ is an $(n-1)\times (n-1)$ principal submatrix of $\mathsf{H}$. Let $\lambda_1(\mathsf{P})\leq \lambda_2(\mathsf{P}) \ldots\leq \lambda_{n-1}(\mathsf{P})$ be the eigenvalues of $\mathsf{P}$, and $\lambda_1(\mathsf{H})\leq \lambda_2(\mathsf{H}) \ldots\leq \lambda_{n}(\mathsf{H})$ be the  eigenvalues of $\mathsf{H}$, then $\lambda_{i}(\mathsf{H})\leq \lambda_i(\mathsf{P})\leq \lambda_{i+1}(\mathsf{H})$ for $i = 1,...,n-1$. 
\label{thm:cauchy}
\end{theorem}

\begin{theorem}[{Gershgorin circle theorem~\cite[p.16, Th. 1.11]{varga2009matrix}}]
Let $\mathsf{M}$ be any $n\times n$ matrix, with entries $m_{ij}$. Define $r_i \triangleq \sum_{j\neq i} \left|m_{ij}\right|,~\forall~i\in [n]$, then for any eigenvalue $\lambda$ of $\mathsf{M}$, there exists $i\in [n]$ such that $|\lambda - m_{ii}|\leq r_i.$
\label{thm:GCT}
\end{theorem}

\begin{proof}[Proof of Lemma~\ref{lemma:eig_approx}]
To indicate the dimension, denote by $\mathsf{A}_n$ the matrix $\mathsf{A}$ defined in~\eqref{def:A}, and denote 
\begin{align}
\mathsf{B}_n & \triangleq  \mathsf{A}_n^\top\mathsf{A}_n \\
& = 
\begin{pmatrix}
1+a^2 & -a          &   0      & 0    &  \ldots & 0\\
-a       & 1+a^2    &  -a      & 0   & \ldots   & 0 \\
0        & -a          & 1+a^2  & -a &  \ddots & \vdots \\
\vdots & \ddots & \ddots & \ddots &\ddots & \vdots \\
 \vdots & \ddots & 0 & -a &1+a^2 & -a \\
 0 & \ldots & \ldots & 0 & -a &1
\end{pmatrix}.
\label{matrix:B_n}
\end{align} 
Notice that we obtain a tridiagonal Toeplitz matrix $\mathsf{W}_n$ if the $(n,n)$-th entry of $\mathsf{B}_{n}$ is replaced by $1+a^2$: 
\begin{align}
\mathsf{W}_n = 
\begin{pmatrix}
1+a^2 & -a          &   0      & 0    &  \ldots & 0\\
-a       & 1+a^2    &  -a      & 0   & \ldots   & 0 \\
0        & -a          & 1+a^2  & -a &  \ddots & \vdots \\
\vdots & \ddots & \ddots & \ddots &\ddots & \vdots \\
 \vdots & \ddots & 0 & -a &1+a^2 & -a \\
 0 & \ldots & \ldots & 0 & -a &1 +a^2
\end{pmatrix},
\label{matrix:W_n}
\end{align}
whose eigenvalues $\xi^{(n)}_1\leq \xi^{(n)}_2\ldots\leq \xi^{(n)}_n$ are given by~\eqref{eig:W}, see~\cite[Eq. (4)]{noschese2013tridiagonal}. At an intuitive level, we expect $\xi^{(n)}_i$'s to approximate $\mu_i$'s well since $\mathsf{B}_n$ and $\mathsf{W}_n$ differ in only one entry. The first part of the proof applies the Cauchy interlacing theorem (Theorem~\ref{thm:cauchy}) to show~\eqref{bound:ximu} for $2\leq i \leq n$. The bound~\eqref{bound:ximu} for $i = 1$ is proved via the Gershgorin circle theorem (Theorem~\ref{thm:GCT}) in the second part.

Applying Theorem~\ref{thm:cauchy} by partitioning $\mathsf{B}_n$ as
\begin{align}
\mathsf{B}_n = \begin{pmatrix}
\mathsf{W}_{n-1} & \star\\
\star & 1
\end{pmatrix},
\end{align}
we obtain 
\begin{align}
\mu_i\leq\xi_{i}^{(n-1)}\leq \mu_{i+1},\quad\forall~i\in [n-1].
\label{bound:CIT}
\end{align}
On the other hand, since $\mathsf{W}_n \succeq \mathsf{B}_n$ in the  semidefinite order, we have 
\begin{align}
\xi_i^{(n)} \geq \mu_i,\quad\forall~i\in [n]. 
\label{bound:SD}
\end{align}
Combining~\eqref{bound:CIT} and~\eqref{bound:SD} yields
\begin{align}
 \xi^{(n-1)}_{i-1}\leq\mu_i\leq \xi^{(n)}_i,\quad\forall~i = 2,..., n.
 \label{bound:sandwich}
\end{align}
Simple algebraic manipulations using~\eqref{eig:W} and~\eqref{bound:sandwich} lead to
\begin{align}
\xi_i^{(n)} - \mu_i &\leq \xi_i^{(n)} -\xi_{i-1}^{(n-1)}  \leq \frac{2\pi a}{n},\quad\forall~i = 2,...,n. 
\end{align}

To bound the difference $\xi_1^{(n)} - \mu_1$, we apply Theorem~\ref{thm:GCT} to $\mathsf{B}_n$. Note that for $\mathsf{B}_n$, we have $r_1 = r_n = a$ and $r_i = 2a,~\forall i =  2,...,n-1$ (recall $r_i$'s defined in Theorem~\ref{thm:GCT}). For the eigenvalue $\mu_1$, there exists $j\in [n]$ such that $|\mu_1 - \mathsf{B}_{jj}| \leq r_{j}$. The following analyses lead to $\mu_1\geq (1 - a)^2$:
\begin{itemize}
\item If $2\leq j\leq n-1$, then $|\mu_1 - (1+a^2)| \leq 2a$, which implies that $\mu_1 \geq 1+a^2 - 2a$.
\item If $j= 1$, then $|\mu_1 - (1+a^2)| \leq a$, which implies that $\mu_1 \geq 1+a^2 - a\geq 1+a^2 - 2a$.
\item If $j = n$, then $|\mu_1 - 1| \leq a$, which implies $\mu_1 \geq 1- a\geq (1- a)^2$.
\end{itemize}
Recall from~\eqref{eig:W} that $\xi_1^{(n)} = 1+a^2 - 2a\cos\left(\frac{\pi}{n+1}\right)$. Hence, 
\begin{align}
\xi_1^{(n)} - \mu_1 &\leq 2a\left[1 - \cos\left(\frac{\pi}{n+1}\right)\right]\\
 &\leq \frac{a\pi^2}{(n+1)^2},\label{mu_1_bd}\\
 &\leq \frac{2a\pi}{n}
\end{align}
where~\eqref{mu_1_bd} is by the inequality $\cos(x)\geq 1- x^2/2$.
\end{proof}

\subsection{Proof of Theorem~\ref{thm:n_LimitingThm}}
\label{app:pf_n_LimitingThm}
\begin{proof}
Since $S(w)$ in~\eqref{spectrum} is even in $w\in[-\pi, \pi]$, we have
\begin{align}
I &\triangleq\frac{1}{2\pi}\int_{-\pi}^{\pi} F\left[S(w)\right]~dw \\
& = \frac{1}{\pi}\int_{0}^{\pi} F\left[S(w)\right]~dw.
\end{align}
We bound the integral $I$ by Riemann sums over intervals of width $\frac{\pi}{n+1}$, see Fig.~\ref{fig:integration}. Since $F\left[S(w)\right]$ is a nonincreasing function in $w\in [0,\pi]$, we have 
\begin{align}
I \geq \frac{1}{\pi}\sum_{i = 1}^n F\left[S\left(\frac{i\pi }{n+1}\right)\right]\frac{\pi}{n+1}.
\label{riemannsum}
\end{align}
\begin{figure}[!ht]
\centering
\includegraphics[scale=0.4]{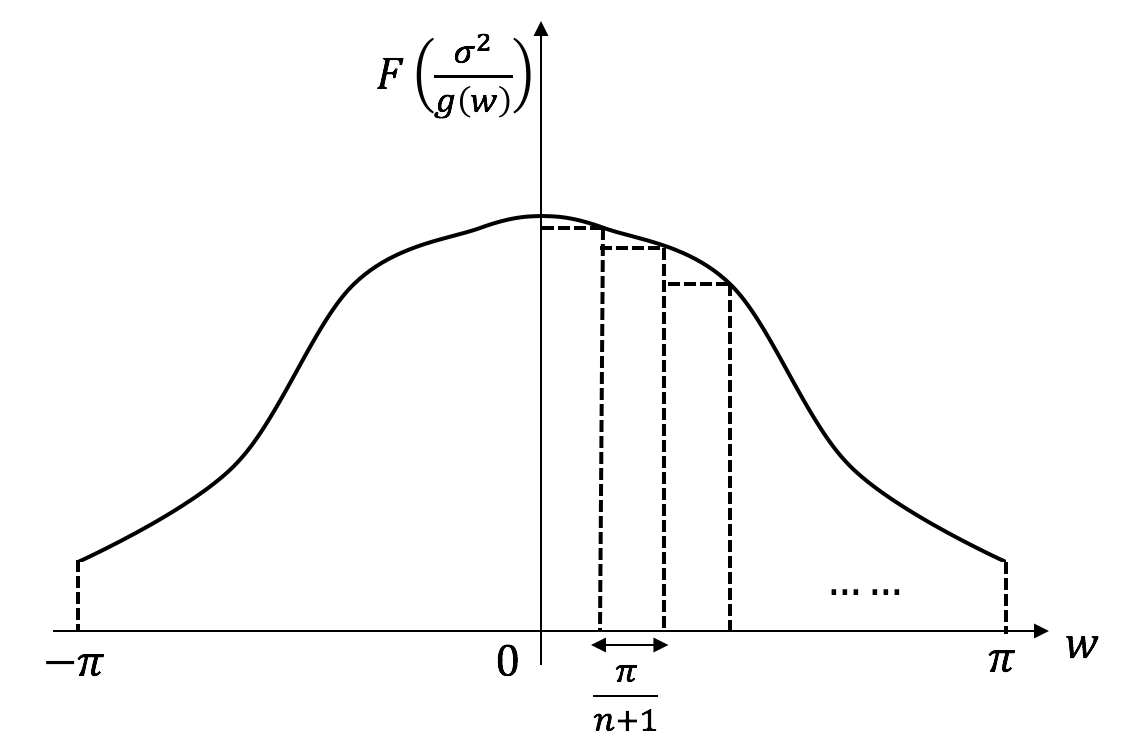}
\caption{Bound the integral $I$ by Riemann sums.}
\label{fig:integration}
\end{figure}
Using Lemma~\ref{lemma:eig_approx}, we can further bound~\eqref{riemannsum} from below as
\begin{align}
I &\geq \frac{1}{n+1} \sum_{i = 1}^n F\left(\frac{\sigma^2}{\mu_i + 2a\pi/n}\right).
\label{bound:riemansum1}
\end{align}
Since $F$ is $L$-Lipschitz, we have for $i\in [n]$,
\begin{align}
F\left(\frac{\sigma^2}{\mu_i + 2a\pi/n}\right) &\geq F\left(\frac{\sigma^2}{\mu_i}\right) - L \left( \frac{\sigma^2}{\mu_i} - \frac{\sigma^2}{\mu_i + 2a\pi/n}\right) \\
&\geq F\left(\frac{\sigma^2}{\mu_i}\right) - \frac{2a\pi L\sigma^2}{n\mu^2_{i}}.
\label{bound:riemannsum2}
\end{align}
Plugging~\eqref{bound:riemannsum2} into~\eqref{bound:riemansum1}, we obtain
\begin{align}
I\geq \frac{1}{n+1}\sum_{i = 1}^n F\left(\frac{\sigma^2}{\mu_i}\right)  - \frac{2aL\pi\sigma^2}{n(n+1)}\sum_{i = 1}^n \frac{1}{\mu_i^2}.
\end{align}
From~\eqref{constant_eig_bound}, we see that
\begin{align}
\frac{1}{n}\sum_{i = 1}^n \frac{1}{\mu_i^2}\leq \frac{1}{(1-a)^4}.
\end{align}
Let $\|F\|_{\infty}$ be the sup norm of $F$ over the interval~\eqref{interval_t}, then 
\begin{align}
I &\geq \frac{1}{n+1}\sum_{i=1}^n F\left(\frac{\sigma^2}{\mu_i}\right) - \frac{2aL\pi\sigma^2}{(n+1)(1-a)^4} \\
&\geq \frac{1}{n}\sum_{i=1}^n F\left(\frac{\sigma^2}{\mu_i}\right)  - \frac{ \|F\|_{\infty}+ 2aL\pi\sigma^2 /(1-a)^4 }{n}.
\end{align}
Similarly, we can derive the upper bound
\begin{align}
I \leq \frac{1}{n} \sum_{i = 1}^{n} F\left(\frac{\sigma^2}{\mu_i}\right) + \frac{2\|F\|_{\infty}}{n}.
\end{align}
Therefore, setting 
\begin{align}
C_L \triangleq \max\left\{\|F\|_{\infty} + \frac{2aL\pi\sigma^2}{(1-a)^4},~2\|F\|_{\infty}\right\}
\end{align}
completes the proof. 
\end{proof}

\subsection{Properties of the conditional relative entropy minimization problem}
\label{app:properties_CREM}
This section presents three results on the CREM problem~\eqref{crem:GMdef}, all of which are necessary to the proof of Theorem~\ref{thm:achievability}.

\subsubsection{Gaussian CREM}
\label{subsec:gaussian_crem}
The optimization problem~\eqref{crem:GMdef} is referred to as the Gaussian CREM when $\rvec{X}$ and $\rvec{Y}$ are Gaussian random vectors with independent coordinates. The optimizer and optimal value of the Gaussian CREM are characterized by the following lemma.
\begin{lemma}
Let $\rvec{X}$ and $\rvec{Y}$ be Gaussian random vectors with independent coordinates, i.e., 
\begin{align}
\rvec{X}&\sim\mathcal{N}(\rvec{0}, \mathsf{\Sigma}_\mathbf{X}), \quad \text{where } \mathsf{\Sigma}_\mathbf{X} = \mathrm{diag}\LRB{\alpha_1^2,\ldots,\alpha_n^2},\label{GCREM:X}\\
\rvec{Y}&\sim\mathcal{N}(\rvec{0}, \mathsf{\Sigma}_\mathbf{Y}), \quad\text{where }\mathsf{\Sigma}_\mathbf{Y} = \mathrm{diag}\LRB{\beta_1^2,\ldots,\beta_n^2}.\label{GCREM:Y}
\end{align}
Then, the optimizer $P_{\rvec{F}^\star | \rvec{X}}$ in the Gaussian CREM~\eqref{crem:GMdef} $\rdf\LRB{\rvec{X},\rvec{Y}, d}$ is
\begin{align}
P_{\rvec{F}^\star | \rvec{X}} =\prod_{i = 1}^n P_{F_i^\star | X_i},
\label{opt:prod}
\end{align}
where for any $\rvec{x}\in\mathbb{R}^n$, the conditional distribution of $F_i^\star$ given $X_i = x_i$ is\footnote{When $\beta_i^2 = 0$ for some $i\in [n]$, the random variable in~(\ref{opt:prod_i}) degenerates to a deterministic random variable taking value 0, and the notation $\mathcal{N}(0,0)$ denotes the Dirac delta function.}
\begin{align}
F_i^\star | \lrbb{X_i = x_i} \sim \mathcal{N}\LRB{\frac{2\delta^\star \beta_i^2x_i}{1+2\delta^\star \beta_i^2}, \frac{\beta_i^2}{1+2\delta^\star \beta_i^2}}, 
\label{opt:prod_i}
\end{align}
and the optimal value is 
\begin{align}
\rdf\LRB{\rvec{X},\rvec{Y}, d} & = -\delta^\star d  +  \frac{1}{2n}\sum_{i  =1}^n \log\left(1+2\delta^\star \beta_i^2\right) + \notag \\ 
&\quad\quad \frac{1}{n}\sum_{i = 1}^n \frac{\delta^\star \alpha_i^2}{1+2\delta^\star \beta_i^2},
\label{eqn:expression}
\end{align} 
where $\delta^\star$ is the negative slope defined as
\begin{align}
\delta^\star = -\rdf'\left(\rvec{X}, \rvec{Y}, d\right).
\end{align}
\label{lemma:GaussianCREM}
\end{lemma}

\begin{proof}
We particularize~\eqref{relation:generaltilted} to the Gaussian CREM. For any fixed $\rvec{x}\in\mathbb{R}^n$, rearranging~\eqref{relation:generaltilted} yields
\begin{align}
& f_{\rvec{F}^\star|\rvec{X}}\LRB{\rvec{y}|\rvec{x}} \notag \\
=& f_{\rvec{Y}}(\rvec{y}) \exp\left\{\Lambda_{\rvec{Y}}(\rvec{x},\delta^\star,d) -  \delta^\star n \dis{\rvec{x}}{\rvec{y}} + \delta^\star n d\right\} \label{gaussianCREM:step1}\\
\propto & f_{\rvec{Y}}(\rvec{y})  \exp\left\{-\delta^\star n \dis{\rvec{x}}{\rvec{y}}\right\} \label{gaussianCREM:step2}\\
\propto & \exp\left\{ -\delta^\star \sum_{i = 1}^n(y_i-x_i)^2 - \sum_{i = 1}^n \frac{y_i^2}{2\beta_i^2}\right\} \label{gaussianCREM:step3}\\
= & \prod_{i = 1}^n \exp\left\{-\frac{\left(y_i - \frac{2\delta^\star \beta_i^2x_i}{1+2\delta^\star \beta_i^2}\right)^2}{\frac{2\beta_i^2}{1+2\delta^\star \beta_i^2}}\right\},\label{gaussianCREM:step4}
\end{align}
where $p_1\propto p_2$ means that $p_1 = c' p_2$ for a positive constant $c'$;~\eqref{gaussianCREM:step2} is by keeping only terms containing $\rvec{y}$ (since $\rvec{x}$ is fixed);~(\ref{gaussianCREM:step3}) is by plugging the pdf of $\rvec{Y}$ into~\eqref{gaussianCREM:step2}; and~\eqref{gaussianCREM:step4} is by completing the squares in $y_i$. Hence,~\eqref{opt:prod} and~\eqref{opt:prod_i} follow. Next, the expression~\eqref{eqn:expression} is obtained by a direct computation using~\eqref{relation:generaltilted},~\eqref{opt:prod} and~\eqref{opt:prod_i}.
\begin{align}
& \mathbb{R}(\rvec{X},\rvec{Y}, d) \notag \\
= & \frac{1}{n} \int_{\mathbb{R}^n} f_{\rvec{X}} (\rvec{x})  \int_{\mathbb{R}^n} f_{\rvec{F}^\star| \rvec{X}} (\rvec{y} | \rvec{x} ) \Big [\Lambda_{\rvec{Y}}(\rvec{x},\delta^\star,d) - \notag\\
&\quad \delta^\star n \dis{\rvec{x}}{\rvec{y}} + \delta^\star n d\Big ]~d\rvec{y}d\rvec{x}\label{gaussianCREM:step5}\\
=& \frac{1}{n} \int_{\mathbb{R}^n} f_{\rvec{X}} (\rvec{x})  \int_{\mathbb{R}^n} f_{\rvec{F}^\star| \rvec{X}} (\rvec{y} | \rvec{x} ) \Lambda_{\rvec{Y}}(\rvec{x},\delta^\star,d)~d\rvec{y}d\rvec{x}\label{gaussianCREM:step6}\\
=&-\delta^\star d  +  \frac{1}{2n}\sum_{i  =1}^n \log\left(1+2\delta^\star \beta_i^2\right) + \frac{1}{n}\sum_{i = 1}^n \frac{\delta^\star \alpha_i^2}{1+2\delta^\star \beta_i^2}, \label{gaussianCREM:step7}
\end{align}
where~\eqref{gaussianCREM:step5} follows by substituting~\eqref{relation:generaltilted} into~\eqref{crem:GMdef};~\eqref{gaussianCREM:step6} holds since $\mathbb{E}[\dis{\rvec{X}}{\rvec{F}^\star}] = d$ by the optimality of $\rvec{F}^\star$; and~\eqref{gaussianCREM:step7} is by direct integration of~\eqref{gaussianCREM:step6}, which relies on the definition of the generalized tilted information~\eqref{def:generalizeddtilted} and the well-known formula for the moment generating function (MGF) of a noncentral $\chi^2_1$-distribution. 
\end{proof}

\subsubsection{Gauss-Markov CREM}
\label{app:GM_CREM}
The optimization problem~\eqref{crem:GMdef} is referred to as the Gauss-Markov CREM if $\rvec{X}$ is the decorrelation of $\rvec{U}$ in~\eqref{def:decorrelation}, and $(\rvec{X}, \rvec{Y}) = (\rvec{X}, \rvec{Y}^\star)$ forms a RDF-achieving pair in $\rdf_{\rvec{X}}(n, d)$. Recall from~\eqref{eig:sigma_i} that $\rvec{X}\sim\mathcal{N}(0,\mathsf{\Sigma}_{\rvec{X}})$, where 
\begin{align}
\mathsf{\Sigma}_{\rvec{X}} = \mathrm{diag}(\sigma_1^2,\ldots, \sigma_n^2), \label{GMCREM:X}
\end{align}
and $\sigma_i^2$'s are given by~\eqref{eig:sigma_i}. Recall from~\eqref{Ystar} that $\rvec{Y}^\star\sim\mathcal{N}(0,\mathsf{\Sigma}_{\rvec{Y}^\star})$, where
\begin{align}
\mathsf{\Sigma}_{\rvec{Y}^\star} = \mathrm{diag}(\nu_1^2,\ldots, \nu_n^2), \label{GMCREM:Y}
\end{align}
and we denote $\nu_i^2$ as
\begin{align}
\nu_i^2 \triangleq \max\left(0, \sigma_i^2 - \theta_n\right). \label{nu_i_app}
\end{align}
And $\theta_n > 0$ is the water level matched to $d$ via the $n$-th order reverse waterfilling~\eqref{eqn:para_d}. From~\eqref{rel:BA}, we have
\begin{align}
\mathbb{R}(\rvec{X}, \rvec{Y}^\star, d) &= \rdf_{\rvec{U}}(n, d),
\label{eqn:equivalence}
\end{align}
and $\rdf_{\rvec{U}}(n, d)$ is given by~\eqref{eqn:para_r_n}. Lemma~\ref{lemma:GaussianCREM} is also applicable to the special case of the Gauss-Markov CREM. Furthermore, the next lemma characterizes the negative slope in the Gauss-Markov CREM.

\begin{lemma}
In the Gauss-Markov CREM, for any $d\in\LRB{0,\dmax}$ and $n\in\mathbb{N}$, let $\theta_n > 0$ be the water level matched to $d$ via the $n$-th order reverse waterfilling~\eqref{eqn:para_d}, the negative slope $\lambda^\star$ defined in~\eqref{lambda_double_star} satisfies
\begin{align}
\lambda^\star  = \frac{1}{2\theta_n}.
\label{para:slope}
\end{align} 
\label{lemma:slope_properties}
\end{lemma}

\begin{proof}
We directly compute the negative slope using the parametric representation~\eqref{eqn:para_r_n} and~\eqref{eqn:para_d}. Taking the derivative with respect to $d$ on both sides of~\eqref{eqn:para_r_n} yields
\begin{align}
\lambda^{\star} = \frac{1}{n} \sum_{i = 1}^n 
\frac{1}{2\theta_n}\frac{d\theta_n}{dd} \mathbbm{1} \left\{\sigma_i^2 > \theta_n\right\}.\label{slope_prop_step1}
\end{align}
Differentiating~(\ref{eqn:para_d}), we obtain
\begin{align}
\frac{dd}{d\theta_n} = \frac{1}{n} \sum_{i = 1}^n\mathbbm{1}\left\{\sigma_i^2 > \theta_n\right\}, \label{slope_prop_step2}
\end{align}
which is independent of $i$. Plugging~(\ref{slope_prop_step2}) into~(\ref{slope_prop_step1}) yields~\eqref{para:slope}.

To justify the formal differentiation in~\eqref{slope_prop_step2}, observe using~\eqref{eqn:para_d} that $d$ is a continuous piecewise linear function of $\theta_n$, and $d$ is differentiable with respect to $\theta_n$ except at the $n$ points: $\theta_n = \sigma_i^2, ~i\in [n]$. The above proof goes through as long as the derivatives at those $n$ points are understood as the left derivatives. Indeed, $\rdf_{\rvec{U}}(n,d)$ is differentiable w.r.t. $d$ for any $d\in (0,\dmax)$, e.g.~\cite[Eq. (16)]{kostina2012fixed}.
\end{proof}

\subsubsection{Sensitivity of the negative slope}
\label{subsec:sen_crem}
The following theorem is a perturbation result, which bounds the change in the negative slope when the variances of the input $\rvec{X}$ to $\rdf\left(\rvec{X}, \rvec{Y}^\star, d\right)$ are perturbed. It is related to lossy compression using mismatched codebook: the codewords are drawn randomly according to the distribution $P_{\rvec{Y}^\star}$ while the source distribution is $\hat{\rvec{X}}$ instead of $\rvec{X}$.
\begin{lemma}
Let $\rvec{X}$ be the decorrelation of $\rvec{U}$ in~\eqref{def:decorrelation}, and let $(\rvec{X}, \rvec{Y}^\star)$ be a RDF-achieving pair in $\rdf_{\rvec{X}}(n, d)$ (recall~\eqref{eqn:nthorderOp}). For any fixed distortion $d\in\LRB{0,\dmax}$, let $\theta>0$ be the water level matched to $d$ via the limiting reverse waterfilling in~\eqref{eqn:para_d_inf}. For any $t\in (0,\theta / 3)$, let $\hat{\sigma}_i^2$'s be such that 
\begin{align}
|\hat{\sigma}_i^2 - \sigma_i^2 |\leq t, \quad\forall~i\in [n].
\label{diff:sigma_hat}
\end{align}
Let the Gaussian random vector $\hat{\rvec{X}}$ be $\hat{\rvec{X}}\sim\mathcal{N}(\rvec{0}, \mathrm{diag}(\hat{\sigma}_1^2, \ldots, \hat{\sigma}_n^2))$, and let $\hat{\lambda}^\star$ be the negative slope of $\rdf (\hat{\rvec{X}}, \rvec{Y}^\star, d )$. Then, for all $n$ large enough, the negative slope $\hat{\lambda}^\star$ satisfies
\begin{align}
|\lambda^\star - \hat{\lambda}^\star|\leq \frac{9t}{4\theta^2},
\label{eqn:bound_slope}
\end{align}
where $\lambda^\star = -\rdf'\left(\rvec{X}, \rvec{Y}^\star, d\right)$ is given by~\eqref{para:slope}.
\label{lemma:bound_slope}
\end{lemma}

\begin{proof}
Consider the Gaussian CREM $\mathbb{R}(\hat{\rvec{X}}, \rvec{Y}^\star,d)$. Let $\hat{\theta}_n > 0$ be the water level matched to $d$ via the $n$-th order reverse waterfilling~\eqref{eqn:para_d} over $\hat{\sigma}_i^2$'s, and let $\theta_n > 0$ be the water level matched to $d$ via the $n$-th order reverse waterfilling~\eqref{eqn:para_d} over $\sigma_i^2$'s. In~\eqref{eqn:expression}, replacing $(\rvec{X},\rvec{Y})$ by $(\hat{\rvec{X}},\rvec{Y}^\star)$, and then taking the derivative with respect to $d$ on both sides yields 
\begin{align}
-\hat{\lambda}^\star & = -\hat{\lambda}^\star + \frac{1}{2n}\sum_{i:\hat{\sigma}_i^2>\hat{\theta}_n} \frac{-2\hat{\lambda}^\star}{1 + 2\hat{\lambda}^\star\nu_i^2}\frac{d\hat{\theta}_n}{dd} + \notag\\
&\quad  \frac{1}{n}\sum_{i:\hat{\sigma}_i^2>\hat{\theta}_n} \frac{2 \hat{\sigma}_i^2\hat{\lambda}^{\star 2}}{(1+2\hat{\lambda}^\star\nu_i^2)^2}\frac{d\hat{\theta}_n}{dd},
\end{align}
where $\nu_i^2$'s are defined in~\eqref{nu_i_app}. Rearranging terms yields
\begin{align}
\hat{\lambda}^{\star} =  \sum_{i:\hat{\sigma}_i^2>\hat{\theta}_n} \frac{1}{(1+2\hat{\lambda}^\star\nu_i^2)^2}  ~\Bigg /  \sum_{i:\hat{\sigma}_i^2>\hat{\theta}_n}  \frac{2 (\hat{\sigma}_i^2 - \sigma_i^2 + \theta_n)}{(1+2\hat{\lambda}^\star\nu_i^2)^2}.
\label{eqn:slope_hat}
\end{align}
Substituting the bound~\eqref{diff:sigma_hat} into~\eqref{eqn:slope_hat}, we obtain
\begin{align}
\hat{\lambda}^{\star} \in \left[\frac{1}{2(\theta_n+t)}, \frac{1}{2(\theta_n-t)}\right].
\label{lambdahat_range}
\end{align}  
Since $\lim_{n\rightarrow\infty} \theta_n = \theta$, for all $n$ large enough, we have 
\begin{align}
\frac{2\theta}{3} \leq  \theta_n \leq \frac{4\theta}{3}.
\label{bound:theta_n_13}
\end{align}
Since $t \in (0,  \theta /3)$,~\eqref{bound:theta_n_13} implies that $0 < t < \theta_n /2 $. From~\eqref{para:slope},~\eqref{lambdahat_range} and~\eqref{bound:theta_n_13}, we see that
\begin{align}
& |\lambda^{\star} - \hat{\lambda}^\star| \notag \\ 
\leq & \max\left\{\lrabs{ \frac{1}{2(\theta_n+t)} - \frac{1}{2\theta_n}}, \lrabs{\frac{1}{2(\theta_n-t)} - \frac{1}{2\theta_n}}\right\}  \\ 
\leq & \frac{t}{\theta_n^2}\\
\leq & \frac{9t}{4\theta^2}.
\end{align}
\end{proof}

\subsection{Proof of Theorem~\ref{thm:limiting_variance}}
\label{app:derivation_vd}
Theorem~\ref{thm:limiting_variance} is a direct consequence of the following lemma.
\begin{lemma}[Parametric representation for the $\mathsf{d}$-tilted information]
Let $\rvec{X}$ be the decorrelation of $\rvec{U}$~\eqref{def:decorrelation}, and let $(\rvec{X}, \rvec{Y}^\star)$ be a RDF-achieving pair in $\rdf_{\rvec{X}}\left(n,d\right)$. For any $d\in \LRB{0,\dmax}$, let $\theta_n > 0$ be the water level matched to $d$ via the $n$-th order reverse waterfilling~\eqref{eqn:para_d} over $\sigma_i^2,~i\in [n]$. Then, for all $\rvec{x}\in\mathbb{R}^n$,
\begin{align}
& \Lambda_{Y^\star_i}(x_i, \lambda^{\star}, \min(\theta_n, \sigma_i^2))  \notag \\
= & \frac{\min(\theta_n,\sigma_i^2)}{2\theta_n}\left(\frac{x_i^2}{\sigma_i^2} - 1\right) + \frac{1}{2}\log\frac{\max(\theta_n,\sigma_i^2)}{\theta_n},
\label{para:tilted_2}
\end{align}
where $\lambda^{\star}$ defined in~\eqref{lambda_double_star} is given by~\eqref{para:slope}.
\label{lemma:para_tilted}
\end{lemma}

\begin{proof}
The proof relies on the Gaussianity of $\rvec{Y}^\star$. For each $i\in [n]$, from~\eqref{def:generalizeddtilted} and~\eqref{para:slope}, we have 
\begin{align}
& \Lambda_{Y^\star_i}(x_i, \lambda^{\star}, \min(\theta_n, \sigma_i^2)) \notag \\
= & -\frac{\min(\theta_n, \sigma_i^2)}{2\theta_n} -\log\mathbb{E}\left[\exp\lrbb{- \lambda^{\star}\LRB{Y^\star_i - x_i}^2}\right]. \label{eqn:unified}
\end{align}
Substituting $Y_i^\star \equiv 0$ a.s. for all $i$ such that $\sigma_i^2 \leq \theta_n$ (recall~\eqref{GMCREM:Y}) into~\eqref{eqn:unified}, we obtain
\begin{align}
\Lambda_{Y^\star_i}(x_i, \lambda^{\star}, \min(\theta_n, \sigma_i^2)) = \frac{x_i^2 - \sigma_i^2}{2\theta_n}. \label{step:low}
\end{align}
Substituting $Y^\star_i\sim\mathcal{N}(0,\sigma_i^2 - \theta_n)$ for all $i$ such that $\sigma_i^2>\theta_n$ (recall~\eqref{GMCREM:Y}) into~\eqref{eqn:unified} and applying the formula for the moment generating function of a noncentral $\chi^2$-distribution with one degree of freedom, we obtain
\begin{align}
\Lambda_{Y^\star_i}(x_i, \lambda^{\star}, \min(\theta_n, \sigma_i^2)) = \frac{1}{2} \LRB{\frac{x_i^2}{\sigma_i^2} - 1} + \frac{1}{2}\log\frac{\sigma_i^2}{\theta_n}.\label{step:para_tilted}
\end{align}
Unifying~\eqref{step:low} and~\eqref{step:para_tilted}, we obtain~\eqref{para:tilted_2}. 
\end{proof}

\begin{proof}[Proof of Theorem~\ref{thm:limiting_variance}]
For any fixed distortion $d\in\LRB{0,\dmax}$, let $\theta > 0$ be the water level matched to $d$ via the limiting reverse waterfilling~\eqref{eqn:para_d_inf}. By the independence of $Y^\star_1,\ldots, Y^\star_n$ and~\eqref{eqn:para_d}, we have for any $\rvec{x}$, 
\begin{align}
\jmath_{\rvec{X}}(\rvec{x}, d) = \sum_{i = 1}^n \Lambda_{Y^\star_i}(x_i, \lambda^{\star}, \min(\theta_n, \sigma_i^2)), \label{para:tilted_1}
\end{align}
where $\lambda^\star = -\rdf'_{\rvec{X}}(n, d)$. Taking the expectation and the variance of~\eqref{para:tilted_1} using~\eqref{para:tilted_2} yields\footnote{The result on expectations was implicitly established by Gray~\cite{gray1970information}, which we recover here. The result on variances is new.}
\begin{align}
\mathbb{E}[\jmath_{\rvec{X}}(\rvec{X}, d) ] &= \sum_{i = 1}^n  \frac{1}{2}\max\left(0, ~\log\frac{\sigma_i^2}{\theta_n}\right), \label{tilted_mean} \\
\var{\jmath_{\rvec{X}}(\rvec{X}, d) } &= \sum_{i = 1}^n \frac{1}{2}\min\left(1,~\left(\frac{\sigma_i^2}{\theta_n}\right)^2\right).
\label{tilted_variance} 
\end{align}
An application of Theorem~\ref{thm:LimitingThm} to~\eqref{tilted_mean} on the function $t\mapsto \frac{1}{2}\max\left(0,\log\frac{t}{\theta}\right)$ yields~\eqref{limiting_expectation_relation}. Similarly, an application of Theorem~\ref{thm:LimitingThm} to~\eqref{tilted_variance} on the function $t\mapsto \frac{1}{2}\min\left[1, \left(\frac{t}{\theta}\right)^2\right]$ yields~\eqref{limiting_var}.
\end{proof}

\section{Proofs in Section~\ref{sec:converse}}
\label{app:pf_converse}

\subsection{Proof of Theorem~\ref{thm:geometric}}
\label{app:proofGeoBd}

\begin{proof}
The result follows from a geometric argument, illustrated in Fig.~\ref{fig:converse}. Let $\mathcal{C}\subset \mathbb{R}^n$ be the set of codewords of an arbitrary $(n,M,d,\epsilon)$ code, and $\mathcal{B}(\rvec{c},d) $ be the distortion $d$-ball centered at a codeword $\rvec{c}\in \mathcal{C}$ (recall~\eqref{def:d_ball}). By the definition of an $(n,M,d,\epsilon)$ code, we know that the union of the distortion $d$-balls centered at codewords in $\mathcal{C}$ has probability mass at least $1-\epsilon$:
\begin{align}
\mathbb{P}\left[\rvec{U} \in \mathcal{B}\right] \geq 1-\epsilon, \label{excess_code_1}
\end{align}
where $\mathcal{B}$ denotes the union of the distortion $d$-balls centered at the codewords in $\mathcal{C}$:
\begin{align}
\mathcal{B} \triangleq \bigcup_{\rvec{c}\in\mathcal{C}}\mathcal{B}(\rvec{c},d). 
\end{align}
For a set $\mathcal{S}\subseteq \mathbb{R}^n$, denote by 
\begin{align}
\mathsf{A}\mathcal{S} \triangleq \left\{\mathsf{A}\rvec{s}: \rvec{s}\in\mathcal{S}\right\}
\end{align}
the linear transformation of $\mathcal{S}$ by the matrix $\mathsf{A}$. Recall from~\eqref{def:A} that $\mathsf{A}$ is invertible and the innovation is $\rvec{Z} = \mathsf{A}\rvec{U}$. Changing variable $\rvec{U} = \mathsf{A}^{-1}\rvec{Z}$ in~\eqref{excess_code_1} yields
\begin{align}
\mathbb{P}\left[\rvec{Z} \in \mathsf{A}\mathcal{B}\right] \geq 1-\epsilon.\label{excess_code_2}
\end{align}   

Next, we give a geometric interpretation of the set $\mathsf{A}\mathcal{B}$. Consider the set $\mathsf{A}\mathcal{C}$, that is, the transformation of the codebook $\mathcal{C}$ by $\mathsf{A}$. For any $\rvec{x}\in\mathbb{R}^n$, notice that the set 
\begin{align}
\mathsf{A}\mathcal{B}(\mathsf{A}^{-1}\rvec{x},d)  = \left\{\rvec{x}'\in\mathbb{R}^n : (\rvec{x}'-\rvec{x})^\top(\mathsf{AA}^\top)^{-1} (\rvec{x}'-\rvec{x}) \leq nd \right\}
\end{align}
is the set of points bounded by the ellipsoid centered at $\rvec{x}$ with principal axes being the eigenvectors of $\mathsf{AA}^\top$. It follows that 
\begin{align}
\mathsf{A}\mathcal{B} &= \mathsf{A} \bigcup_{\rvec{c}\in\mathcal{C}}\mathcal{B}(\rvec{c},d)\\
&= \bigcup_{\rvec{c}'\in \mathsf{A}\mathcal{C}}\mathsf{A}\mathcal{B}(\mathsf{A}^{-1}\rvec{c}',d),
\end{align}
i.e., $\mathsf{A}\mathcal{B}$ is the union of ellipsoids centered at transformed codewords. See Fig.~\ref{fig:mouse} for an illustration of the set $\mathsf{A}\mathcal{B}$.

Finally, the following volumetric argument completes the proof of Theorem~\ref{thm:geometric}. Since the volume of a union of sets is less than or equal to the sum of the sets' volumes, we have
\begin{align}
M \geq \frac{\text{Vol}(\mathsf{A}\mathcal{B})}{\text{Vol}(\mathsf{A}\mathcal{B}(\rvec{0}, d))}.
\label{bound:union_bd}
\end{align}
Moreover, $\text{Vol}(\mathsf{A}\mathcal{B}(\rvec{0}, d)) = \text{Vol}(\mathcal{B}(\rvec{0}, d))$ due to $\det \mathsf{A} = 1$. On the other hand, due to the spherical symmetry of the distribution of $\rvec{Z}$, the ball $\mathcal{B}(\rvec{0}, r(n,\epsilon))$, where $ r(n,\epsilon)$ satisfies~\eqref{r_n_epsilon}, has the smallest volume among all sets in $\mathbb{R}^n$ with probability greater than or equal to $1-\epsilon$, and so
\begin{align}
\text{Vol}(\mathsf{A}\mathcal{B}) \geq \text{Vol}(\mathcal{B}(\rvec{0}, r(n,\epsilon))).
\end{align} 
Therefore, we can weaken \eqref{bound:union_bd} as 
\begin{align}
M \geq \frac{\text{Vol}(\mathcal{B}(\rvec{0}, r(n,\epsilon)))}{\text{Vol}(\mathcal{B}(\rvec{0}, d))} = \left(\frac{r(n,\epsilon)}{d}\right)^{n/2}.
\label{bound:vol_bd}
\end{align}

\begin{figure*}[!htbp]
    \centering
    \begin{subfigure}[t]{0.31\textwidth}
        \includegraphics[width=\textwidth]{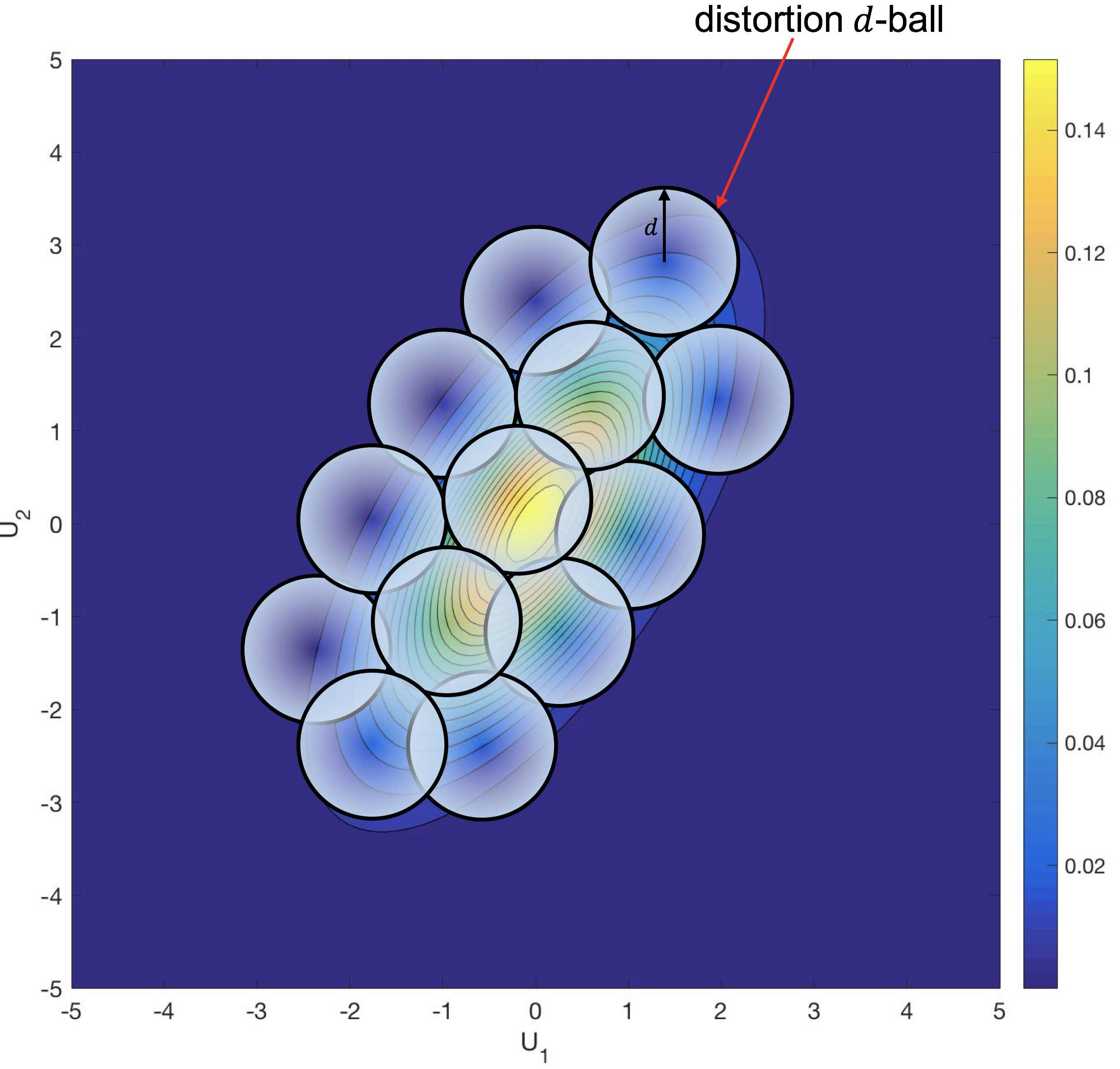}
        \caption{The $\rvec{U}$ space: Given $d,\epsilon$, the goal is to cover at least $1-\epsilon$ probability mass, under distribution of $\rvec{U}$, using the least number of distortion $d$-balls.}
        \label{fig:gull}
    \end{subfigure}
    ~ 
    \begin{subfigure}[t]{0.31\textwidth}
        \includegraphics[width=\textwidth]{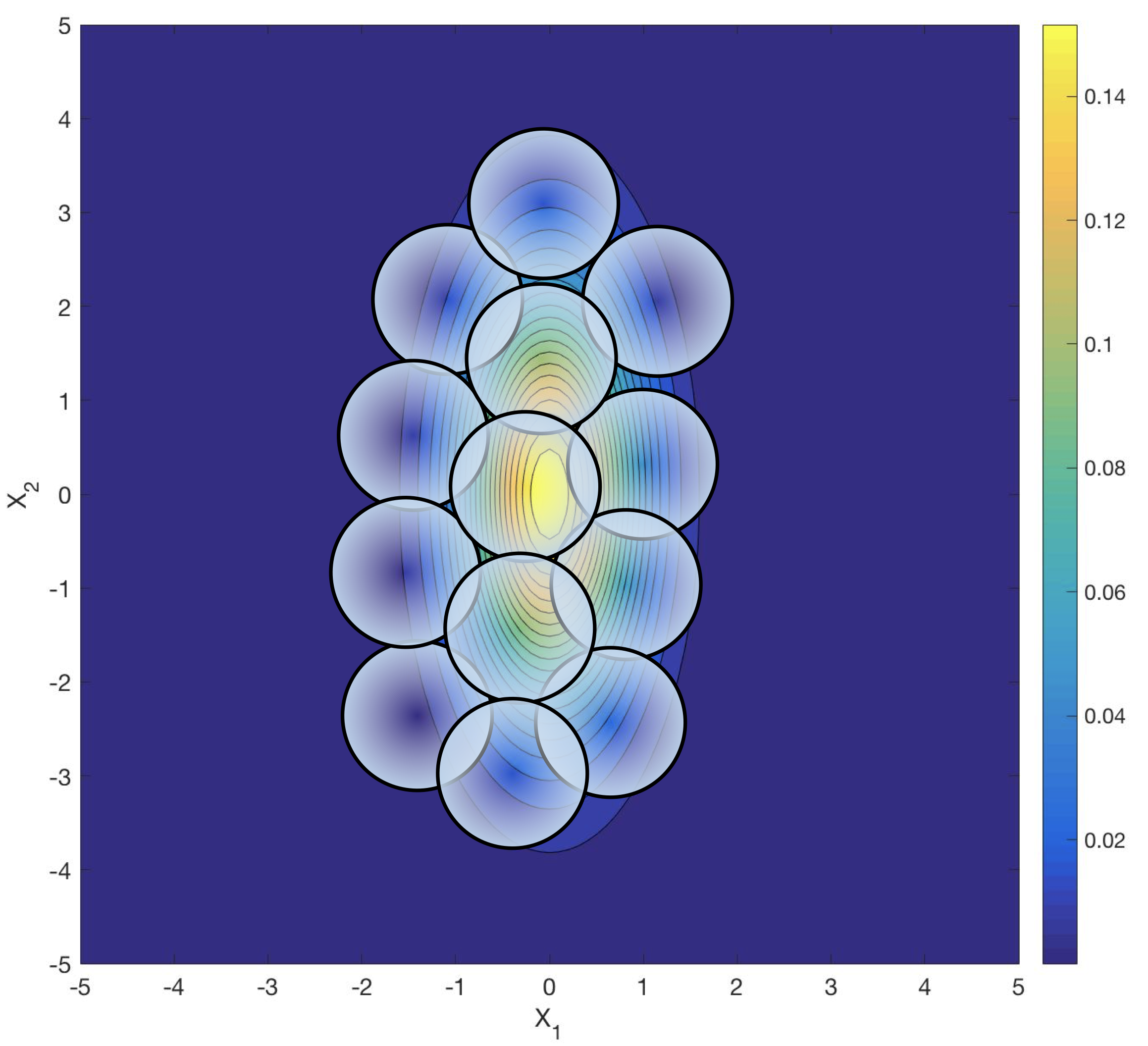}
        \caption{The $\rvec{X}$ space is simply a unitary transformation of the $\rvec{U}$ space.}
        \label{fig:tiger}
    \end{subfigure}
    ~ 
    \begin{subfigure}[t]{0.31\textwidth}
        \includegraphics[width=\textwidth]{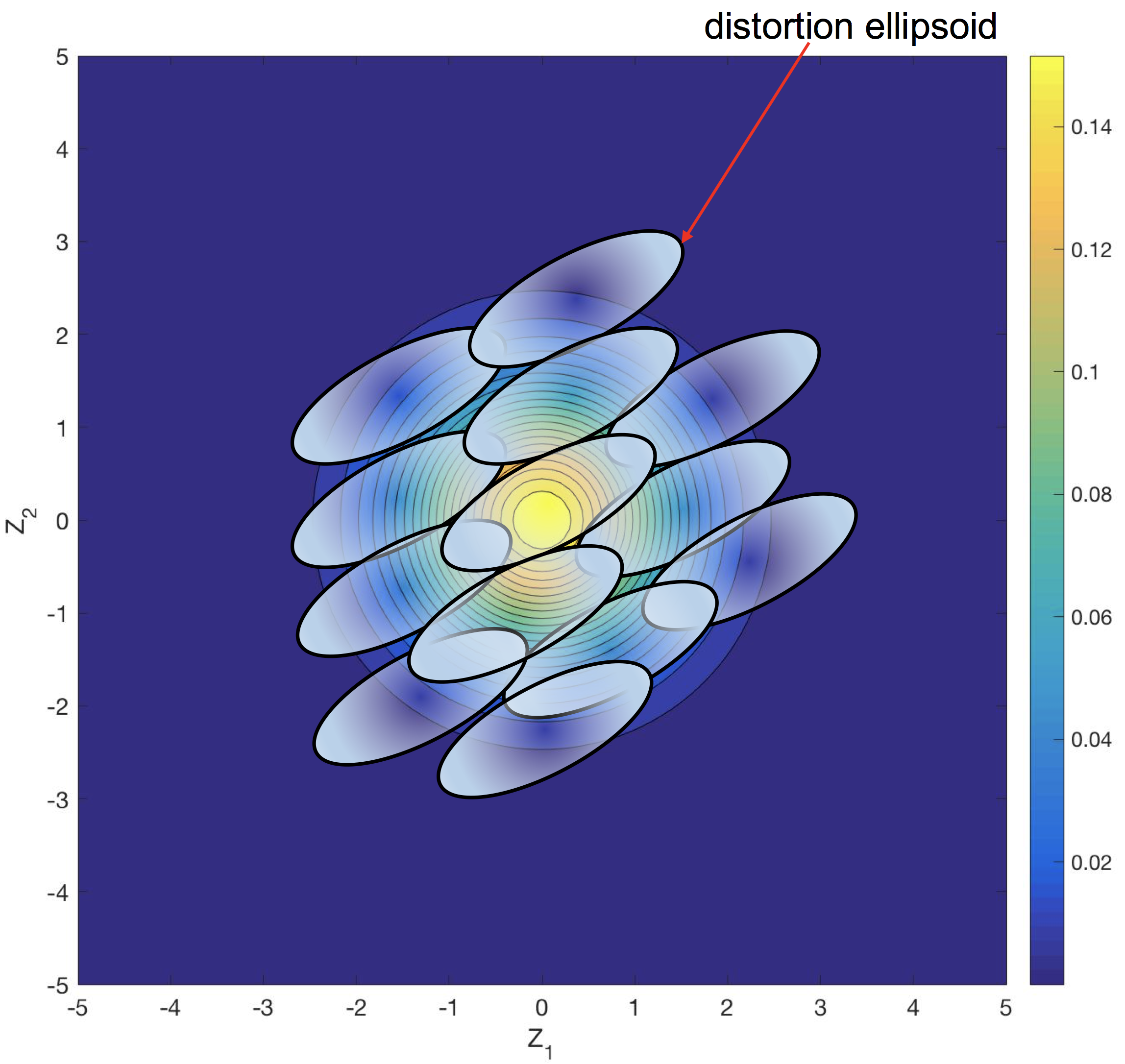}
        \caption{The $\rvec{Z}$ space: Given $d,\epsilon$, the goal is to cover at least $1-\epsilon$ probability mass, under distribution of $\rvec{Z}$, using the least number of distortion ellipsoids, each of which has the same volume as the distortion $d$-ball since $\det(\mathsf{AA}^\top) = 1$.}
        \label{fig:mouse}
    \end{subfigure}
    \caption{Converse proof in figures. The contour plot in each figure shows the underlying probability distribution.}
    \label{fig:converse}
\end{figure*}
\end{proof}

\subsection{Proof of Theorem~\ref{thm:tilted_info_concentrate}}
\label{app:tilted_info_concentrate}
\begin{proof}
The proof is based on Chebyshev's inequality. Fix $d\in\LRB{0,\dmax}$. For each fixed $n\in\mathbb{N}$, let $\theta_1,\theta_2>0$ be the water levels matched to $d$ and $d_n$, respectively, in the $n$-th order reverse waterfilling~\eqref{eqn:para_d} over $\sigma_i^2,~i\in [n]$, that is, 
\begin{align}
d &= \frac{1}{n}\sum_{i = 1}^n \min\LRB{\theta_1, \sigma_i^2}, \label{d_theta_1} \\
d_n &= \frac{1}{n}\sum_{i = 1}^n \min\LRB{\theta_2, \sigma_i^2}.
\end{align}
Obviously, both $\theta_1$ and $\theta_2$ depend on $n$. We now proceed to show that there exists a constant $h_2>0$ such that for all $n$ large enough,
\begin{align}
\lrabs{\theta_1 - \theta_2}\leq \frac{h_2}{n}.
\label{diff:theta}
\end{align}
Indeed, without loss of generality, assume $d<d_n$\footnote{Otherwise, switch $\theta_1$ and $\theta_2$ in the rest of the proof.}, then $\theta_1 < \theta_2$ by the mononicity of the reverse waterfilling~\eqref{eqn:para_d}. Define the following index sets 
\begin{align}
I_1 &\triangleq \lrbb{i\in [n]: \sigma_i^2 \leq \theta_1},\\
I_2 &\triangleq \lrbb{i\in [n]: \theta_1 < \sigma_i^2 < \theta_2},\\
I_3 &\triangleq \lrbb{i\in [n]: \theta_2 \leq \sigma_i^2}.
\end{align}
Then,
\begin{align}
d_n-d &= \frac{1}{n}\sum_{i = 1}^n \LRB{\min\LRB{\theta_2,\sigma_i^2} - \min\LRB{\theta_1,\sigma_i^2} } \\
& = \frac{1}{n} \sum_{i\in I_1} 0 + \frac{1}{n} \sum_{i\in I_2} \LRB{\sigma_i^2 - \theta_1} +  \frac{1}{n} \sum_{i\in I_3} \LRB{\theta_2 - \theta_1} \\
& \geq \frac{|I_3|}{n} \LRB{\theta_2 - \theta_1}.
\end{align}
Since $d_n < \dmax$, there exists a constant $\tilde{c}\in (0,1)$ such that for all $n$ large enough, $|I_3| \geq \tilde{c} n$, hence~\eqref{diff:theta} holds with $h_2 =  \frac{h_1}{\tilde{c}}$. 

Now, let $G_1, \ldots, G_n$ be i.i.d. $\mathcal{N}(0,1)$. To simplify notations, we denote the random variable as
\begin{align}
\Delta(d,d_n) \triangleq  \jmath_{\rvec{X}}\LRB{\rvec{X},d} - \jmath_{\rvec{X}}\LRB{\rvec{X},d_n}.
\end{align}
From~\eqref{para:tilted_1} and~\eqref{para:tilted_2}, we have 
\begin{align}
\Delta(d,d_n) &=  \sum_{i = 1}^n \Bigg \{\left[\frac{\min\LRB{\theta_1, \sigma_i^2}}{2\theta_1} - \frac{\min\LRB{\theta_2, \sigma_i^2}}{2\theta_2}\right] (G_i^2 - 1) \notag \\
& + \frac{1}{2}\log\left[\frac{\max\LRB{\theta_1, \sigma_i^2}}{2\theta_1}\cdot \frac{2\theta_2}{\max\LRB{\theta_2, \sigma_i^2}}\right]\Bigg\}.
\end{align}
To apply Chebyshev's inequality, we bound the mean and the variance of $\Delta(d,d_n)$ as follows. 
\begin{align}
& \mathbb{E}\left[\Delta(d,d_n)\right] \notag \\
= & \sum_{i = 1}^n \frac{1}{2}\log\left[\frac{\max\LRB{\theta_1, \sigma_i^2}}{2\theta_1}\cdot \frac{2\theta_2}{\max\LRB{\theta_2, \sigma_i^2}}\right]  \\
=& \sum_{i\in I_1} 0 + \sum_{i\in I_2} \frac{1}{2}\log\frac{\sigma_i^2}{\theta_1}+  \sum_{i\in I_3}\frac{1}{2}\log\frac{\theta_2}{\theta_1} \\
\leq & \frac{h_2}{\theta_1},\label{eqn:I3}
\end{align}
where~\eqref{eqn:I3} holds since for $i\in I_2$, we have $\frac{1}{2}\log\frac{\sigma_i^2}{\theta_1} \leq \frac{1}{2}\log\frac{\theta_2}{\theta_1}$, while for $i\in I_3$, due to~\eqref{diff:theta}, we have 
\begin{align}
\frac{1}{2}\log\frac{\theta_2}{\theta_1} \leq \frac{1}{2}\log\left(1+\frac{h_2}{n\theta_1}\right) \leq \frac{h_2}{2n\theta_1}.
\end{align}
By a similar argument, we can bound the variance as 
\begin{align}
\var{\Delta(d,d_n)}\leq \frac{h_2^2}{\theta_1^2n}.
\label{chebyshev:var_bound}
\end{align}
In conjunction with~\eqref{eqn:I3},~\eqref{chebyshev:var_bound}, Chebyshev's inequality yields that for all $n$ large enough and $\forall\ell > 0$,
\begin{align}
\prob{\lrabs{\Delta(d,d_n) - \mathbb{E}\left[\Delta(d,d_n)\right]} \geq\ell }\leq \frac{h_2^2}{\theta_1^2n\ell^2}.
\label{eqn:chebyshev}
\end{align}
Choosing $\ell = \frac{u h_2}{\theta_1}$ in~\eqref{eqn:chebyshev} and applying~\eqref{eqn:I3} yields that $\forall u > 0$,
\begin{align}
\prob{\lrabs{\Delta(d,d_n)} \geq \frac{(1+u)h_2}{\theta_1} }\leq \frac{1}{nu^2}.
\end{align}
Let $\theta > 0$ be the water level matched to $d$ via the limiting reverse waterfilling~\eqref{eqn:para_d_inf}, then $\lim_{n\rightarrow\infty} \theta_1 = \theta$ by~\eqref{d_theta_1} and~\eqref{eqn:para_d_inf}. Therefore, we have $\theta_1\geq \frac{\theta}{2}$ for all $n$ large enough. Hence,  for all $n$ large enough and $\forall u > 0$, we have 
\begin{align}
\prob{\lrabs{\Delta(d,d_n)} \geq \frac{2(1+u)h_2}{\theta} }\leq \frac{1}{nu^2}. \label{eqn:concen_app_pf}
\end{align}
Rearranging terms in~\eqref{eqn:concen_app_pf} completes the proof.
\end{proof}

\section{Proofs in Section~\ref{sec:achievability}}
\label{app:pf_achievability}

\subsection{Proof of Lemma~\ref{lemma:lemma_2}}
\label{sec:main_lemma}
In addition to new concentration inequalities, shown in Lemma~\ref{lemma:lemma_5} and Lemma~\ref{lemma:lemma_4} below, the proof leverages the following bound, which is a direct application of~\cite[Lem. 1]{kostina2012fixed} to the random vector $\rvec{X}$.
\begin{lemma}[Lower bound on probability of distortion balls]
Fix $d\in\LRB{0,\dmax}$,  $n\in\mathbb{N}$, and the distribution $P_{\rvec{Y}}$ on $\mathbb{R}^n$. Then for any $\rvec{x}\in \mathbb{R}^n$,  it holds that
\begin{align}
P_{\rvec{Y}}\left(\mathcal{B}(\rvec{x},d)\right) &\geq \sup_{ P_{\hat{\rvec{X}}}, \gamma > 0} \exp\left\{-\hat{\lambda}^\star n \gamma - \Lambda_{\rvec{Y}}(\rvec{x}, \hat{\lambda}^\star, d)\right\} \times \notag \\
&\quad \prob{d-\gamma \leq \dis{\rvec{x}}{\hat{\rvec{F}}^\star}\leq d ~| \hat{\rvec{X}} = \rvec{x}}, \label{eqn:lower_bound}
\end{align}
where the supremum is over all pdfs $P_{\hat{\rvec{X}}}$ on $\mathbb{R}^n$; $\Lambda_{\rvec{Y}}(\rvec{x}, \hat{\lambda}^\star, d)$ is the generalized tilted information defined in~\eqref{def:generalizeddtilted} with 
\begin{align}
\hat{\lambda}^\star =-\rdf'(\hat{\rvec{X}}, \rvec{Y}, d);
\label{eqn:lambda_hat_star}
\end{align}
and the random variable $\hat{\rvec{F}}^\star$ achieves $\rdf(\hat{\rvec{X}}, \rvec{Y}, d)$.
\label{lemma:lowerbound}
\end{lemma}

The high-level idea in proving Lemma~\ref{lemma:lemma_2} is the following. In Lemma~\ref{lemma:lowerbound}, we replace $\rvec{Y}$ by $\rvec{Y}^\star$ defined in~\eqref{Ystar} and~\eqref{GMCREM:Y}, and we choose $\hat{\rvec{X}}$ to be the proxy Gaussian random variable $\hat{\rvec{X}}(\rvec{x})$ defined in~\eqref{eqn:Xhat}. With such choices of $\hat{\rvec{X}}$ and $\rvec{Y}$, the next two lemmas provide further lower bounds on the two factors on the right side of~\eqref{eqn:lower_bound}. The first one is a concentration inequality on the generalized tilted information.
\begin{lemma}
For any fixed $d\in (0,\dmax)$ and excess-distortion probability $\epsilon\in (0,1)$, there exist constants $C$ and $C_2 >0$ such that for all $n$ large enough,
\begin{align}
\mathbb{P}\left[ \Lambda_{\rvec{Y}^\star}(\rvec{X}, \hat{\lambda}^{\star}(\rvec{X}), d) \leq \Lambda_{\rvec{Y}^\star}(\rvec{X}, \lambda^{\star}, d)  + C\log n \right] \geq 1 - \frac{C_2}{\sqrt{n}},
\label{con_gen_inf}
\end{align}
where $\hat{\lambda}^{\star}(\rvec{x})$ is given by~\eqref{lambda_x_hat_star} with $\hat{\rvec{X}}$ defined in~\eqref{eqn:Xhat}, and $\lambda^{\star}$ is in~\eqref{para:slope}.
\label{lemma:lemma_5}
\end{lemma}
\begin{proof}
Appendix~\ref{app:proof_lemma_5}.
\end{proof}

The second bound, presented in Lemma~\ref{lemma:lemma_4} below, is referred to as the shell-probability lower bound. For any $\rvec{x}\in\mathbb{R}^n$ and any $\gamma\in (0,d)$, define the shell 
\begin{align}
\mathcal{S}(\rvec{x}, d, \gamma) \triangleq \lrbb{\rvec{x}'\in\mathbb{R}^n: d-\gamma\leq \dis{\rvec{x}}{\rvec{x}'}\leq d}.
\end{align}
Geometrically, Lemma~\ref{lemma:lowerbound} provides a quantitative connection between the probability of a distortion $d$-ball and the probability of its shell, and Lemma~\ref{lemma:lemma_4} below gives a lower bound on the probability of the shell $\mathcal{S}(\rvec{x}, d, \gamma)$  for ``typical'' sequences $\rvec{x}$.

\begin{lemma}[Shell-probability lower bound]
Fix any distortion $d\in (0,\dmax)$ and any excess-distortion probability $\epsilon\in (0,1)$. For any constant $\alpha > 0$ and any $n\in \mathbb{N}$, consider the set $\TS(n,\alpha, p)$ defined in~Definition~\ref{def:TS}, where $p$ is the constant in Theorem~\ref{thm:typicalset}. Let 
\begin{align}
\gamma \triangleq \frac{\log^q n}{n} \label{choice:gamma}, 
\end{align}
where $q>1$ is a constant defined in~\eqref{defq} in Appendix~\ref{app:proof_lemma_4} below. Then, there exists a constant $C_1 > 0 $ such that for all $n$ large enough, for any $\rvec{u}\in \TS(n,\alpha, p)$ and $\rvec{x} = \mathsf{S}^\top\rvec{u}$ with $\mathsf{S}$ in~\eqref{diag:B}, it holds that
\begin{align}
\prob{\hat{\rvec{F}}^\star\in \mathcal{S}(\rvec{x}, d, \gamma) ~| \hat{\rvec{X}} = \rvec{x}} \geq \frac{C_1}{\sqrt{n}},
\label{shell_lower_bd}
\end{align}
where $\hat{\rvec{X}}$ is given in~\eqref{eqn:Xhat}.
\label{lemma:lemma_4}
\end{lemma}

\begin{proof}
Appendix~\ref{app:proof_lemma_4}.
\end{proof}

We now present the proof of Lemma~\ref{lemma:lemma_2}.
\begin{proof}[Proof of Lemma~\ref{lemma:lemma_2}]
Let $\rvec{X}$ be the decorrelation of $\rvec{U}$ in~\eqref{def:decorrelation}. Replace $\rvec{Y}$ by $\rvec{Y}^\star$ in Lemma~\ref{lemma:lowerbound}. Let $\TS(n,\alpha, p)$ be the set defined in Definition~\ref{def:TS}, and let $p$ be the constant in Theorem~\ref{thm:typicalset}. Let $C, C_1, C_2, q$ be the constants in Lemmas~\ref{lemma:lemma_5} and~\ref{lemma:lemma_4}. Consider any $n$ that is large enough such that Theorem~\ref{thm:typicalset}, Lemma~\ref{lemma:lemma_5} and Lemma~\ref{lemma:lemma_4} hold. Let $\theta > 0$ be the water level matched to $d$ via the limiting reverse waterfilling~\eqref{eqn:para_d_inf}. Denote the event
\begin{align}
\mathcal{E} \triangleq \left\{\log \frac{1}{P_{\rvec{Y}^\star}\LRB{\mathcal{B}(\rvec{X}, d)} } > \jmath_{\rvec{X}}\LRB{\rvec{X},d} + \beta_1\log^q n + \beta_2\right\},
\end{align}
where $\beta_1$ and $\beta_2$ are constants defined by 
\begin{align}
\beta_1 &\triangleq \frac{1}{2\theta} + C_{d} + \frac{1}{2} + C, \label{choice:beta_1}\\
\beta_2 &\triangleq -\log C_1,\label{choice:beta_2}
\end{align}
and $C_d >0$ is a constant such that 
\begin{align}
\lrabs{\hat{\lambda}^\star(\rvec{x}) - \frac{1}{2\theta}}\leq C_{d}
\label{C_d_choice}
\end{align}
for any $\rvec{u}\in \TS(a,\alpha,p)$ and $\rvec{x} = \mathsf{S}^\top\rvec{u}$. The existence of such $C_d$ is guaranteed by~\eqref{eqn:error_lambda} in Theorem~\ref{thm:typicalset} and the fact that $\lim_{n\rightarrow\infty} \theta_n = \theta$.

Using elementary probability rules, we write  
\begin{align}
&\prob{\log \frac{1}{P_{\rvec{Y}^\star}\LRB{\mathcal{B}(\rvec{X}, d)} } > \jmath_{\rvec{X}}\LRB{\rvec{X},d} + \beta_1\log^q n +\beta_2} \notag\\
=& \prob{\mathcal{E},~\rvec{U}\in \TS(n,\alpha, p)} + \prob{\mathcal{E},~\rvec{U}\not\in \TS(n,\alpha, p)} \label{lemma2:step1}\\
\leq & \mathbb{P}\Bigg [ \rvec{U}\in \TS(n,\alpha, p),~\hat{\lambda}^\star(\rvec{X}) n \gamma + \Lambda_{\rvec{Y}^\star}(\rvec{X}, \hat{\lambda}^\star(\rvec{X}), d) - \notag \\
& \quad\quad \log\prob{\hat{\rvec{F}}^\star\in \mathcal{S}(\hat{\rvec{X}}, d, \gamma) ~| \hat{\rvec{X}} = \rvec{X}} > \notag \\
&\quad\quad\quad\quad \jmath_{\rvec{X}}\LRB{\rvec{X},d} + \beta_1\log^q n +\beta_2\Bigg ] \notag\\ 
&\quad\quad + \prob{\rvec{U}\not\in \TS(n,\alpha, p)} \label{lemma2:step2}\\
\leq & \mathbb{P}\Bigg [  \rvec{U}\in \TS(n,\alpha, p),~\Lambda_{\rvec{Y}^\star}(\rvec{X}, \hat{\lambda}^\star(\rvec{X}), d) - \jmath_{\rvec{X}}\LRB{\rvec{X},d} > \notag \\
&\quad\quad  -\hat{\lambda}^\star(\rvec{X}) n\gamma + \log \frac{C_1}{\sqrt{n}}  + \beta_1\log^q n + \beta_2 \Bigg ] \notag \\
& \quad\quad + \prob{\rvec{U}\not\in \TS(n,\alpha,p)} \label{lemma2:step3}\\ 
\leq & \mathbb{P}\Bigg [ \rvec{U}\in \TS(n,\alpha, p),~\Lambda_{\rvec{Y}^\star}(\rvec{X}, \hat{\lambda}^\star(\rvec{X}), d) - \jmath_{\rvec{X}}\LRB{\rvec{X},d} > \notag \\
&  \quad\quad C\log n\Bigg ] + \prob{\rvec{U}\not\in \TS(n,\alpha,p)} \label{lemma2:step34}\\ 
\leq & \frac{K}{(\log n)^{\kappa\alpha}} \label{lemma2:step4},
\end{align}
where~\eqref{lemma2:step2} is by Lemma~\ref{lemma:lowerbound};~\eqref{lemma2:step3} is by~\eqref{shell_lower_bd};~\eqref{lemma2:step34} is by the choice of $\gamma$ in~\eqref{choice:gamma} and $q>1$;~\eqref{lemma2:step4} is by Lemma~\ref{lemma:lemma_5} and~\eqref{TS_large}; and $K>0$ is a constant.
\end{proof}

\subsection{Proof of Theorem~\ref{thm:typicalset}}
\label{app:typical_set}
\begin{proof}
We first prove the property (1). First, Theorem~\ref{thm:error_concentration} states that for all $n$ large enough the condition~\eqref{estimation_error_small} is violated with probability at most $\frac{2}{\left(\log n\right)^{\kappa\alpha}}.$ Second, we bound the probability of violating condition~\eqref{moments_bd}. Note that since $X_i\sim\mathcal{N}(0,\sigma_i^2)$ by~\eqref{eig:sigma_i}, we have $G_i\triangleq \frac{X_i}{\sigma_i}\sim\mathcal{N}(0,1)$ for all $i\in [n]$. For each $k = 1,2,3$, applying the Berry-Esseen theorem\footnote{The Berry-Esseen theorem suffices here, though tighter bounds are possible via other concentration inequalities, say Chernoff's bound.} to the zero-mean random variables $G_i^{2k} - (2k-1)!!$, we obtain
\begin{align}
\mathbb{P}\left[\left| \frac{1}{n}\sum_{i = 1}^n G_i^{2k} - (2k-1)!!\right| > 2 \right] \leq 2Q\left(\frac{2\sqrt{n}}{r_k}\right) + \frac{12 T_k}{r_k^3\sqrt{n}}, 
\end{align}
where $r_k^2$ and $T_k$ are the variance and the third absolute moment of $G_i^{2k}- (2k-1)!!$, respectively; $r_k$ and $T_k$ are both positive constants since $G_i$'s have bounded finite-order moments. Therefore, there exists a constant $A_1'>0$ such that for all $n$ large enough, 
\begin{align}
\mathbb{P}\left[\left| \frac{1}{n}\sum_{i = 1}^n G_i^{2k} - (2k-1)!!\right| > 2 \right] \leq  \frac{A_1'}{\sqrt{n}}.
\label{bound_moments_pf}
\end{align}
The bound~\eqref{bound_moments_pf} implies that the condition~\eqref{moments_bd} is violated with probability at most $\frac{3A_1'}{\sqrt{n}}$ by the union bound.

Verifying that the condition~\eqref{mean_near_d} is satisfied with high probability is more involved. The high-level procedure is the following. The expressions for $m_i(\rvec{x})$'s in~\eqref{def:m_i} can be directly obtained from~\eqref{opt:prod_i} in Lemma~\ref{lemma:GaussianCREM} in Appendix~\ref{subsec:gaussian_crem}. Then we approximate $m_i(\rvec{x})$ using carefully crafted $\bar{m}_i(x_i)$, for which it is easier to obtain a concentration bound of the form~\eqref{mean_near_d}. At the end, the approximation gaps between $m_i(\rvec{x})$ and $\bar{m}_i(x_i)$ are shown to be sufficiently small, and~\eqref{mean_near_d} ensues. 

We now present the details. We start with a closer look at the optimizer $\hat{\rvec{F}}^\star$ in $\mathbb{R}(\hat{\rvec{X}},\rvec{Y}^\star, d)$. Recall the distributions of $\rvec{X}$ and $\rvec{Y}^\star$ in~\eqref{GMCREM:X} and \eqref{GMCREM:Y}, and the distribution of $\hat{\rvec{X}}$ in~\eqref{eqn:Xhat}. An application of Lemma~\ref{lemma:GaussianCREM} in Appendix~\ref{subsec:gaussian_crem} to $\mathbb{R}(\hat{\rvec{X}},\rvec{Y}^\star, d)$ yields that for all $\rvec{x}\in\mathbb{R}^n$,
\begin{align}
P_{\hat{\rvec{F}}^{\star}|\hat{\rvec{X}} = \rvec{x}} &= \prod_{i = 1}^n P_{\hat{F}_i^{\star}|\hat{X}_i = x_i}, \label{optimizer_Xhat_Y_prod}\\ 
\hat{F}^{\star}_i~|~\hat{X}_i = x_i  &~\sim \mathcal{N}\left(\frac{2\hat{\lambda}^\star(\rvec{x})\nu_i^2x_i}{1+2\hat{\lambda}^\star(\rvec{x})\nu_i^2}, \frac{\nu_i^2}{1+2\hat{\lambda}^\star(\rvec{x})\nu_i^2}\right),\label{optimizer_Xhat_Y}
\end{align}
where $\hat{\lambda}^\star(\rvec{x})$ is given by~\eqref{lambda_x_hat_star} and $\nu_i^2$'s are defined in~\eqref{nu_i_app}. Then, from~\eqref{optimizer_Xhat_Y} and the definition of $m_i(\rvec{x})$ in~\eqref{def:m_i}, it is straightforward to obtain the expression 
\begin{align}
m_i(\rvec{x}) = \frac{\nu_i^2}{1+2\hat{\lambda}^\star(\rvec{x})\nu_i^2} +  \frac{x_i^2}{(1+2\hat{\lambda}^\star(\rvec{x})\nu_i^2)^2}. \label{m_i_expression}
\end{align}
The quantity $m_i(\rvec{x})$ in the form of~\eqref{m_i_expression} is hard to analyze since there is no simple formula for $\hat{\lambda}^\star(\rvec{x})$. We instead consider $\bar{m}_i(x_i)$'s, defined as 
\begin{align}
\bar{m}_i(x_i) \triangleq \frac{\nu_i^2}{1+2\lambda^{\star}\nu_i^2} +  \frac{x_i^2}{(1+2\lambda^{\star}\nu_i^2)^2}, \label{bar_m_i_expression}
\end{align}
which is obtained from~\eqref{m_i_expression} by replacing $\hat{\lambda}^\star(\rvec{x})$ with $\lambda^{\star}$. The random variable $\bar{m}_i(X_i)$ is much easier to analyze, since $\lambda^{\star} = \frac{1}{2\theta_n}$ by Lemma~\ref{lemma:slope_properties} in Appendix~\ref{app:GM_CREM}, with which~\eqref{bar_m_i_expression} is simplified as 
\begin{equation}
\bar{m}_i(x_i) = \frac{\LRB{\min\LRB{\sigma_i^2,\theta_n}}^2}{\sigma_i^2}\LRB{\frac{x_i^2}{\sigma_i^2} - 1} + \min\LRB{\sigma_i^2,\theta_n}.
\label{bar_m_i_simplified}
\end{equation}
We will control the difference between $m_i(\rvec{x})$ and $\bar{m}_i(x_i)$ by bounding $|\lambda^{\star} - \hat{\lambda}^\star(\rvec{x}) |$. Indeed, a lengthy but elementary calculation, deferred to the end of the proof, shows that there exists a constant $A_1'' > 0$ (depending only on $d$) such that for all $n$ large enough, $\forall \rvec{x}\in\mathbb{R}^n$ satisfying~\eqref{estimation_error_small} and~\eqref{moments_bd}, we have
\begin{align}
\lrabs{\frac{1}{n}\sum_{i = 1}^n \bar{m}_i(x_i)  - \frac{1}{n}\sum_{i = 1}^n m_i(\rvec{x}) } \leq A_1''\eta_n.
\label{close:m_mbar}
\end{align}
With~\eqref{close:m_mbar}, we proceed to explain how to apply the Berry-Esseen theorem to obtain the following bound: there exists a constant $A_1''' > 0$ such that for all $n$ large enough and $\forall\omega > 0$, 
\begin{align}
\prob{\lrabs{\frac{1}{n}\sum_{i = 1}^n \bar{m}_i(X_i) - d}\geq \omega\sqrt{\frac{\log\log n}{n}} }\leq \frac{A_1'''}{\LRB{\log n}^{\frac{\omega^2}{2\beta^2}}}, 
\label{BE:m_bar}
\end{align}
where 
\begin{align}
\beta^2 \triangleq \frac{1}{n}\sum_{i=1}^n\var{\bar{m}_i(X_i)}.
\end{align}
To that end, first note from~\eqref{bar_m_i_simplified} and~\eqref{eqn:para_d} that
\begin{align} 
\frac{1}{n}\sum_{i = 1}^n \mathbb{E}[\bar{m}_i(X_i)] = d,
\end{align}
then an application of the Berry-Esseen theorem to $\bar{m}_i(X_i) - \min(\sigma_i^2, \theta_n)$ yields 
\begin{align}
& \prob{\lrabs{\frac{1}{n}\sum_{i = 1}^n \bar{m}_i(X_i) - d}\geq \omega\sqrt{\frac{\log\log n}{n}} } \notag \\
\leq & 2Q\left(\frac{\omega\sqrt{\log\log n}}{\beta}\right) + \frac{12T}{\beta^3\sqrt{n}} \\
\leq & \frac{2}{\left(\log n\right)^{\frac{\omega^2}{2\beta^2}}} + \frac{12T}{\beta^3\sqrt{n}}, \label{293}
\end{align}
where $T \triangleq \frac{1}{n}\sum_{i = 1}^n \mathbb{E}[ |\bar{m}_i(X_i) - \min(\sigma_i^2, \theta_n) |^3]$ is bounded. Using~\eqref{bar_m_i_simplified}, it is easy to check that there exists a constant $\beta_d > 0$ (depending only on $d$) such that $0 < \beta_d< \beta \leq \frac{\sqrt{2}\sigma^2}{(1-a)^2}$. Therefore,~\eqref{BE:m_bar} follows from~\eqref{293}. Now, we combine~\eqref{close:m_mbar} and~\eqref{BE:m_bar} to conclude that the condition~\eqref{mean_near_d} is satisfied with high probability. Define the set $\mathcal{L}\subset\mathbb{R}^n$ as
\begin{align}
\mathcal{L} \triangleq \left\{\rvec{u}\in\mathbb{R}^n : \rvec{u}~\mathrm{satisfies }~\eqref{estimation_error_small}~\mathrm{and}~\eqref{moments_bd}\right\}.
\end{align} 
Then, by Theorem~\ref{thm:error_concentration},~\eqref{bound_moments_pf} and the union bound, we have
\begin{align}
\prob{\rvec{U}\in \mathcal{L}^c} \leq \frac{2}{\left(\log n\right)^{\kappa\alpha}} + \frac{3A_1'}{\sqrt{n}}.
\end{align}
Hence, we have 
\begin{align}
&\prob{\lrabs{\frac{1}{n}\sum_{i = 1}^n m_i(\rvec{X})  - d}\geq p\eta_n} \notag \\ 
\leq & \mathbb{P}\Bigg [ \lrabs{\frac{1}{n}\sum_{i = 1}^n \bar{m}_i(X_i)  - \frac{1}{n} \sum_{i  = 1}^n m_i(\rvec{X})} + \notag \\
&\quad\quad\quad\quad \lrabs{\frac{1}{n}\sum_{i = 1}^n \bar{m}_i(X_i)  - d}\geq  p\eta_n \Bigg ] \label{step_near_d_1}\\
= & \prob{\cdot, \rvec{U}\in\mathcal{L}} + \prob{\cdot, \rvec{U}\in\mathcal{L}^c} \label{step_near_d_2}\\
\leq & \prob{\lrabs{\frac{1}{n}\sum_{i = 1}^n \bar{m}_i(X_i)  - d}\geq (p - A_1'') \eta_n } + \prob{\rvec{U}\in\mathcal{L}^c} \label{step_near_d_3}\\
\leq &\frac{A_1'''}{\LRB{\log n}^{\frac{(p-A_1'')^2\alpha}{2\beta^2}}} + \frac{2}{\left(\log n\right)^{\kappa\alpha}} + \frac{3A_1'}{\sqrt{n}},\label{step_near_d_4}
\end{align}
where~\eqref{step_near_d_1} is due to the triangle inequality;~\eqref{step_near_d_3} holds by~\eqref{close:m_mbar};~\eqref{step_near_d_4} follows from~\eqref{BE:m_bar} for $p>A_1''$. Hence, for any $p$ such that
\begin{align}
p \geq A_1'' + \frac{2\sqrt{\kappa}\sigma^2}{(1-a)^2}, 
\label{value_p}
\end{align}
we conclude from~\eqref{step_near_d_4} that there exists a constant $\tilde{A_1} > 0$ such that for all $n$ large enough, 
\begin{align}
\prob{\lrabs{\frac{1}{n}\sum_{i = 1}^n m_i(\rvec{X})  - d}\geq p\eta_n } \leq \frac{\tilde{A_1}}{\left(\log n\right)^{\kappa\alpha}}.\label{near}
\end{align}
Therefore,~Theorem~\ref{thm:error_concentration},~\eqref{bound_moments_pf} and~\eqref{near}  altogether imply the property (1) in Theorem~\ref{thm:typicalset}.

Next, we show property (2) in Theorem~\ref{thm:typicalset}. By the triangle inequality, we have $\forall\rvec{u}\in\mathbb{R}^n$ and $\forall i\in [n]$,
\begin{align}
\lrabs{\hat{\sigma}_i^2 - \sigma_i^2} \leq \lrabs{\hat{\sigma}_i^2 - \frac{\sigma^2}{\xi_i}} + \lrabs{\frac{\sigma^2}{\xi_i} - \sigma_i^2},
\label{step_sigma_i_hat_1}
\end{align}
where $\xi_i$ is given in~\eqref{eig:W}. We bound the two terms in~\eqref{step_sigma_i_hat_1} separately. From~\eqref{eig:sigma_i},~\eqref{bound:ximu} and~\eqref{constant_eig_bound}, we have
\begin{align}
\lrabs{\frac{\sigma^2}{\xi_i} - \sigma_i^2} \leq \frac{2a\pi\sigma^2}{(1-a)^4 n}.
\label{step_sigma_i_hat_2}
\end{align}
To simplify notations, let $\phi_i\triangleq \frac{i\pi}{n+1}$ and denote by $\varphi(t)$ the function
\begin{align}
\varphi(t) \triangleq \frac{\sigma^2}{1 + t^2 -2t\cos\phi_i}.
\end{align}
It is easy to see that the derivatives $\phi ' (a)$ and $\phi''(a)$ are bounded for any fixed $a\in [0,1)$. By the Taylor expansion and the triangle inequality, we have
\begin{align}
\lrabs{\hat{\sigma}_i^2 - \frac{\sigma^2}{\xi_i}} &= \lrabs{\varphi(\hat{a}(\rvec{u})) - \varphi(a)} \\
& \leq \lrabs{\varphi'(a)}\lrabs{\hat{a}(\rvec{u}) - a} + o\left(\lrabs{\hat{a}(\rvec{u}) - a}\right).
\label{step_sigma_i_hat_3}
\end{align}
Hence, combining~\eqref{step_sigma_i_hat_1},~\eqref{step_sigma_i_hat_2} and~\eqref{step_sigma_i_hat_3}, we conclude that there exists a constant $A_2>0$ such that for all $n$ large enough~\eqref{eqn:error_sigma} holds for any $\rvec{u}\in\TS(n,\alpha, p)$. 

Finally, the bound~\eqref{eqn:error_lambda} follows immediately from a direct application of Lemma~\ref{lemma:bound_slope} to~\eqref{eqn:error_sigma}.

\emph{Calculations to show~\eqref{close:m_mbar}:} 
From~\eqref{m_i_expression} and~\eqref{bar_m_i_expression}, we have
\begin{align}
&\frac{1}{n}\sum_{i = 1}^n \bar{m}_i(x_i)  - \frac{1}{n}\sum_{i = 1}^n m_i(\rvec{x}) \\
= & \frac{1}{n}\sum_{i = 1}^n \frac{2\nu_i^4\LRB{\hat{\lambda}^\star(\rvec{x}) - \lambda^{\star}}}{\LRB{1 + 2\hat{\lambda}^\star(\rvec{x})\nu_i^2}\LRB{1 + 2\lambda^{\star}\nu_i^2}} + \notag \\
&\quad \frac{1}{n}\sum_{i = 1}^n \frac{2x_i^2\nu_i^2\LRB{2 + 2\hat{\lambda}^\star(\rvec{x})\nu_i^2 +2\lambda^{\star}\nu_i^2}\LRB{\hat{\lambda}^\star(\rvec{x}) - \lambda^{\star}}}{\LRB{1 + 2\hat{\lambda}^\star(\rvec{x})\nu_i^2}^2\LRB{1 + 2\lambda^{\star}\nu_i^2}^2}.\label{close_m_derivation}
\end{align}
By~\eqref{eqn:error_lambda}, for all $n$ large enough, $\forall \rvec{u}\in\TS(n,\alpha,p)$ and $\rvec{x} = \mathsf{S}^\top\rvec{u}$, we have
\begin{align}
\lrabs{\LRB{1 + 2\hat{\lambda}^\star(\rvec{x})\nu_i^2} - \LRB{1 + 2\lambda^{\star}\nu_i^2}} \leq \frac{1 + 2\lambda^{\star}\nu_i^2}{2}. 
\label{close:deno}
\end{align}
Using~\eqref{nu_i_app} and~\eqref{para:slope}, we deduce that $1 \leq 1 + 2\lambda^{\star}\nu_i^2 \leq  \frac{\sigma_i^2}{\theta_n}$. Therefore,~\eqref{close:deno} implies that 
\begin{align}
\frac{1}{2} \leq 1 + 2\hat{\lambda}^\star(\rvec{x})\nu_i^2 \leq \frac{3\sigma_i^2}{2\theta_n}. \label{close:deno_sim}
\end{align}
We continue to bound~\eqref{close_m_derivation} as
\begin{align}
&\lrabs{\frac{1}{n}\sum_{i = 1}^n \bar{m}_i(x_i)  - \frac{1}{n}\sum_{i = 1}^n m_i(\rvec{x}) } \label{close_m_step0}\\
\leq  & \frac{1}{n}\sum_{i = 1}^n 4\nu_i^4 \lrabs{\hat{\lambda}^\star(\rvec{x}) - \lambda^{\star}} + \frac{1}{n}\sum_{i = 1}^n \frac{20 x_i^2\nu_i^2\sigma_i^2}{\theta_n} \lrabs{\hat{\lambda}^\star(\rvec{x}) - \lambda^{\star}} \label{close_m_step1}\\
\leq & \left[ \frac{1}{n}\sum_{i = 1}^n 4\nu_i^4+ \frac{1}{n}\cdot \frac{20 \sigma^6 }{\theta_n (1-a)^6}\sum_{i  =1}^n\frac{x_i^2}{\sigma_i^2}\right]  \lrabs{\hat{\lambda}^\star(\rvec{x}) - \lambda^{\star}} \label{close_m_step2}\\
\leq & A_1''\eta_n, \label{close_m_step3}
\end{align}
where~\eqref{close_m_step1} is by plugging~\eqref{close:deno_sim} into~\eqref{close_m_derivation};~\eqref{close_m_step2} is by $\nu_i^2 \leq \sigma_i^2$ and  
\begin{align}
\sigma_i^6 = \frac{\sigma^6}{\mu_i^3} \leq \frac{\sigma^6}{(1-a)^6},\quad \forall i\in [n], 
\end{align}
which is due to~\eqref{constant_eig_bound}; and~\eqref{close_m_step3} holds for some constant $A_1'' > 0$ (depending on $d$ only) by~\eqref{moments_bd},~\eqref{eqn:error_lambda} and~\eqref{bound:theta_n_13}. 
\end{proof}

\subsection{Proof of Lemma~\ref{lemma:lemma_5}}
\label{app:proof_lemma_5}
\begin{proof}
We sketch the proof, which is similar to~\cite[Lem. 5]{kostina2012fixed} except for some slight changes. Since $Y_1^\star,\ldots, Y_n^\star$ are independent and the distortion measure $\dis{\cdot}{\cdot}$ is separable, we use the definition~\eqref{def:generalizeddtilted}, the distribution formula~\eqref{GMCREM:Y} for $\rvec{Y}^\star$, and the formula for the moment generating function of a noncentral $\chi_1^2$-distributed random variable to obtain for $\delta > 0$ and $\rvec{x}\in\mathbb{R}^n$,
\begin{align}
\Lambda_{\rvec{Y}^\star }\LRB{\rvec{x}, \delta, d} = -n\delta d + \sum_{i = 1}^n \frac{\delta x_i^2}{1 + 2\delta \nu_i^2} + \sum_{i = 1}^n \frac{1}{2} \log\LRB{1 + 2\delta \nu_i^2}, \label{eqn:app_generalized_dtilted}
\end{align}
where $\nu_i^2$'s are in~\eqref{nu_i_app}. Let $\hat{\lambda}^\star(\rvec{x})$ be defined in~\eqref{lambda_x_hat_star}. Similar to~\cite[Eq. (315)-(320)]{kostina2012fixed}, by the Taylor expansion of $\Lambda_{\rvec{Y}^\star }\LRB{\rvec{x}, \delta, d}$ in $\delta$ at the point $\delta = \lambda^\star$, we have for any $\rvec{x}\in\mathbb{R}^n$, 
\begin{align}
\Lambda_{\rvec{Y}^\star}\LRB{\rvec{x}, \hat{\lambda}^\star(\rvec{x}),d} - \Lambda_{\rvec{Y}^\star}\LRB{\rvec{x}, \lambda^{\star},d}  \leq \frac{(S'(\rvec{x}))^2}{2S''(\rvec{x})},  \label{bd:quadratic}
\end{align}
where we denoted 
\begin{align}
S'(\rvec{x}) \triangleq \Lambda'_{\rvec{Y}^\star }\LRB{\rvec{x}, \lambda^\star, d} = \sum_{i = 1}^n \frac{\left[\min(\theta_n, \sigma_i^2)\right]^2}{\sigma_i^2} \left(\frac{x_i^2}{\sigma_i^2} - 1\right),\label{first_derivative}
\end{align}
and 
\begin{align}
S''(\rvec{x}) \triangleq -\Lambda''_{\rvec{Y}^\star }\LRB{\rvec{x}, \lambda^\star + \xi(\rvec{x}), d}, \label{second_derivative}
\end{align}
where~\eqref{first_derivative} is by first taking derivatives of~\eqref{eqn:app_generalized_dtilted} with respect to $\delta$ and then plugging $\lambda^\star = \frac{1}{2\theta_n}$; $\theta_n > 0$ is the water level matched to $d$ via the $n$-th order reverse waterfilling~\eqref{eqn:para_d} over $\sigma_i^2$'s; and 
\begin{align}
\xi(\rvec{x}) \triangleq \rho (\hat{\lambda}^\star(\rvec{x}) - \lambda^\star),
\end{align}
for some $\rho \in [0,1]$. Note from the definition~\eqref{crem:GMdef} that for any $\rvec{x}$, $\rdf(\hat{\rvec{X}}(\rvec{x}), \rvec{Y}^\star, d)$ is a nonincreasing function in $d$, hence $\hat{\lambda}^\star(\rvec{x}) \geq 0$, which, combined with direct computations using~\eqref{eqn:app_generalized_dtilted}, implies $S''(\rvec{x}) \geq 0$, see~\eqref{SSS} below. This fact was used in deriving the inequality in~\eqref{bd:quadratic}.

Next we show concentration bounds for the two random variables $S'(\rvec{X})$ and $S''(\rvec{X})$. From~\eqref{first_derivative}, we see that $S'(\rvec{X})$ is a sum of $n$ zero-mean and bounded-variance random variables. Then, by the Berry-Esseen theorem, we have 
\begin{align}
\prob{|S'(\rvec{X})| \geq \sqrt{V'n\log n} } \leq \frac{K_1'}{\sqrt{n}}, \label{bound:SS}
\end{align}
where $K_1' > 0$ is a constant and $V'$ is a constant such that 
\begin{align}
V' \geq \frac{1}{n}\sum_{i = 1}^n \frac{2\left[\min(\theta_n, \sigma_i^2)\right]^4}{\sigma_i^4}.
\end{align}
To treat $S''(\rvec{X})$, consider the following event $\mathcal{E}$: 
\begin{align}
\mathcal{E} \triangleq \{\rvec{x}\in\mathbb{R}\colon | \hat{\lambda}^\star(\rvec{x}) - \lambda^\star | \leq n^{-1/4}\}.
\end{align}
Using~Lemma~\ref{lemma:bound_slope} in Appendix~\ref{subsec:sen_crem},~\eqref{step_sigma_i_hat_1}-\eqref{step_sigma_i_hat_3} and Theorem~\ref{thm:ECG} in Section~\ref{subsubsec:Teit}, one can show (similar to the proof of Theorem~\ref{thm:error_concentration}) that there exists a constant $c' > 0$ such that for all $n$ large enough, 
\begin{align}
\prob{\mathcal{E}^c} \leq \exp\{ - c'\sqrt{n}\}.
\end{align}
Computing the derivatives using~\eqref{eqn:app_generalized_dtilted} yields
\begin{align}
S''(\rvec{x}) & = \sum_{i = 1}^n 4x_i^2\nu_i^2\left[1 + 2\nu_i^2(\lambda^\star + \xi(\rvec{x}))\right]^{-3} + \notag \\
&\quad 2\nu_i^4\left[1 + 2\nu_i^2(\lambda^\star + \xi(\rvec{x}))\right]^{-2}. \label{SSS}
\end{align}
By conditioning on $\mathcal{E}$ and $\mathcal{E}^c$, we see that for any $t>0$, 
\begin{align}
&\prob{S''(\rvec{X}) \leq nt} \notag \\
= & \prob{S''(\rvec{X}) \leq nt,~\mathcal{E}} + \prob{S''(\rvec{X}) \leq nt,~\mathcal{E}^c} \\
\leq &\prob{S''(\rvec{X}) \leq nt,~\mathcal{E}} + \prob{\mathcal{E}^c}. \label{bound:SSS1}
\end{align}
Using~\eqref{SSS} and the simple bound $\frac{1}{x + y} \geq \frac{\exp(-y/x)}{x}$ for any $x, y>0$, we have for any $\rvec{x} \in \mathcal{E}$,
\begin{align}
& S''(\rvec{x}) \notag \\
\geq &  \exp(-3n^{-1/4}/\lambda^\star) \times \notag \\ 
&  \sum_{i = 1}^n \left[4x_i^2\nu_i^2\left(1 + 2\nu_i^2\lambda^\star \right)^{-3} + 2\nu_i^4\left(1 + 2\nu_i^2\lambda^\star \right)^{-2}\right]  \\
= & \sum_{i = 1}^n A_i x_i^2 + B_i,
\end{align}
where $A_i, B_i \geq 0$ are defined as
\begin{align}
A_i &\triangleq \exp(-3n^{-1/4}/\lambda^\star)4 \nu_i^2\left(1 + 2\nu_i^2\lambda^\star\right)^{-3}, \label{def:Ai}\\
B_i &\triangleq \exp(-3n^{-1/4}/\lambda^\star)2\nu_i^4\left(1 + 2\nu_i^2\lambda^\star \right)^{-2}.\label{def:Bi}
\end{align}
Let $V_n$ be defined as 
\begin{align}
V_n \triangleq \frac{1}{n}\sum_{i = 1}^n 2A_i\sigma_i^4,  
\end{align}
and choose the constant $t$ in~\eqref{bound:SSS1} such that for all $n$ large enough,
\begin{align}
0 < t < \frac{1}{n}\sum_{i = 1}^n ( A_i\sigma_i^2 + B_i ) - \sqrt{\frac{V_n\log n}{n}}. \label{choice:t}
\end{align}
By the Berry-Esseen theorem, there exists a constant $K_1'' >0$ such that
\begin{align}
\prob{S''(\rvec{X}) \leq nt,~\mathcal{E}} \leq \prob{ \sum_{i = 1}^n A_i X_i^2 + B_i \leq nt }  \leq \frac{K_1''}{\sqrt{n}}.\label{bound:SSS2}
\end{align}
The existence of such a constant $t$ satisfying~\eqref{choice:t} is guaranteed since for any $d<\dmax$ and for all $n$ sufficiently large, there is a constant fraction of strictly positive $A_i$'s and $B_i$'s, which can be verified by plugging~\eqref{nu_i_app} and~\eqref{para:slope} into~\eqref{def:Ai} and~\eqref{def:Bi}. Combining~\eqref{bound:SSS1} and~\eqref{bound:SSS2} implies that there exists a constant $K_1''' >0$ such that for all $n$ large enough
\begin{align}
\prob{S''(\rvec{X}) \leq nt} \leq \frac{K_1'''}{\sqrt{n}},\label{bound:SSS3}
\end{align}
where the constant $t$ satisfies~\eqref{choice:t}. Finally, similar to~\cite[Eq. (339)]{kostina2012fixed}, combining~\eqref{bd:quadratic},~\eqref{bound:SS} and~\eqref{bound:SSS3} yields~\eqref{con_gen_inf}.
\end{proof}

\subsection{Proof of Lemma~\ref{lemma:lemma_4}}
\label{app:proof_lemma_4}
\begin{proof} 
Due to Lemma~\ref{lemma:GaussianCREM} in Appendix~\ref{subsec:gaussian_crem}, for any $\rvec{x}\in\mathbb{R}^n$, we can write the random variable involved in~\eqref{shell_lower_bd} as a sum of independent random variables, to which the Berry-Esseen theorem is applied. The details follow. From~\eqref{optimizer_Xhat_Y_prod} and~\eqref{optimizer_Xhat_Y}, we have 
\begin{align}
\dis{\rvec{x}}{\hat{\rvec{F}^\star}}  = \frac{1}{n}\sum_{i = 1}^n M^2_i(\rvec{x}),
\label{shell_prob_sum}
\end{align}
where the random variables $M_i(\rvec{x})$'s are defined as
\begin{align}
M_i(\rvec{x})\triangleq (\hat{F}^\star_i - x_i). \label{M_i}
\end{align}
From~\eqref{optimizer_Xhat_Y}, we know that the conditional distribution satisfies 
\begin{align}
M_i(\rvec{x})~|~\hat{\rvec{X}} = \rvec{x} \sim \mathcal{N}\left( \frac{-x_i}{1+2\hat{\lambda}^\star(\rvec{x})\nu_i^2}, \frac{\nu_i^2}{1+2\hat{\lambda}^\star(\rvec{x})\nu_i^2}\right), 
\label{distribution:M_i}
\end{align}
where $\hat{\lambda}^\star(\rvec{x})$ is in~\eqref{lambda_x_hat_star}. Hence, conditioned on $\hat{\rvec{X}} = \rvec{x}$,  the random variable $\dis{\rvec{x}}{\hat{\rvec{F}^\star}}$ follows the noncentral $\chi^2$-distribution with (at most) $n$ degrees of freedom. Applying the Berry-Esseen theorem to~\eqref{shell_prob_sum} yields that $\forall \gamma >0$,
\begin{align}
&\prob{d - \gamma < \dis{\rvec{x}}{\hat{\rvec{F}^\star}} \leq d ~ | \hat{\rvec{X}} = \rvec{x} } \notag \\
= & \mathbb{P} \Bigg [ \frac{nd - n\gamma -\sum_{i = 1}^n m_i(\rvec{x})}{s\sqrt{n}} < \frac{1}{s\sqrt{n}}\sum_{i = 1}^n \LRB{M_i^2(\rvec{x}) - m_i(\rvec{x})} \notag \\
& \quad\quad \leq \frac{nd  -\sum_{i = 1}^n m_i(\rvec{x})}{s\sqrt{n}} ~ | \hat{\rvec{X}} = \rvec{x} \Bigg ]  \label{lemma_pf_step1}\\
\geq & \Phi\left(\frac{nd  -\sum_{i = 1}^n m_i(\rvec{x})}{s\sqrt{n}}\right) - \Phi\left(\frac{nd - n\gamma -\sum_{i = 1}^n m_i(\rvec{x})}{s\sqrt{n}}\right) \notag \\
&\quad\quad - \frac{12t}{s^3\sqrt{n}}, \label{lemma_pf_step2}
\end{align}
where $m_i(\mathbf{x})$, defined in~\eqref{def:m_i}, is the expectation of $M_i^2(\rvec{x})$ conditioned on $\hat{\rvec{X}} = \rvec{x}$; and 
\begin{align}
s^2 &\triangleq \frac{1}{n}\sum_{i = 1}^n \var{M_i^2(\rvec{x})~|~\hat{\rvec{X}} = \rvec{x}},\label{eqn:s}\\
t     &\triangleq \frac{1}{n}\sum_{i = 1}^n \mathbb{E}[\left|M_i^2(\rvec{x}) - m_i(\rvec{x})\right|^3~|~\hat{\rvec{X}} = \rvec{x}].\label{eqn:t}
\end{align}
By the mean value theorem,~\eqref{lemma_pf_step2} equals
\begin{align}
\frac{ n\gamma}{s\sqrt{n}\sqrt{2\pi}} e^{-\frac{b^2}{2}}  - \frac{12t}{s^{3}\sqrt{n}} \label{lemma_pf_step3}
\end{align}
for some $b$ satisfying 
\begin{align}
\frac{nd  - \sum_{i = 1}^n m_i(\rvec{x}) - n\gamma}{s\sqrt{n}}\leq b \leq \frac{nd  -\sum_{i = 1}^n m_i(\rvec{x})}{s\sqrt{n}}.\label{condition:b}
\end{align}
To further lower-bound \eqref{lemma_pf_step3}, we bound $b^2$ as follows.
\begin{align}
b^2 &\leq 2\LRB{\frac{nd  -\sum_{i = 1}^n m_i(\rvec{x})}{s\sqrt{n}}}^2 + 2\LRB{\frac{n\gamma}{s\sqrt{n}}}^2 \label{small_step0}\\
&\leq \frac{2p^2\alpha\log\log n}{s^2} + \frac{2\log^{2q}n}{s^2n} \label{small_step1}\\
&\leq \frac{4p^2\alpha\log\log n}{\mathsf{c}_s^2}, \label{small_step2}
\end{align}
where~\eqref{small_step0} is by~\eqref{condition:b} and the elementary inequality $(x+y)^2 \leq 2(x^2 + y^2)$; ~\eqref{small_step1} is by~\eqref{mean_near_d} and the choice of $\gamma$ in~\eqref{choice:gamma}. The constant $q$ in~\eqref{choice:gamma} is chosen to be 
\begin{align}
q \triangleq \frac{2p^2\alpha}{\mathsf{c}_s^2} + 1. \label{defq}
\end{align}
The constant $\mathsf{c}_s > 0$ is a lower bound of $s$, whose existence is justified below at the end of the proof. Finally,~\eqref{small_step2} holds for all sufficiently large $n$. Using~\eqref{small_step2}, we can further lower-bound~\eqref{lemma_pf_step3} as
\begin{align}
& \LRB{ \frac{\log^{q} n}{s\sqrt{2\pi}} \left(\log n\right)^{-\frac{2p^2\alpha}{\mathsf{c}_s^2}} - \frac{12t}{s^3} }\frac{1}{\sqrt{n}} \label{lemma_pf_step4}\\
 = & \LRB{ \frac{\log n}{s\sqrt{2\pi}} - \frac{12t}{s^3} }\frac{1}{\sqrt{n}} \label{lemma_pf_step5}\\ 
\geq &\frac{C_1}{\sqrt{n}} \label{lemma_pf_step6}, 
\end{align}
where~\eqref{lemma_pf_step4} is by plugging~\eqref{small_step2} into~\eqref{lemma_pf_step3};~\eqref{lemma_pf_step5} is by using~\eqref{defq}; and~\eqref{lemma_pf_step6} holds for all sufficiently large $n$ and some constant $C_1>0$. Therefore,~\eqref{shell_lower_bd} follows.

Finally, to justify that $s$ and $t$, defined in~\eqref{eqn:s} and~\eqref{eqn:t}, are bounded as we assumed in obtaining~\eqref{small_step2} and~\eqref{lemma_pf_step6}, we compute using  \eqref{distribution:M_i}
\begin{align}
& \var{M_i^2(\rvec{x})~|~\hat{\rvec{X}} = \rvec{x}} \notag \\
= & \frac{4x_i^2\nu_i^2}{(1+2\hat{\lambda}^\star(\rvec{x})\nu_i^2)^3} + \frac{2\nu_i^4}{(1+2\hat{\lambda}^\star(\rvec{x})\nu_i^2)^2}.
\end{align}
Then, using~\eqref{close:deno_sim} to bound $1+2\hat{\lambda}^\star(\rvec{x})\nu_i^2$ and~\eqref{moments_bd} to bound $x_i^2$, we can lower- and upper-bound $s$ by positive constants; $t$ can be bounded similarly.
\end{proof}

\section{}
\label{app:detailed_proofs}

\subsection{Derivation of the maximum likelihood estimator}
\label{app:MLEderivation}
This section presents the details in obtaining~\eqref{MLEformula}. The random vector $(U_1, U_2-aU_1, ..., U_n - aU_{n-1})^\top$ is distributed according to $\mathcal{N}(0,\sigma^2\mathsf{I})$. Let $p_a(\cdot)$ be the probability density function of $\rvec{U}$ with parameter $a$, then 
\begin{align}
& \hat{a}(\rvec{u}) \notag \\
\triangleq & \arg\max_{a} p_a(\rvec{u}) \label{MLE:step1}\\
 = &  \arg\max_{a} \prod_{i=1}^n \frac{1}{\sqrt{2\pi\sigma^2}}e^{-\frac{1}{2\sigma^2}(u_i-au_{i-1})^2} \label{MLE:step2}\\
= & \arg\min_{a} \sum_{i = 2}^n(u_i-au_{i-1})^2 + u_1^2 \label{MLE:step3}\\
= & \arg\min_{a} \sum_{i = 2}^n u_{i-1}^2 a^2 - 2\sum_{i = 2}^n u_{i-1}u_i a + \sum_{i = 1}^n u_i^2  \label{MLE:step4} \\
= & \frac{\sum_{i = 1}^{n-1} u_{i}u_{i+1} }{\sum_{i = 1}^{n-1} u_i^2},\label{MLE:step5}
\end{align}  
where~(\ref{MLE:step3}) is by collecting the terms in the exponent, and~(\ref{MLE:step5}) is by minimizing the quadratic function of $a$.

\subsection{Proof of Theorem~\ref{thm:ECG}}
\label{proof:thm_ECG}
\begin{proof}
We first bound $\mathbb{P}\left[\hat{a}(\rvec{U}) - a > \eta\right]$. Define the random variable $W(n, \eta)$ as 
\begin{align}
W(n, \eta) \triangleq \sum_{i=1}^{n-1} \left( U_i Z_{i+1} - \eta U_i^2 \right). \label{def:W}
\end{align}
Then, from~\eqref{MLEformula} in Section~\ref{subsubsec:Teit} above and~\eqref{def:GaussMarkov}, it is easy to see that 
\begin{align}
\mathbb{P}\left[\hat{a}(\rvec{U}) - a > \eta\right] = \mathbb{P}\left[W(n, \eta) > 0\right]. \label{eqn:errW}
\end{align}
From the definition of the Gauss-Markov source in~\eqref{def:GaussMarkov}, we have 
\begin{align}
U_i = \sum_{j = 1}^i a^{i-j} Z_j.  \label{eqn:uz}
\end{align}
Plugging~\eqref{eqn:uz} into~\eqref{def:W}, we obtain 
\begin{align}
W(n, \eta) & = \sum_{i = 2}^{n} \sum_{j = 1}^{i-1} a^{i-j -1} Z_{i} Z_{j} - \notag \\
& \quad \frac{\eta}{1 - a^2} \sum_{i = 1}^{n-1}\sum_{j = 1}^{n-1}\left(a^{|i - j|} - a^{2n - i-j}\right) Z_i Z_{j}. \label{eqn:Wexpan}
\end{align}
Notice that~\eqref{eqn:Wexpan} can be further rewritten as the following quadratic form in the $n$ i.i.d. random variables $Z_1, \ldots, Z_n$: 
\begin{align}
W(n, \eta) = \rvec{Z}^\top \mathsf{Q}(n, \eta) \rvec{Z}, \label{Quadratic}
\end{align}
where $\mathsf{Q}(n, \eta)$ is an $n\times n$ symmetric matrix defined as 
\begin{align}
\mathsf{Q}_{i,j}(n, \eta) = \begin{cases}
 -\eta\frac{1 - a^{2(n-i)}}{1 - a^2}, & i = j < n \\
0, & i = j = n \\
\frac{1}{2}a^{|i-j|-1} - \eta \frac{a^{|i-j|} - a^{2n - i - j}}{1 - a^2}, &\text{otherwise}.
\end{cases} \label{Q}
\end{align}
For simplicity, in the rest of the proof, we write $W$ and $\mathsf{Q}$ for $W(n, \eta)$ and $\mathsf{Q}(n, \eta)$, respectively. By the Hanson-Wright inequality~\cite[Th. 1.1]{rudelson2013hanson}, there exists a universal constant $c>0$ such that for any $t> 0$~\footnote{The sub-gaussian norm $\|Z_i\|_{\psi_2}$~\cite[Def. 5.7]{vershynin2010introduction} of $Z_i\sim\mathcal{N}(0,\sigma^2)$ satisfies $\|Z_i\|_{\psi_2} \leq C_1 \sigma$ for a universal constant $C_1 > 0$, see~\cite[Example 5.8]{vershynin2010introduction}. By~\cite[Th. 1.1]{rudelson2013hanson}, there exists a universal constant $C_2 > 0$ such that for any $t>0$, 
\begin{align*}
\prob{W - \mathbb{E} \left[ W\right] > t } \leq \exp\left[ - C_2 \min\left( \frac{t^2}{C_1^4 \sigma^4\| \mathsf{Q} \|_{F}^2},~ \frac{t}{C_1^2 \sigma^2\|\mathsf{Q}\|}\right)\right].
\end{align*} 
In~\eqref{HW}, we absorb $C_1$ and $C_2$ into a single universal constant $c$.
}
\begin{align}
\prob{ W - \mathbb{E} \left[ W\right] > t } \leq \exp\left[ - c \min\left( \frac{t^2}{\sigma^4\| \mathsf{Q} \|_{F}^2},~ \frac{t}{\sigma^2\|\mathsf{Q}\|}\right)\right],\label{HW}
\end{align} 
where $\| \mathsf{Q} \|_{F}$ and $\| \mathsf{Q} \|$ are the Frobenius and operator norms of $\mathsf{Q}$, respectively. Taking $t = - \mathbb{E} \left[ W \right]$ (which is $>0$ for all $n$ large enough, as shown in~\eqref{bound:E} below) in~\eqref{HW}, we can bound~\eqref{eqn:errW} as 
\begin{align}
\prob{ W  > 0} \leq \exp\left[ - c \min\left( \frac{\left(- \mathbb{E} \left[ W \right]\right)^2}{\sigma^4\| \mathsf{Q} \|_{F}^2},~ \frac{- \mathbb{E} \left[ W \right]}{\sigma^2\|\mathsf{Q}\|}\right)\right]. \label{E}
\end{align}

It remains to bound $\mathbb{E} \left[ W \right]$, $\| \mathsf{Q} \|_{F}^2$ and $\| \mathsf{Q} \|$. In the following, we show that  $- \mathbb{E} \left[ W\right] = \Theta( \eta n)$, $\| \mathsf{Q} \|_{F}^2 = \Theta ( n )$ and $\| \mathsf{Q} \| = O(1)$. Plugging these estimates into~\eqref{E} yields~\eqref{mle1} up to constants. The details follow. We first consider $\mathbb{E}\left[W \right]$. From~\eqref{Quadratic} and~\eqref{Q}, we have 
\begin{align}
\mathbb{E}\left[W\right] = \sigma^2 \text{tr}(\mathsf{Q}) =  -\frac{\eta\sigma^2 n}{1 - a^2} + \frac{\eta \sigma^2 (1 - a^{2n})}{(1-a^2)^2}.\label{comE}
\end{align}
Define the constant $K_1 > 0$ as 
\begin{align}
K_1\triangleq \frac{1}{2(1-a^2)}.\label{eqn:K1}
\end{align}
Then, for all $n$ large enough, we have
\begin{align}
-\mathbb{E}\left[W\right] \geq K_1 \sigma^2 \eta n. \label{bound:E}
\end{align}
We then consider $\|\mathsf{Q}\|_{F}^2$. Direct computations using~\eqref{Q} yield 
\begin{align}
\|\mathsf{Q}\|_{F}^2 &= \Bigg [\frac{1}{2(1-a^2)} + \frac{(1+a^2)\eta^2 - 2a(1-a^2)\eta}{(1-a^2)^3} + \notag \\
&\quad \frac{\left( 4a\eta^2 - 2(1-a^2)\eta\right)a^{2n}}{a(1-a^2)^3}\Bigg ] \cdot n  \notag \\
&\quad\quad + \frac{4a\eta}{(1-a^2)^3} - \frac{1}{2(1-a^2)^2}-\frac{\eta^2 (4a^2 + 1)}{(1-a^2)^4} \notag \\
&\quad\quad + \left[\frac{4a^2\eta^2}{(1-a^2)^4} + \frac{1}{2(1-a^2)^2} - \frac{4a\eta}{(1-a^2)^3}\right]\cdot a^{2n} \notag \\
&\quad\quad + \frac{\eta^2 }{(1-a^2)^4}\cdot a^{4n}. \label{compQF}
\end{align}
Define the constant $K_2 > 0$ as 
\begin{align}
K_2 \triangleq \frac{1}{1-a^2} + \frac{2(1+5a^2)\eta^2}{(1-a^2)^3}. \label{eqn:K2}
\end{align}
Since $\eta >0$ and $a \in [0, 1)$, from~\eqref{compQF}, we have for all $n$ large enough,  
\begin{align}
\|\mathsf{Q}\|_{F}^2 \leq K_2 n. \label{bound:QF}
\end{align}
Finally, we bound $\|\mathsf{Q}\|$. Using the Gershgorin circle theorem~\cite[p.16, Th. 1.11]{varga2009matrix}, one can easily show that  
\begin{align}
\|\mathsf{Q}\| \leq \max_{i\in [n]}~\|q_i\|_1, 
\end{align}
where $q_i$ denotes the $i$-th row in $\mathsf{Q}$. Direct computations using~\eqref{Q} yield that $\forall i\in [n]$, 
\begin{align}
\|q_i\|_1 \leq K_3,
\end{align}
where $K_3$ is a positive constant given by 
\begin{align}
K_3 \triangleq  \frac{1}{1 - a} + \frac{\eta}{1-a^2} + \frac{2\eta}{(1-a)^2}. \label{eqn:K3}
\end{align}
Therefore, $\forall n\geq 1$, we have 
\begin{align}
\|\mathsf{Q}\| \leq K_3. \label{bound:QO}
\end{align}
Plugging the bounds~\eqref{bound:E},~\eqref{bound:QF} and~\eqref{bound:QO} into~\eqref{E}, we have for all $n$ large enough 
\begin{align}
\mathbb{P}\left[W > 0 \right] \leq \exp\left[-c\min\left(\frac{K_1^2\eta^2 n}{K_2}, \frac{K_1\eta n}{ K_3}\right)\right].\label{bound:general}
\end{align}

Notice that all the arguments up to this point are valid for any $\eta >0$. However, the constants $K_2$ and $K_3$ in the bound~\eqref{bound:general} depend on $\eta$ via~\eqref{eqn:K2} and~\eqref{eqn:K3}, respectively. Since we are interested in small $\eta$, in the rest of the proof, we assume $\eta \in (0,1)$. Besides, with the restriction $\eta \in (0,1)$, we can get rid of the dependence of $K_2$ and $K_3$ on $\eta$, and simplify~\eqref{bound:general}. For any $\eta \in (0,1)$, we can bound $K_2$ and $K_3$ as follows:
\begin{align}
K_2 \leq K_2' \triangleq \frac{1}{1-a^2} + \frac{2(5+a^2)}{(1-a^2)^3},\label{K2p}
\end{align}
\begin{align}
K_3 \leq K_3' \triangleq \frac{a+4}{(1-a)^2 (1+a)}. \label{K3p}
\end{align}
Applying the bounds~\eqref{K2p} and~\eqref{K3p} to~\eqref{bound:general}, and then setting 
\begin{align}
c_1 \triangleq \frac{K_1^2}{K_2'} \quad\text{and}\quad  c_2 \triangleq \frac{K_1}{K_3'} \label{c1c2}
\end{align}
yields
\begin{align}
\mathbb{P}\left[\hat{a}(\rvec{U}) - a > \eta \right] \leq \exp\left[-c\min\left( c_1 \eta^2 n, c_2 \eta n\right)\right]. \label{eqn:bigger}
\end{align}

Finally, to bound $\mathbb{P}\left[\hat{a}(\rvec{U}) - a < -\eta \right]$, in the above proof, we replace the random variable $W(n, \eta)$ by $V(n, \eta)$, defined as 
\begin{align}
V(n, \eta) \triangleq \sum_{i=1}^{n-1} \left( -U_i Z_{i+1} - \eta U_i^2 \right).
\end{align}
In quadratic forms, $V(n, \eta) = \rvec{Z}^\top \mathsf{S}(n, \eta) \rvec{Z}$, where $\mathsf{S}(n, \eta)$ is an $n\times n$ symmetric matrix, defined in a way similar to $\mathsf{Q}(n, \eta)$: 
\begin{align}
\mathsf{S}_{i,j}(n, \eta) = \begin{cases}
 -\eta\frac{1 - a^{2(n-i)}}{1 - a^2}, & i = j < n \\
0, & i = j = n \\
-\frac{1}{2}a^{|i-j|-1} - \eta \frac{a^{|i-j|} - a^{2n - i - j}}{1 - a^2}, &\text{otherwise}.
\end{cases} \label{S}
\end{align}  
With the same techniques as above, we obtain 
\begin{align}
\mathbb{P}\left[\hat{a}(\rvec{U}) - a < - \eta \right] \leq \exp\left[-c\min\left( c_1 \eta^2 n, c_2 \eta n\right)\right],
\end{align}
where $c$, $c_1$ and $c_2$ are the same constants as those in~\eqref{eqn:bigger}.
\end{proof}

\subsection{Proof of Theorem~\ref{thm:error_concentration}}
\label{proof:thm_error_concentration}
\begin{proof}
The proof is similar to that of Theorem~\ref{thm:ECG}. In particular,~\eqref{bound:general} still holds: 
\begin{align}
\mathbb{P}\left[ \hat{a}(\rvec{U}) - a > \eta_n \right] \leq \exp\left[-c \min\left(\frac{K_1^2\eta_n^2 n}{K_2}, \frac{K_1 \eta_n n}{K_3}\right)\right]. \label{rhonGeneral}
\end{align}
Instead of~\eqref{K2p} and~\eqref{K3p}, we bound $K_2$ and $K_3$ by 
\begin{align}
K_2 \leq K_2 '' \triangleq \frac{2}{1-a^2}, \label{K2dp}
\end{align}
and 
\begin{align}
K_3 \leq K_3'' \triangleq \frac{2}{1-a}, \label{K3dp}
\end{align}
where~\eqref{K2dp} and~\eqref{K3dp} hold for all $n$ large enough in view of~\eqref{rhon},~\eqref{eqn:K2} and~\eqref{eqn:K3}. Applying the bounds~\eqref{K2dp} and~\eqref{K3dp} to~\eqref{rhonGeneral} and using the fact that for all $n$ large enough, 
\begin{align}
\min\left(\frac{K_1^2\eta_n^2 n}{K_2''}, \frac{K_1 \eta_n n}{K_3''}\right) = \frac{\eta_n^2 n}{8(1-a^2)},
\end{align} 
we obtain 
\begin{align}
\mathbb{P}\left[\hat{a}(\rvec{U}) - a > \eta_n\right] \leq \frac{1}{\left(\log n\right)^{\kappa \alpha}},
\end{align}
where $\kappa$ is given in~\eqref{kappa}. Finally, we can bound $\mathbb{P}\left[\hat{a}(\rvec{U}) - a < -\eta_n\right]$ in the same way.
\end{proof}

\bibliographystyle{IEEEtran}
{\small
\bibliography{GaussMarkov}}

\begin{IEEEbiographynophoto}{Peida Tian}
is currently a PhD student in the Department of Electrical Engineering at California Institute of Technology. He received a B. Engg. in Information Engineering and a B. Sc. in Mathematics from the Chinese University of Hong Kong (2016), and a M.S. in Electrical Engineering from Caltech (2017). He is interested in optimization and information theory.
\end{IEEEbiographynophoto}

\begin{IEEEbiographynophoto}{Victoria Kostina}(S'12--M'14)
joined Caltech as an Assistant Professor of Electrical Engineering in the fall of 2014. She holds a Bachelor's degree from Moscow Institute of Physics and Technology (2004), where she was affiliated with the Institute for Information Transmission Problems of the Russian Academy of Sciences, a Master's degree from University of Ottawa (2006), and a PhD from Princeton University (2013). She received the Natural Sciences and Engineering Research Council of Canada postgraduate scholarship (2009--2012), the Princeton Electrical Engineering Best Dissertation Award (2013), the Simons-Berkeley research fellowship (2015) and the NSF CAREER award (2017).  Kostina's research spans information theory, coding, control and communications. 
\end{IEEEbiographynophoto}
\end{document}